\documentclass{article}
\usepackage[sc]{mathpazo}
\usepackage{amsmath, amssymb, amsfonts, amsthm, paralist, bm}
\usepackage{enumitem}
\usepackage{palatino, soul, courier}
\usepackage{url, hyperref}
\usepackage{graphicx, xcolor}
\usepackage{orcidlink}
\usepackage{mathrsfs}
\usepackage{setspace}
\usepackage{float, booktabs, multirow, tabularx, longtable}
\usepackage[section]{placeins} 
\usepackage{algorithm, algpseudocode}
\usepackage[left= 2.5cm, right = 2.5cm, top = 3cm, bottom = 3cm]{geometry}
\allowdisplaybreaks
\definecolor{alizarin}{rgb}{0.82, 0.1, 0.26}
\hypersetup{
     colorlinks,
     linkcolor = alizarin,
     urlcolor = alizarin,
     citecolor = black,
     pdftitle = {When composite likelihood meets stochastic approximation},
     pdfauthor ={Giuseppe Alfonzetti, Ruggero Bellio, Yunxiao Chen, Irini Moustaki}
}

\definecolor{orcidlogocol}{HTML}{A6CE39}

\theoremstyle{definition}
\newtheorem{theorem}{Theorem}
\newtheorem{lemma}{Lemma}
\newtheorem{definition}{Definition}
\newtheorem{assumption}{Assumption}
\newtheorem{example}{Example}

\newtheorem{corollary}{Corollary}
\newtheorem{proposition}{Proposition}
\newcommand*{\cln}{c\ell_n(\bm{\theta})}

\usepackage[
    backend=biber,
    style=authoryear,
    sortlocale=de_DE,
    natbib=true,
    url=false, 
    doi=true,
    eprint=false
]{biblatex}
\title{When Composite Likelihood Meets \\ Stochastic Approximation}
\author{\orcidlink{https://orcid.org/0000-0002-6118-6553} Giuseppe Alfonzetti*, Ruggero Bellio \\ {\small Department of Economics and Statistics,}\\ {\small University of Udine, via Tomadini,}\\ {\small 33100, Udine, Italy} \\ {\small *\texttt{giuseppe.alfonzetti@uniud.it}} \and Yunxiao Chen, Irini Moustaki \\ {\small Department of Statistics,}\\ {\small London School of Economics,} \\ {\small
    Houghton Street, WC2A 2AE, London, UK} }
\date{}

\begin{document}
\maketitle
\begin{abstract}
A composite likelihood is an inference function derived by multiplying a set of likelihood components. This approach provides a flexible framework for drawing inferences when the likelihood function of a statistical model is computationally intractable. While composite likelihood has computational advantages, it can still be demanding when dealing with numerous likelihood components and a large sample size. This paper tackles this challenge by employing an approximation of the conventional composite likelihood estimator based on a stochastic optimization procedure. This novel estimator is shown to be asymptotically normally distributed around the true parameter. In particular,
based on the relative divergent rate of the sample size and the number of iterations of the optimization, the variance of the limiting distribution is shown to compound for two sources of uncertainty: the sampling variability of the data and the optimization noise, with the latter depending on the sampling distribution used to construct the stochastic gradients.
The advantages of the proposed framework are illustrated through simulation studies on two working examples: an Ising model for binary data and a gamma frailty model for count data. Finally, a real-data application is presented, showing its effectiveness in a large-scale mental health survey. 
\bigskip

\end{abstract}

\section{Introduction}

The seminal work of \citet{Besag1974} and the general framework proposed by \citet{lindsay_statistical_1988} have paved the way for the wide adoption of composite likelihood methods as a practical approach for modelling multivariate responses with complex dependence structures (e.g., \citealp{henderson2003, bellio_pairwise_2005, katsikatsou_pairwise_2012, lee_learning_2015}). Such methods replace an intractable likelihood with an inference function constructed by multiplying many lower-dimensional marginal or conditional likelihood components, enabling frequentist estimation when traditional maximum likelihood approaches are infeasible or unattainable; see \citet{varin_overview_2011} for an overview.
However, in settings with large sample sizes and moderate response dimensions, the numerical optimization of the composite likelihood function requires evaluating many likelihood components at each iteration. Thus, it becomes, in turn, computationally unattainable. 

Natural candidates for such settings are stochastic approximations, computationally convenient alternatives to numerical optimization that replace the score used by gradient-based routines with an adequately defined stochastic substitute \citep{robbins_stochastic_1951}. Thanks to their computational convenience, methods based on stochastic gradients (SGs) have quickly gained popularity among practitioners, becoming the standard choice for estimating complex models on large-scale data \citep{bottou_optimization_2018}. While their success is due to the capability of providing computationally affordable point estimates of the parameters of interest, the last decade has seen rising attention to conducting statistical inference with such estimates. Most of the recent developments in this regard build on the seminal work of \citet{ruppert1988efficient} and \citet{polyak1992acceleration}, who first established the asymptotic normality and statistical optimality of averaged stochastic estimators. We identify two challenges that have attracted the interest of researchers in recent years. The first one is the theoretical extension of the asymptotic results of \citet{polyak1992acceleration} to more general and flexible settings than the ones outlined in the original paper. In this regard, \citet{toulis_asymptotic_2017} establish the asymptotic optimality of the Polyak-Ruppert averaged version of their proposed estimator based on implicit stochastic updates; \citet{lee_etal_2022} extend \citet{polyak1992acceleration} results to a functional form; \citet{su2023} relax the original assumptions to allow for globally convex and locally strongly convex objective functions; \citet{wei_etal_2023} include the effect of general averaging schemes in the asymptotic distribution of the averaged estimates; \citet{chen_etal_2024} establish the asymptotic properties of the averaged version of the Kiefer-Wolfowitz algorithm. The second challenge is the online estimation of the uncertainty of stochastic estimates. With few exceptions (e.g. \citealp{chee_etal_2023}),  most recent works in this area focus on Polyak-Ruppert type estimators, including the online bootstrap \citep{fang_etal_2018}, mean-batch estimators \citep{chen_etal_2020,zhu_etal_2023}, the random scaling approach \citep{lee_etal_2022, chen_etal_2023, chen_etal_2024}, and HiGrad \citep{su2023}.

While most of the papers mentioned above explicitly refer to settings with online data (i.e., new observations sampled from the true data-generating distribution), stochastic optimization can also be used offline on an observed dataset, which is the classical setting of frequentist estimation. In the online-data setting, the result in \citet{polyak1992acceleration} and its extensions directly guarantee the asymptotic statistical optimality of the averaged stochastic estimator for the true parameter value. However, in the offline-data scenario, at each iteration, new observations are resampled from the empirical distribution of data, and stochastic estimators converge to the maximum likelihood estimator (MLE) (e.g., \citealp{moulines2011, needell2014}). It follows that, for a given dataset, the inferential procedures based on \citet{polyak1992acceleration} only quantify the variability of stochastic estimators around their target, i.e. the MLE, but neglect the sampling variability of the data. In addition, to our knowledge, the combination of stochastic approximations and composite likelihood inference has not been formally investigated.

Thus, our contribution in the following is two-fold. First, in Section~\ref{sec:clsa}, we show how different sampling schemes for the margins involved in the composite likelihood affect the statistical efficiency of SGs. Second, in Section~\ref{sec:theory}, we extend \citet{polyak1992acceleration} result by establishing the consistency and asymptotic normality of the stochastic estimator around the true parameter in the offline-data setting. In particular,  we show that, according to the relative divergence rate of the sample size and the number of iterations, the variance of the limiting distribution compounds for two sources of uncertainty: the sampling variability of the data and the noise injected by the SGs. 
While intuitive, combining the two sources of variability is technically non-trivial. Allowing the data to be random implies that all the SGs share a common source of variability, which complicates the applicability of central limit theorems for martingale sequences. 
Nevertheless, the asymptotic distribution can still be identified with the sum of exchangeable summands, which allows us to use the central limit theorem outlined in \citet{blum_chernoff_rosenblatt_teicher_1958} to establish the asymptotic normality of the estimator. Furthermore, by taking advantage of the second Bartlett's identity for the single likelihood components, it is possible to lower the noise injected in the optimization with the SGs at a fixed computational cost. In Section~\ref{sec:sims}, we  investigate the established theoretical results with simulation experiments on two working examples, one of which concerns estimating the Ising model based on full-conditional margins and the other concerns estimating a gamma-frailty model based on bivariate margins. Finally, Section~\ref{sec:real} provides a real data application to a United States national health survey\footnote{\href{https://catalog.data.gov/dataset/national-epidemiologic-survey-on-alcohol-and-related-conditions-nesarcwave-1-20012002-and-}{https://catalog.data.gov/dataset/national-epidemiologic-survey-on-alcohol-and-related-conditions-nesarcwave-1-20012002-and-}}.

\section{Composite Likelihood, Stochastic Approximations}
\label{sec:clsa}
\subsection{Composite Likelihood}

Let $\bm{Y} = (Y_1, \dots, Y_p)^\top$ be a $p$-variate random vector that follows a parametric distribution with probability density/mass function $p(\bm{y};\bm{\theta})$ and parameters 
$\bm{\theta} \in \mathbb{R}^d$. 
Let $\{\mathcal{A}_1, \dots, \mathcal{A}_K\}$ be a set of marginal or conditional events with likelihood functions $\mathcal{L}_k(\bm{\theta};\bm{y})\propto p(\bm{y}\in \mathcal{A}_k; \bm{\theta})$.
A composite log-likelihood is obtained by summing the logarithms of the $K$ likelihood objects. Thus, with
 $\bm{y_1}, \dots, \bm{y_n}$ being independent and identically distributed (i.i.d.) realizations of  $\bm{Y}$, inference on $\bm{\theta}$ can be drawn based on the composite log-likelihood function 
\begin{equation}
   \label{eq:cln}
c\ell_n(\bm{\theta})=\sum_{i=1}^n\sum_{k=1}^Kw_k\ell_k(\bm{\theta}; \bm{y}_i),  
\end{equation}
where $\ell_k(\bm{\theta}; \bm{y}_i) = \log \mathcal{L}_k(\bm{\theta};\bm{y}_i)$ is usually referred to as the $k$-th log-likelihood component or sub-log-likelihood and $\{w_1,\dots,w_K\}$ is a set of weights to be defined depending on the model being estimated.
The Composite Likelihood Estimator (CLE) is given by  
\begin{equation}\label{eq:cle}
\hat{\bm{\theta}} = \text{argmin}_{\bm{\theta}} ~ \{- c\ell_n(\bm{\theta})\}.
\end{equation}
Typically, the CLE\ is used when the sub-log-likelihoods  have  simple forms while the joint density function $p(\bm{y};\bm{\theta})$ is analytically intractable. We give two illustrative examples below. 

\begin{example}
\label{ex:graph}
\label{sec:ccl}
Let $\bm{Y}$ be a binary vector with $Y_j \in \{0, 1\}$ following an Ising model (\citealp{ising1924}). Under this model,  

	$$p(\bm{y}; \bm{\theta}) =  \exp\left\{\sum_{j = 1}^p \beta_{j0}y_j + \sum_{j< j'}\beta_{jj'}y_jy_{j'} \right\}/{Z(\bm{\theta})},$$

where 
   $ Z(\bm{\theta}) = \sum_{y\in\{0,1\}^p}\exp\left\{\sum_{j = 1}^p \beta_{j0}y_j + \sum_{j< j'}\beta_{jj'}y_jy_{j'} \right\}$
is the so-called partition function which is needed to guarantee 
$p(\bm{y};\bm{\theta})$ to be a proper probability mass function. In this model, the parameters of 
interest are $\bm{\theta} = (\beta_{i0}, \beta_{jj'}: i =1, \dots, p,  1 \leq j < j' \leq p)^\top$, whose dimension is $d =p+ p(p-1)/2$.
As $Z(\bm{\theta})$ involves a summation over all possible binary vectors, the complexity of computing $Z(\bm{\theta})$ grows exponentially with $p$. Thus, the likelihood function quickly becomes intractable when $p$ is large. 
To draw inference under the Ising model, \citet{Besag1974} proposed a composite likelihood estimator that considers $K = p$ component likelihood terms
$$\mathcal{L}_j(\bm{\theta};\bm{y}) = {\exp(y_{j}(\beta_{j0}+\sum_{j'}\beta_{jj'}y_{j'}))}/{(1+\exp(\beta_{j0}+\sum_{j'}\beta_{jj'}y_{j'}))}, j =1, \dots, p,$$ which are derived from the conditional distribution of $Y_j$ given the rest of the entries of $Y$. The composite log-likelihood function 
for a random sample of size $n$ can then be written as 
$c\ell_n(\bm{\theta})=\sum_{i=1}^n\sum_{j=1}^p   y_{ij}\big(\beta_{j0}+\sum_{j'\neq j}\beta_{jj'}y_{ij'}) - \log\big\{1+\exp\big(\beta_{j0}+\sum_{j'\neq j}\beta_{jj'}y_{ij'} \big)\big\}.$
\end{example}

\begin{example}
\label{ex:pl}
Consider the correlated gamma frailty model proposed in \citet{henderson2003}. The authors impose an autoregressive correlation structure to model the underlying gamma process. In such a setting, it is convenient to consider only pairs within a certain time lag, drastically lowering the model's estimation cost. For this reason, we consider an exchangeable correlation structure such that no pairs can be ignored a priori.  For illustration purposes, we consider a simplified version without covariates.
Let $\bm{Y}$ be a multivariate count vector of dimension $p$. Its generic element, $Y_j \in\mathbb{N}$ for $j = 1,\dots,p$, is distributed as
 $   Y_{j}|V_{j}\sim \text{Poisson}\left\{V_{j}\exp(\lambda_{j})\right\}, \text{ for } j = 1, \dots p,$
where $\lambda_{j}$ the $j$-th baseline rate.
The $p$-dimensional frailty term $\bm{V}$ has unidimensional margins distributed as $V_j\sim\text{Gamma}(\xi^{-1},\xi^{-1})$, and correlation matrix $\bm{C}$, with generic element $C_{jj'}=\rho$, for $0\leq\rho\leq 1$. The interest is in estimating $\bm{\theta} = (\bm{\lambda}, \rho, \xi)$, with $\bm{\lambda}=(\lambda_{1},\dots,\lambda_{p})^\top$, which has dimension $d = p+2$. 
The density function of the model can be written as
\begin{align}
    p(Y_1=y_1,\dots, Y_p=y_p;\bm{\theta}) &= \left(\prod_{j=1}^p\frac{u_j^{y_j}}{y_j!}\right)\int_{\mathbb{R}^p} \left(\prod_{j=1}^pV_{j}^{y_j}\right)\exp(-\bm{u}^\top \bm{V} ) d\bm{V} \notag\\
    \label{eq:pl:mvd}
    &= (-1)^{\sum_j y_j} \left(\prod_{j=1}^p\frac{u_j^{n_j}}{n_j!}\right)\frac{\partial^{(\sum_jy_j)}L(\bm{u})}{\partial^{y_1}u_1,\dots,\partial^{y_p}u_p},
\end{align}
where $u_j= \exp(\lambda_{j0})$, $\bm{u}=(u_1, \dots, u_p)^\top$, $\bm{V}=(V_1,\dots,V_p)^\top$ and $L(\bm{u}) = E_V\left\{\exp(-\bm{u}^\top \bm{V})\right\}$. The computational challenge is that the random number of derivatives involved in \eqref{eq:pl:mvd}, $\sum_{j}y_j$, can be too large to handle even in small $p$ settings.
\citet{henderson2003} substitute \eqref{eq:pl:mvd} with the composition of bivariate log-margins $\ell_{jj'}(\bm{\theta};y_j, y_{j'})$, computed via
\begin{align}
\label{eq:pl:pair}
    \ell_{jj'}(\bm{\theta};y_j, y_{j'}) &= y_j\log u_j + y_{j'}\log u_{j'} -\log (y_{j}!) - \log(y_{j'}!) + \notag\\
    &+\sum_{s = 0}^{m_2-1}\log(1+s\xi) + y_j\log D_j + y_{j'}\log D_{j'} - (y_j+y_{j'} + \xi^{-1})\log\Delta + \notag\\
    &+ \log\sum_{s=0}^{m_1} \left[(-1)^{s}\binom{m_1}{s}\binom{m_2}{s}s!\left\{\prod_{s'=m_2}^{m_1+m_2-s-1}(1+s'\xi)\right\}\xi^sf^s\right],
\end{align}
with $m_1 = \min(y_j, y_{j'})$, $m_2 = \max(y_j, y_{j'})$, $\Delta = 1+\xi u_j +\xi u_{j'} + \xi^{2} u_j u_{j'}(1-\rho^{|j-j'|})$, $D_j = 1 + \xi u_{j'}(1-\rho^{|j-j'|})$ and $f = \frac{\Delta(1-\rho)}{D_jD_{j'}}$. Therefore, the composite log-likelihood for a random sample of size $n$ can then be written as
$
	c\ell_n(\bm{\theta}) = \sum_{i=1}^n\sum_{j<j'}\ell_{jj'}(\bm{\theta};y_{ij}, y_{ij'}).
$
\end{example}
\vspace{.5cm}
\noindent
When the parametric model is correctly specified, 
the CLE is consistent and asymptotically normal \citep{lindsay_statistical_1988}. Let ${\bm{\theta}}^*$ denote the true parameter value. Then $\hat {\bm{\theta}}$ converges in probability to ${\bm{\theta}}^*$, and further 
\begin{equation}
\label{eq:cl:asyd}
    \sqrt{n}(\hat {\bm{\theta}} -{\bm{\theta}}^*)\xrightarrow{\text{d}} \mathcal{N}_p\left(0, \bm{H}^{-1}\bm{J}\bm{H}^{-1}\right),
\end{equation}
where 
$\bm{H} = E_{\bm Y}\{-\nabla^2c\ell_n({\bm{\theta}}^*)/n\}$ and $\bm{J} = \text{Var}_{\bm Y}\{\nabla c\ell_n({\bm{\theta}}^*)/n\}$. Note that such asymptotic results are obtained under a classical asymptotic regime where
$p$ is fixed while $n$ diverges. In the case of \eqref{eq:cln}, such an assumption implies considering both the number of contributions, $K$, and the parameter space, $d$, as fixed. 

\subsection{Proposed Method}
As discussed in Section~\ref{subsec:comp}, it can be computationally intensive to obtain a numerical solution to \eqref{eq:cle} that can be used for statistical inference when large sample sizes and numerous sub-likelihood components are involved. To leverage the trade-off between statistical and computational efficiencies, we propose to use an algorithm based on SGs to get an approximation of $\hat{\bm{\theta}}$. First, consider that the gradient of $-c\ell_n(\bm{\theta})$ with respect to $\bm{\theta}$ is simply given by  $-\sum_{i=1}^n\sum_{k=1}^K \nabla\ell_k(\bm{\theta}; \bm{y}_i).$ Considering its double-sum structure and the proper-likelihood nature of each $\ell_k(\bm{\theta};\bm{y}_i)$, a SG of $-c\ell_n(\bm{\theta})$ can be constructed as 
\begin{equation}
\label{eq:csg}
    S(\bm{\theta}; \bm{W}) = -\frac{1}{\gamma_1}\sum_{i=1}^n\sum_{k=1}^K w_{ik}\nabla\ell_k(\bm{\theta}; \bm{y}_i).
\end{equation}
Here, $\bm{W}=(w_{ik})_{n\times K}$ is a random matrix following a joint distribution $\mathcal{P}$, under which $w_{ik} \in \{0, 1\}$ and $P(w_{ik} = 1) = \gamma_1$. Differently from (\ref{eq:cln}), we let the weights change across observations and iterations. Such broader specification introduces greater flexibility in managing the variability of the SG by allowing for different choices of $\mathcal{P}$.

As summarized in Algorithm~\ref{algo:csgd} below, the proposed method iterates between updating $\bm{\theta}$ and constructing an SG under the current value of $\bm{\theta}$. We anchor the number of iterations to the sample size for theoretical reasons that will be shown in detail in Section~\ref{sec:theory}. For the moment, just let the notation $T_n$ refer to the number of iterations performed, where the dependence on $n$ is such that $T_n\xrightarrow{}\infty$ as $n\xrightarrow{}\infty$.

\begin{algorithm}[h]
\caption{Composite Stochastic Gradient Descent}
\begin{algorithmic}[1]
\State \text{\bf Input: }{$\bm{y}$,$\mathcal{P}$, $\bm{\theta}_0$, $\eta_0$, $T_n$, $B$.}
\For{$t$ \text{\bf in} $1,\dots, T_n$}
    \State \text{\bf Sampling Step:} Draw $\bm{W}_t$ from $\mathcal{P}$;
    \State \text{\bf Approximation Step:} Construct a stochastic gradient $\bm{S}_t = S(\bm{\theta}_{t-1}; \bm{W}_t)$;
    \State \text{\bf Update Step:} Update the parameter estimate via $\bm{\theta}_{t} = \bm{\theta}_{t-1}-\frac{\eta_{t}}{n} \bm{S}_t$;
\EndFor
\State \text{\bf Trajectories Averaging:} Compute $\bar{\bm{\theta}}_\mathcal{P}=\frac{1}{T_n-B}\sum_{t=B+1}^{T_n}\bm{\theta}_{t}$;
\State \text{\bf Output:}{ Return $\bar{\bm{\theta}}_\mathcal{P}$.}
\end{algorithmic}
\label{algo:csgd}
\end{algorithm}

In each iteration, the Sampling Step injects some randomness into the procedure by drawing $\bm{W}_t$ from $\mathcal{P}$. Note that $\mathcal{P}$ must be defined such that the SG computed during the Approximation Step, $\bm{S}_t$, is computationally cheaper than the full gradient $-\nabla c\ell_n(\bm{\theta}_{t-1})$, but still unbiased, i.e. $E_{\bm W}\{\bm{S}_t\} = -\nabla c\ell_n(\bm{\theta}_{t-1})$.
In other words, although the descent direction identified by $\bm{S}_t$ is noisy due to the randomness of $\bm{W}$, it still recovers the exact negative gradient on average.
The computational convenience of $\bm{S}_t$ is the key advantage of Algorithm~\ref{algo:csgd} and, together with $\mathcal{P}$, it controls the trade-off between computational and statistical efficiencies.
At the end of each iteration, the estimates are updated via a Robbins-Monro step \citep{robbins_stochastic_1951}, where $\eta_{t}$ is the stepsize at the $t$-th iteration computed given an initial value $\eta_0$ and a suitable decreasing schedule, formalized in Assumption~\ref{ass:stepsize} in Section~\ref{sec:theory}. The rescaling by $1/n$ only affects the initial value $\eta_0$ and not the scheduling per se. It serves to standardize the SG, such that $n^{-1}\bm{S}_t$ is of magnitude $O(1)$, which is used in the proof of Proposition~\ref{prop:conv}. It is worth remarking that the Update Step can be generalized to a second-order update, where the stepsize is substituted by a $d\times d$ matrix approximating $\bm{H}^{-1}$. For ease of exposition, we limit the discussion to one-dimensional stepsizes. Finally, after completing $T_n$ iterations, the output of the algorithm, $\bar{\bm{\theta}}_\mathcal{P}$, is computed by averaging the stochastic estimates along their trajectories \citep{ruppert1988efficient,polyak1992acceleration}. 

In practice, it is often useful to account for a burn-in period, $B$, to avoid including estimates too close to the starting point $\bm{\theta}_0$ in the computation of $\bar{\bm{\theta}}_\mathcal{P}$.

It is straightforward to notice that Algorithm~\ref{algo:csgd} closely follows the averaging SG descent outlined in \citet{polyak1992acceleration}, and it inherits, in fact, its theoretical properties in approximating $\hat{\bm{\theta}}$. However, the theoretical framework discussed in Section~\ref{sec:theory} allows the output of Algorithm~\ref{algo:csgd} to be directly used to draw inference on ${\bm{\theta}}^*$.
 
A further advantage of our proposal is that, by specifying $\bm{S}_t$ as in \eqref{eq:csg}, Algorithm~\ref{algo:csgd} gives the user the freedom to leverage the peculiar structure of the composite likelihood to improve the efficiency of the approximation by adequately choosing $\mathcal{P}$. To this end, let us introduce three possible choices for the distribution of the weights, as described in Definitions \ref{def:simple:w} through \ref{def:hyper:w}.
\begin{definition} 
\label{def:simple:w}
Let $\mathcal{P}_1$ be the joint distribution of $\bm{W}$ such that  $\bm{W} = \bm{D}\bm{1}_{nK}$, where $\bm{1}_{nK}$ is a $n\times K$ matrix  of ones and $\bm{D}$ is a $n\times n$ diagonal matrix, with diagonal distributed as $\text{Multinomial}\left\{1,(1/n,\dots,1/n)\right\}$.
\end{definition}
\begin{definition} 
\label{def:be:w}
Let $\mathcal{P}_2$ be the joint distribution of $\bm{W}$ such that $W_{i,k}\stackrel{i.i.d.}{\sim}\text{Bernoulli}\left(1/n\right)$ for $i=1,\dots,n$ and $k=1,\dots,K$.
\end{definition}
\begin{definition} 
\label{def:hyper:w}
Let $\mathcal{P}_3$ be the joint distribution of $\bm{W}$ such that $W_{i,k} = v_{iK-K+k}$, where $\bm{v}$ is a $nK$-dimensional random vector following a multivariate hypergeometric distribution with $K$ draws over $nK$ categories of dimension $1$ and $v_j$ is its $j$-th element. 
\end{definition}

All three definitions lead to drawing, on average, $K$ weights equal to $1$ and the remaining $nK-K$ equal to zero. The implied per-iteration complexity is $O(K)$, independent of $n$, and thus particularly suitable for large-scale applications. 
However, according to the sampling scheme chosen, a different dependence layout is induced on the cells of $\bm{W}$. In Section~\ref{sec:theory}, we show how this noise structure affects the asymptotic efficiency of the stochastic estimates.

We discuss the three proposed sampling schemes in more detail. A weight matrix $\bm{W}$ sampled according to Definition~\ref{def:simple:w} is constrained to have all elements equal to zero, apart from a single row filled with ones. With such weights, $\bm{S}_t$ evaluates the gradient on a single observation selected randomly from $\{1,\dots,n\}$. We refer to this construction as \textit{standard} SG to stress its widespread adoption at the core of many stochastic algorithms. Note that using such a sampling scheme, we are ignoring the double sum structure of (\ref{eq:cln}) since $\mathcal{P}_1$ constraints the $K$ selected sub-likelihood component to belong to the same observation.

Nevertheless, \eqref{eq:csg} is very flexible in defining the SG and allows for different choices of $\mathcal{P}$. Consider $\mathcal{P}_2$ as described by Definition~\ref{def:be:w}. All the elements of the matrix are now independent and identically distributed as Bernoulli random variables with proportion parameter $1/n$.
It means that, at each iteration, $K$ sub-likelihood components are selected on average by the weighting matrix. Therefore, the complexity of the approximation matches the standard SG one. However, the proportion parameter in Definition~\ref{def:be:w} can be set as low as $(nK)^{-1}$, implying an iteration complexity $O(1)$, which is unattainable using Definition~\ref{def:simple:w}. Regardless, for comparison purposes, we stick to the proportion parameter $1/n$.
Note that the structure of the noise injected by $\mathcal{P}_2$ is very different from the one implied by $\mathcal{P}_1$. While the $K$ components drawn by $\mathcal{P}_1$ share a very specific covariance stemming from the dependence among the summands of $c\ell_n(\bm{\theta})$, $\mathcal{P}_2$ completely breaks this structure by independently selecting sub-likelihoods possibly belonging to independent observations.
Finally, consider the sampling scheme $\mathcal{P}_3$. It can be seen as a random scramble of the vectorization of $\bm{W}$, where only the elements in the first $K$ positions are retained.  Like $\mathcal{P}_2$, the complexity per iteration can be lowered to $O(1)$ by decreasing the number of components retained per iteration.
However, in this case, the weights are not completely independent since, given the fixed number of components drawn, a weak negative correlation is induced among the elements of $\bm{W}$.
In Section~\ref{sec:theory}, we show how $\mathcal{P}_2$ and $\mathcal{P}_3$ improve the efficiency of $\bar{\bm{\theta}}_\mathcal{P}$ while maintaining the same computational cost as the standard SG. 

\subsection{Comparison with Gradient Descent}\label{subsec:comp}
Before investigating the theoretical properties of Algorithm~\ref{algo:csgd} in the next section, let us consider solving (\ref{eq:cle}) by a gradient descent algorithm with fixed stepsize to compute $\tilde{\bm{\theta}}$, a numerical estimate of ${\bm{\theta}}^*$. Given the subtleness of the notation, see Table~\ref{tab:thetas} as a reference for the symbols used.
\begin{table}[h]
\footnotesize
\centering
\caption{\label{tab:thetas} Notation summary for $\bm{\theta}$ values.}
\begin{tabular}{ p{1in}p{1in}p{1in}p{1in}} 
\toprule
${\bm{\theta}}^*$ & $\hat{\bm{\theta}}$  & $\tilde{\bm{\theta}}$ & $\bar{\bm{\theta}}_{\mathcal{P}}$\\
\midrule
 True parameter. Not observed. &
 CLE. Usually not available analytically.
 & Numerical approximation of $\hat{\bm{\theta}}$ used as estimate of ${\bm{\theta}}^*$.
 & Stochastic estimator of ${\bm{\theta}}^*$ based on the chosen $\mathcal{P}$.\\
 \bottomrule
\end{tabular}
 
\end{table}
To describe the relative divergence rate of two positive sequences $a_n$ and $b_n$, we write $a_n = o(b_n)$ if $\lim_{n\rightarrow\infty} a_n/b_n = 0$, $a_n = \omega(b_n)$ if $\lim_{n\rightarrow\infty} a_n/b_n = \infty$, while $a_n = \Theta(b_n)$ if $\lim_{n\rightarrow\infty} a_n/b_n = \gamma$ with $0<\gamma<\infty$.

Given a fixed starting point $\bm{\theta}_{0}$ and assuming $c\ell_n(\bm{\theta})$ to be strongly convex and its gradient to be Lipschitz continuous with constant $L>0$, then, at each iteration $t$, the numerical procedure updates via

	$\bm{\theta}_{t} = \bm{\theta}_{t-1} + \eta\nabla c\ell_n(\bm{\theta}_{t-1}),\quad t=1\,\dots,T_n,$
where $0<\eta<2/L$ is a fixed stepsize. The final parameter estimate is taken as the output of the algorithm, namely $\tilde{\bm{\theta}} = \bm{\theta}_{T_n}$. This gradient-based update highlights that the critical quantity to be computed at each iteration is $\nabla c\ell_n(\bm{\theta})$, which costs $O(nK)$ operations. Furthermore, to draw statistical inference with the numerical solution $\tilde{\bm{\theta}}$, one requires $\tilde{\bm{\theta}}$ to have the same limiting distribution as $\hat{\bm{\theta}}$, which implies that 
$\tilde{\bm{\theta}} - \hat{\bm{\theta}} = o(1/\sqrt{n})$.  
Given the linear convergence rate of gradient descent (see, for example, Theorem 2 in Section 1.4.2 of \citet{polyakbook}), the total number of iterations needs to satisfy $T_n =\omega( -\frac{1}{2}\log_c n)$ with $c\in[0,1)$, hence $T_n = \omega(\log n)$. In Table~\ref{tab:efficiency}, we compare the total computational budget, $\mathcal{B}$, and asymptotic variance of $\bar{\bm{\theta}}_\mathcal{P}$ and $\tilde{\bm{\theta}}$ according to the relative divergence rate of $T_n$ and $n$.

\begin{table}[ht]
\footnotesize
\centering
\caption{\label{tab:efficiency} Computational and statistical efficiency comparison between gradient descent and Algorithm~\ref{algo:csgd} with $\mathcal{P}$ chosen according to Definitions \ref{def:simple:w}, \ref{def:be:w} and \ref{def:hyper:w}.}
\begin{tabular}{  l c c c c } 
\toprule
&Per iteration & Number of  & Total                & Asymptotic\\
&complexity    & iterations & complexity & variance\\

\cmidrule[0.4pt](r{0.125em}){1-1}%
\cmidrule[0.4pt](lr{0.125em}){2-2}%
\cmidrule[0.4pt](lr{0.125em}){3-3}%
\cmidrule[0.4pt](lr{0.125em}){4-4}%
\cmidrule[0.4pt](lr{0.125em}){5-5}%

GD & $O(nK)$ & $T_n= \omega (\log n) $ & $\mathcal{B}= \omega(n K\log n)$ &$n^{-1}{\bm{H}}^{-1}{\bm{J}}{\bm{H}}^{-1}$\\
 \midrule
  \multirow{3}{*}{$\mathcal{P}_1$} & \multirow{3}{*}{$O(K)$} & $T_n = \omega(n) $  & $\mathcal{B}=\omega(nK)$ & $n^{-1}{\bm{H}}^{-1}{\bm{J}}{\bm{H}}^{-1}$ \\
   & & $T_n = o(n) $ & $\mathcal{B} = o(nK)$ & $T_n^{-1}{\bm{H}}^{-1}{\bm{J}}{\bm{H}}^{-1}$\\
   & & $T_n = \Theta(n)$ & $\mathcal{B}=\Theta(nK)$ & $(T_n^{-1}+n^{-1}){\bm{H}}^{-1}{\bm{J}}{\bm{H}}^{-1}$\\

 \midrule
   \multirow{3}{*}{$\mathcal{P}_2$} & \multirow{3}{*}{$O(K)$} & $T_n = \omega(n) $  & $\mathcal{B}=\omega(nK)$ & $n^{-1}{\bm{H}}^{-1}{\bm{J}}{\bm{H}}^{-1}$ \\
   & & $T_n = o(n) $ & $\mathcal{B} = o(nK)$ & $T_n^{-1}{\bm{H}}^{-1}$\\
   & & $T_n = \Theta(n)$ & $\mathcal{B}=\Theta(nK)$ & $n^{-1}{\bm{H}}^{-1}{\bm{J}}{\bm{H}}^{-1}+T_n^{-1}{\bm{H}}^{-1}$\\
\midrule
   \multirow{3}{*}{$\mathcal{P}_2$} & \multirow{3}{*}{$O(K)$} & $T_n = \omega(n) $  & $\mathcal{B}=\omega(nK)$ & $n^{-1}{\bm{H}}^{-1}{\bm{J}}{\bm{H}}^{-1}$ \\
   & & $T_n = o(n) $ & $\mathcal{B} = o(nK)$ & $T_n^{-1}{\bm{H}}^{-1}$\\
   & & $T_n = \Theta(n)$ & $\mathcal{B}=\Theta(nK)$ & $n^{-1}{\bm{H}}^{-1}{\bm{J}}{\bm{H}}^{-1}+T_n^{-1}{\bm{H}}^{-1}$\\

 \bottomrule
\end{tabular}
 
\end{table}

First, to achieve the asymptotic efficiency in \eqref{eq:cl:asyd}, Algorithm~\ref{algo:csgd} needs many more iterations compared to gradient descent.
However, given the extreme computational affordability of its iterations, the total budget needed to reach such an asymptotic behavior is lower than what is needed by gradient descent, whatever the choice of $\mathcal{P}$. While theoretically appealing, Algorithm~\ref{algo:csgd} might need to tune $\eta_0$ adequately in practice, like most stochastic optimization methods, which increases its computational cost. Furthermore, numerical optimization is typically carried out with Newton and quasi-Newton algorithms. They are well-known to be more efficient than gradient descent in practical applications. Still, they are particularly sensitive to the dimension of the parameter space since they require to compute, or approximate, the inverse of the $d\times d$ Hessian matrix. From a practical point of view, on problems of moderate dimensions, it might still be preferable to use numerical optimizers to take advantage of the asymptotic behavior of $\hat{\bm{\theta}}$.

The real advantage of Algorithm~\ref{algo:csgd} shows up when $O(nK)$ is the maximum computational budget available, and we want to quantify uncertainty around our estimates. In such a scenario, using the numerical solution $\tilde{\bm{\theta}}$ can be infeasible since the computational inaccuracy remaining might be too large. Regardless, using $\bar{\bm{\theta}}_\mathcal{P}$ is a viable option, and the reason is rather intuitive. Although subtle, when running numerical optimization, one does not use $\tilde{\bm{\theta}}$ directly to draw inference on ${\bm{\theta}}^*$. Instead, the requirement for $\tilde{\bm{\theta}}$ is to be close enough to $\hat{\bm{\theta}}$ in order to safely replace $\hat{\bm{\theta}}$ with $\tilde {\bm{\theta}}$ in (\ref{eq:cl:asyd}). We argue that this is not the case when using stochastic optimizers since $\bar{\bm{\theta}}_\mathcal{P}-\hat{\bm{\theta}}$ is a random variable itself, with distribution depending on $\mathcal{P}$. Hence, quantifying the noise injected by $\bm{W}$ in the optimization makes it possible to directly use $\bar{\bm{\theta}}_\mathcal{P}$ for inference on ${\bm{\theta}}^*$, without strict requirements on its distance from $\hat{\bm{\theta}}$. Fixing $\mathcal{B}=O(nK)$  implies the running length of the algorithm to be either $T_n = o(n)$ or $T_n = \Theta(n)$. With such divergence rates, the choice of $\mathcal{P}$ plays a central role in the asymptotic variance of $\bar{\bm{\theta}}_\mathcal{P}$ since the noise in the stochastic approximation is still non-negligible. Table~\ref{tab:efficiency} shows that, in those cases, relying on different choices of $\mathcal{P}$ is not equivalent. Estimates based on $\mathcal{P}_2$ and $\mathcal{P}_3$ enjoy a lower asymptotic variance than $\mathcal{P}_1$, as will become more apparent in the next section.
Furthermore, Section~\ref{sec:theory} also establishes the asymptotic distribution of $\bar{\bm{\theta}}_\mathcal{P}$ around ${\bm{\theta}}^*$ under some mild assumptions on the choice of $\mathcal{P}$.
\subsection{Implementation Remarks}
\label{sec:imprem}
From an implementation perspective, Algorithm~\ref{algo:csgd} allows for some practical expedients to enhance the computational performance. First, the computation of $\bm{S}_t$ at each iteration can be easily parallelized by taking advantage of the double-sum structure of $\nabla c\ell_n(\bm{\theta})$, assigning the gradient computation for a different sub-likelihood component to each available thread. 

Second, as typical of stochastic algorithms, not all the data are needed at each iteration, such that memory resources can be saved by carefully passing only the portion of the data needed to compute $\bm{S}_t$ at a given $t$. In this regard, $\mathcal{P}_1$ is the cheaper choice memory-wise since it fixes the memory cost to $O(K)$ no matter the structure of $\nabla c\ell_n(\bm{\theta})$. By choosing $\mathcal{P}_2$ or $\mathcal{P}_3$, one typically needs to store data coming from multiple observations, such that the maximum amount of memory necessary strictly depends on the model. In the case of Example~\ref{ex:graph}, each sub-log-likelihood component has a memory cost $O(K)$. Thus, since $\mathcal{P}_2$ and $\mathcal{P}_3$ can draw components potentially belonging to $K$ different observations, their maximum memory cost per iteration is $O(K^2)$. If we consider Example \ref{ex:pl}, each component accesses only $O(1)$ data instead. Consequently, by drawing $K$ components on different observations, both $\mathcal{P}_2$ and $\mathcal{P}_3$ have a maximum per-iteration memory cost of $O(K)$.

Third, the Sampling Step can be recycled across iterations to save computational resources. Namely, with $\mathcal{P}_1$, one can scramble the vector $(1,\dots, n)$ once and then use each of the first $l$ elements as indices of the observations drawn in the following $l$-dimensional window of iterations. Intuitively, as long as $l$ is low enough, the dependence induced by the recycling among iterations within the same window is negligible, such that they can still be considered independent. Thus, recycling trims the cost of the Sampling Step by a factor of $l$. The same approach can be implemented when using $\mathcal{P}_3$ by scrambling the vector $(1,\dots, nK)$ and allocating the first $l$ $K$-dimensional sequences of indices to the subsequent $l$ iterations. Unfortunately, recycling is impossible when $\mathcal{P}_2$ is chosen since the number of components drawn per iteration is not deterministic.
\section{Theoretical Results}
\label{sec:theory}
In what follows, we establish the asymptotic properties of the proposed estimator $\bar{\bm{\theta}}_\mathcal{P}$. Proposition~\ref{prop:conv} states the convergence of $\bar{\bm{\theta}}_\mathcal{P}$ to ${\bm{\theta}}^*$, while Theorem~\ref{th:main} and Corollary~\ref{cor:comp} provide novel theoretical results describing the asymptotic distribution of $\bar{\bm{\theta}}_\mathcal{P}$ under different choices of $T_n$ and $\mathcal{P}$. The following assumptions combine classical domination conditions on the log-likelihood components (e.g.  \citealp{white_maximum_1982}), with the flexible setting outlined in \citet{su2023} for globally convex and locally strongly convex objective functions.
\begin{assumption}
\label{ass:regularity}
      For $k=1,\dots,K$,  $\ell_k(\bm{\theta};\bm{y})$  exists for every $\bm{\theta}$, is continuous in $\bm{\theta}$ for all $\bm{y}\in\mathcal{Y}$ and $\lvert \ell_k(\bm{\theta};\bm{y})|$ is dominated by a function integrable with respect to the distribution of $\bm{Y}$. The vector $\bm{\theta}^*$ is the unique solution to $E_{\bm{Y}}\{\sum_{k=1}^K\ell_k(\bm{\theta};\bm{Y})\}=0$.
\end{assumption}
\begin{assumption}
    \label{ass:domination}
    Each sub log-likelihood $\ell_k({\bm{\theta}};\bm{y})$ is differentiable in $\bm{\theta}$, and its gradient continuous, for all $\bm{y}\in\mathcal{Y}$. Furthermore, with $\theta_r$ being the generic $r$-th element of $\bm{\theta}$, the quantity $|\partial\ell_k(\bm{\theta};y)/\partial \theta_{r}|$ is dominated by a function integrable with respect to $p(\bm{Y};\bm{\theta}^*)$ up to the fourth power.
    In addition, $|\partial^2\ell_k(\bm{\theta};y)/\partial \theta_{r}\partial \theta_{r'}|$ exists in a $\delta$-neighbourhood of $\bm{\theta}^*$, i.e. for some $\delta>0$ such that $\lVert \bm{\theta}-\bm{\theta}^*\rVert<\delta$, where it is continuous in $\bm{\theta}$ and dominated by a function integrable with respect to $p(\bm{Y};\bm{\theta}^*)$.
\end{assumption} 
\begin{assumption}
\label{ass:localsc}
    The expected sub log-likelihoods are concave in $\bm{\theta}$ and their gradients satisfy  $\rVert \nabla E_{\bm{Y}}\{\ell_k(\bm{\theta}', \bm{Y})\} - \nabla E_{\bm{Y}}\{\ell_k(\bm{\theta}, \bm{Y})\}\lVert \leq L_{1}\rVert \bm{\theta} - \bm{\theta}' \lVert$, for $L_{1} > 0$, $k=1,\dots,K$ and $\bm{\theta},\bm{\theta}'\in\mathbb{R}^d$. Furthermore,
    for some $L_2,\delta>0$, the expected Hessian of each sub log-likelihood verifies
    $\lVert \nabla^2E_{\bm{Y}}\{\ell_k(\bm{\theta},\bm{Y})\}-\nabla^2E_{\bm{Y}}\{ \ell_k(\bm{\theta}^*,\bm{Y})\}\rVert \leq L_{2}\lVert\bm{\theta}-\bm{\theta}^*\rVert$
    with $\lVert\bm{\theta}-\bm{\theta}^*\rVert<\delta$. Additionally $\sum_{k=1}^K\nabla^2E_{\bm{Y}}\{\ell_k(\bm{\theta}^*,\bm{Y})\}$ is negative-definite.
\end{assumption}
\begin{assumption}
    \label{ass:W}
    The sampling scheme $\mathcal{P}$ is such that $E\{W_{i,k}\}=\gamma_1$, with $0<\gamma_1<1$, for $i=1,\dots,n$ and $k=1,\dots,K$. Additionally, $\lim_{n\xrightarrow{}\infty} n\gamma_1 > 0$,
        $E\left\{W_{i,k}W_{i,k'}\right\}, E\left\{W_{i,k}W_{i,k'}W_{i,k''}\right\}$, and $E\left\{W_{i,k}W_{i,k'}W_{i,k''}W_{i,k'''}\right\}$ are of order $O(\gamma_1)$, and 
        $E\left\{W_{i,k}W_{j,k'}\right\}, E\left\{W_{i,k}W_{i,k'}W_{j,k''}\right\}$, and $E\left\{W_{i,k}W_{i,k'}W_{j,k''}W_{j,k'''}\right\}$ are of order $O(\gamma_1^2)$,
    with $i\neq j$ and $k,k',k'',k'''=1,\dots,K$.
    For notation aims, let $E\{W_{i,k}W_{i,k'}\}=\gamma_2$ if $k\neq k'$.
\end{assumption}
\begin{assumption}
    \label{ass:stepsize}
    Given $\eta_0>0$, the stepsize scheduling is chosen as $\eta_t = \eta_0 t^{-c}  $ with $1/2 < c < 1$, implying $\lim_{n\xrightarrow{}\infty}\sum_t^{T_n}\eta_t=\infty$,  $\lim_{n\xrightarrow{}\infty}\sum_t^{T_n}\eta_t/\sqrt{t}<\infty$, $\lim_{n\xrightarrow{}\infty}\sum_t^{T_n}\eta_t^2<\infty$. 
\end{assumption}
Assumption~\ref{ass:regularity} collects the regularity conditions on the behavior of the likelihood function that guarantee the existence and uniqueness of $\bm{\theta}^*$ (see \citealp{varin2005}). 
Assumption~\ref{ass:domination} allows exchanging the order of integration and differentiation when working with sub-log-likelihood objects.
In addition, it guarantees the existence of $E_{\bm{Y}}\{\nabla^2\ell_k(\bm{\theta}; \bm{Y})\}$ on $\lVert \bm{\theta}-\bm{\theta}^*\rVert<\delta$, with $\delta$ small enough. Assumption~\ref{ass:localsc}, imposes Lipschitz regularity for $E_{\bm{Y}}\{\nabla \ell_k(\bm{\theta})\}$ on all $\mathbb{R}^d$, and for $E_{\bm{Y}}\{\nabla^2\ell_k(\bm{\theta})\}$ on the ball $\lVert\bm{\theta}-\bm{\theta}^*\rVert<\delta$. Furthermore, it outlines the local strong concavity of the composition of the expected composite log-likelihood in a neighborhood of $\bm{\theta}^*$, thus guaranteeing its local identifiability. Together with Assumption~\ref{ass:domination}, it allows recovering the conditions on the objective function required by \citet{su2023}.
Assumption~\ref{ass:W} requires all the weights to share the same expected value, $\gamma_1$, which does not decay faster than $1/n$. Furthermore, it assumes the cross-products across the weights of different sub log-likelihoods to be bounded up to a constant by $\gamma_1$, when referring to the same observations, and by $\gamma_1^2$, when referring to two different ones. Together with Assumption~\ref{ass:domination}, such bounds control the stochastic behavior of the SGs constructed via \eqref{eq:csg}. Assumption~\ref{ass:W} is more general than what is needed for $\mathcal{P}_1$, $\mathcal{P}_2$ and $\mathcal{P}_3$ and potentially allows for different sampling schemes. Note that $\gamma_1$ can be fixed across different choices of $\mathcal{P}$ to match the computational cost of constructing~\eqref{eq:csg}. However, for a given $\gamma_1$, different sampling schemes lead to different $\gamma_2$, which affects the correlation across weights belonging to the same observation. Theorem~\ref{th:main} shows that the pair $\gamma_1$,$\gamma_2$ is sufficient to explain the statistical efficiency implied by different choices of $\mathcal{P}$. Assumption~\ref{ass:stepsize} is a standard requirement on the decreasing scheduling of the stepsize (e.g. \citealp{polyak1992acceleration}).

With Assumptions~\ref{ass:regularity}-\ref{ass:stepsize}, it is straightforward to adapt \citet{su2023} results to prove $\bar{\bm{\theta}}_\mathcal{P}$ being a stochastic approximation of $\hat{\bm{\theta}}$ as the number of iterations $T_n$ diverge. However, $\hat{\bm{\theta}}$ gets closer to ${\bm{\theta}}^*$ as the sample size grows. It follows that $\bar{\bm{\theta}}_\mathcal{P}$ can also be used as a consistent estimator of ${\bm{\theta}}^*$ as long as both $T_n$ and $n$ diverge simultaneously.
\begin{proposition}
    \label{prop:conv}
    With Assumptions~\ref{ass:regularity}-\ref{ass:stepsize}, the output of Algorithm~\ref{algo:csgd} provides a consistent estimator of ${\bm{\theta}}^*$ when $T_n\xrightarrow{}{}\infty$ as $n\xrightarrow{}\infty$, i.e. $\bar{\bm{\theta}}_\mathcal{P} \xrightarrow[n]{a.s.}{} {\bm{\theta}}^*$.
\end{proposition}
The proof of Proposition~\ref{prop:conv} in the online Appendix is an adaptation of Lemma 17 in \citet{su2023}. However, it is worth emphasizing that it takes advantage of the expected behavior of $\bm{S}_t$ when evaluated at ${\bm{\theta}}^*$ rather than $\hat{\bm{\theta}}$. In other words, it acknowledges the data and $\bm{W}_t$ as random variables when constructing $\bm{S}_t$. This is critical to stress since it allows considering ${\bm{\theta}}^*$ as the target of $\bar{\bm{\theta}}_{\mathcal{P}}$ given the double asymptotics in $n$ and $T_n$. 

While Proposition~\ref{prop:conv} formalizes the consistency of $\bar{\bm{\theta}}_{\mathcal{P}}$, nothing has been said so far about its distributional properties. Theorem~\ref{th:main} tackles this aspect by outlining the asymptotic distribution of $\bar{\bm{\theta}}_\mathcal{P}$ according to the relative divergence rate of $n$ and $T_n$. Let us anticipate how Theorem~\ref{th:main} differs from the original inference framework for averaged stochastic optimization. When used to find the root of $\nabla c\ell_n(\bm{\theta}) = 0$, the original result in Theorem 2 in \citet{polyak1992acceleration}

describes the asymptotic behavior of $\bar{\bm{\theta}}_\mathcal{P}$ around $\hat{\bm{\theta}}$. It assumes the data is fixed and does not quantify the uncertainty of the stochastic estimates around ${\bm{\theta}}^*$. From a technical point of view, directly combining such a result with (\ref{eq:cl:asyd}) is not straightforward since the two asymptotic statements are defined on different probability spaces, namely with and without the conditioning on the observed values of $\bm{Y}$. Furthermore, as soon as we allow the data to be random, we are not able anymore to take advantage of the independence of the stochastic quantities $S_t({\bm{\theta}}^*; \bm{W}_t)$, with $t=1,\dots, T_n$, which is a critical step in the proofs presented in \citet{polyak1992acceleration}. In other words, while all the iterations still share the same dataset, its random nature induces dependence among them.

It follows that, when stochastically optimizing $c\ell_n(\bm{\theta})$, if the interest is drawing inference about ${\bm{\theta}}^*$, which is typically the case of composite likelihood methods and maximum likelihood estimation in general, the available results building on the asymptotic covariance matrix outlined in \citet{polyak1992acceleration} only provide a partial answer to the research question. To fill this gap, we provide Theorem~\ref{th:main}, which shows that $\bar{\bm{\theta}}_\mathcal{P}$ is asymptotically normally distributed around ${\bm{\theta}}^*$ and its covariance matrix changes according to both  $\mathcal{P}$ and the relative divergent behavior of $T_n$ and $n$. The choice of $\mathcal{P}$ affects the shape of the noise coming from the optimization, while the divergent behavior of $T_n$ and $n$ quantifies its relative magnitude compared to the sampling variability of the data.

\begin{theorem}
\label{th:main} Let $ n/(T_n+n)\xrightarrow[n]{}{\alpha}$, with $0\leq\alpha\leq 1$.  Under Assumptions~\ref{ass:regularity}-\ref{ass:stepsize}, for the sampling schemes $\mathcal{P}_1, \mathcal{P}_2, \mathcal{P}_3$ in Definitions~\ref{def:simple:w}-\ref{def:hyper:w}, it holds that: \\
Regime 1: If $\alpha = 0$, then 
$\sqrt{n}(\bar{\bm{\theta}}_\mathcal{P}-{\bm{\theta}}^*)\xrightarrow[n]{d}\mathcal{N}(0,\bm{H}^{-1}\bm{J}\bm{H}^{-1}) ;$ \\
Regime 2: If $\alpha = 1$  and $n^{7/9}=o(T_n)$, then 
$
\sqrt{T_n}(\bar{\bm{\theta}}_\mathcal{P}-{\bm{\theta}}^*)\xrightarrow[n]{d}\mathcal{N}(0, \bm{H}^{-1}\bm{V}_{\mathcal{P}}\bm{H}^{-1}) ;
$\\
Regime 3: If $0<\alpha <1$, then 
$\sqrt{T_n+n}(\bar{\bm{\theta}}_\mathcal{P}-{\bm{\theta}}^*) \xrightarrow[n]{d}\mathcal{N}(0,\bm{H}^{-1}\bm{V}_{\mathcal{P}}\bm{H}^{-1}/(1-\alpha) + \bm{H}^{-1}\bm{J}\bm{H}^{-1}/\alpha ),$
where $\bm{V}_{\mathcal{P}} = \lim_{n\xrightarrow{}\infty}{\gamma_1^{-2}n^{-1}}(\gamma_1-\gamma_2)\bm{H} +n^{-1}\left(\gamma_1^{-2}\gamma_2-1\right)\bm{J}=O(1)$.
\end{theorem}
The asymptotic covariance matrices in Theorem~\ref{th:main} can be described as a weighted average between $\bm{H}^{-1}\bm{V}_{\mathcal{P}}\bm{H}^{-1}$ and $\bm{H}^{-1}\bm{J}\bm{H}^{-1}$, with weights depending on the divergence rate of $T_n$ and $n$. Note that $\bm{H}$ and $\bm{J}$ are the usual matrices entering the asymptotic efficiency of the CLE, as discussed in Section~\ref{sec:clsa}. While $\bm{H}^{-1}\bm{J}\bm{H}^{-1}$ is already well known from (\ref{eq:cl:asyd}) and quantifies the variability due to $\bm{Y}$, the matrix $\bm{H}^{-1}\bm{V}_{\mathcal{P}}\bm{H}^{-1}$ can be shown to describe the noise coming from the optimization. As the notation stresses, the value of $\bm{V}_{\mathcal{P}}$ depends on the choice of $\mathcal{P}$. In particular, it results in a linear combination of the matrices $\bm{H}$ and $\bm{J}$, with coefficients based on the quantities $\gamma_1$ and $\gamma_2$. For a detailed derivation of $\bm{V}_{\mathcal{P}}$, see the proof of the theorem in the online Appendix B. Before describing Corollary~\ref{cor:comp}, which outlines the different shapes of $\bm{V}_{\mathcal{P}}$ according to the choices $\mathcal{P}$, we briefly summarize the three asymptotic regimes described in Theorem~\ref{th:main}.

When the algorithm runs for $T_n = \omega(n)$ iterations, it holds that $n/(T_n+n)\xrightarrow{}{}0$, and the estimates obtained fall under Regime 1. With such a setting, the noise component generated by the optimization is negligible compared to the sampling variability of the data. In this scenario, the algorithm runs until closely approximating the CLE, i.e., there is no difference between the asymptotic behaviors of $\bar{\bm{\theta}}_\mathcal{P}$ and $\hat{\bm{\theta}}$. Inference can be carried out based on the familiar $\bm{H}^{-1}\bm{J}\bm{H}^{-1}$, as described in (\ref{eq:cl:asyd}). 

In the opposite setting, where the algorithm is stopped at $T_n = o(n)$, we get $n/(T_n+n)\xrightarrow{}{}1$ and Regime 2 holds. Note that setting $\gamma_1=\Theta(1/n)$, as in the three sampling schemes considered, requires a minimum growing rate on the number of iterations to establish asymptotic normality, i.e. $n^{7/9}=o(T_n)$. While such a condition is not particularly restrictive in practice, its technical derivation can be found in Lemma 4 in the online Appendix. In such an asymptotic regime, the dominant variance component is the one induced by $\mathcal{P}$, such that inference can potentially ignore the variability of the data. In this case, the asymptotic distribution resembles the one in \citet{polyak1992acceleration} and related works, apart from having the parameter-dependent quantities evaluated at ${\bm{\theta}}^*$ rather than $\hat{\bm{\theta}}$. To glimpse the connection between Regime 2 and the conditional inference setting traditionally adopted in stochastic optimization, imagine $n$ being so large that the distance between $\hat{\bm{\theta}}$ and ${\bm{\theta}}^*$ is negligible. Then, there is not much difference in practice between using $\bar{\bm{\theta}}_\mathcal{P}$ to infer either on ${\bm{\theta}}^*$ or $\hat{\bm{\theta}}$. In other words, the results in \citet{polyak1992acceleration} are equivalent, in a suitable sense, to inference under Regime 2 of Theorem~\ref{th:main}.

Regime 3 describes an intermediate setting between the previous two. Since $T_n$ and $n$ grow at the same rate, i.e., $T_n=\Theta(n)$, it holds that $0<\alpha <1$ and the asymptotic covariance matrix around ${\bm{\theta}}^*$ compounds for both the optimization error and the sampling variability of the data. As it is difficult to assess whether $\bar{\bm{\theta}}_\mathcal{P}$ is obtained strictly under Regime 1 or Regime 2, it is always recommended to use Regime 3.

Hence, according to the divergent behavior of $n/(T_n+n)$, Theorem~\ref{th:main} highlights which variability component can be neglected and which can not when quantifying the uncertainty around stochastic estimates.

Furthermore, note that for ease of exposition, Theorem~\ref{th:main} assumes $\gamma_1=\Theta(1/n)$ as in the three sampling schemes $\mathcal P_1$-$\mathcal P_3$. However, such results can be extended for whatever choice of $\mathcal{P}$ satisfying Assumption~\ref{ass:W}. The three asymptotic regimes and the variance decomposition are still valid, but different sampling rates $\gamma_1$ require a careful analysis of the magnitude of $V_{\mathcal{P}}$. Nevertheless, it is interesting to comprehend how the distribution of $\bm{W}$ affects the optimization noise. In this regard, Corollary~\ref{cor:comp} outlines the effects of choosing $\mathcal{P}$ according to Definitions \ref{def:simple:w},  \ref{def:be:w} and \ref{def:hyper:w}. In particular, the choice of $\mathcal{P}$ affects the values of $\gamma_1$ and $\gamma_2$ which control the shape of $\bm{V}_{\mathcal{P}}$. 
\begin{corollary}
\label{cor:comp} Let Theorem~\ref{th:main} hold. Then, Definition~\ref{def:simple:w}, implies $\bm{V}_{\mathcal{P}_1}= \bm{J}$ and $\bm{H}^{-1}\bm{V}_{\mathcal{P}_1}\bm{H}^{-1} = \bm{H}^{-1}\bm{J}\bm{H}^{-1}$; Definition~\ref{def:be:w} implies $\bm{V}_{\mathcal{P}_2}= \bm{H}$ and $\bm{H}^{-1}\bm{V}_{\mathcal{P}_2}\bm{H}^{-1} = \bm{H}^{-1}$; Definition~\ref{def:hyper:w} implies $\bm{V}_{\mathcal{P}_3}= \bm{H}$ and $\bm{H}^{-1}\bm{V}_{\mathcal{P}_3}\bm{H}^{-1} = \bm{H}^{-1}$.
\end{corollary}
While it is apparent that $\gamma_1=1/n$ for all three sampling schemes, we leave the details about the implied values of $\gamma_2$ in the proof of Corollary~\ref{cor:comp} provided in the online Appendix B.
When $\mathcal{P}$ is chosen according to Definition~\ref{def:simple:w}, the Sampling Step of Algorithm~\ref{algo:csgd} keeps untouched the correlation structure among the sub-likelihood components of the objective function. In other words, it samples from the empirical distribution of the data. Therefore, the variability due to $\bm{W}$ takes the same shape as the one stemming from $\bm{Y}$, represented by the matrix $\bm{J}$. Note, in fact, that when $\mathcal{P}_1$ is chosen, $\gamma_1 = \gamma_2$ and therefore $\bm{H}$ asymptotically disappears when computing $\bm{V}_{\mathcal{P}}$ following Theorem~\ref{th:main}.
Instead, if $\mathcal{P}$ is chosen according to Definition~\ref{def:be:w}, the sub log-likelihood components are drawn independently, even when belonging to the same observation. This step breaks the original correlation structure among the summands in $\nabla c\ell_n(\bm{\theta})$, such that $\bm{V}_{\mathcal{P}}$ collapses onto the expected second derivative of $\bm{S}_t$, $\bm{H}$. The weights independence, in fact, implies $\gamma_1^2=\gamma_2$ and, therefore, a zero weight for $\bm{J}$ when computing $\bm{V}_{\mathcal{P}}$.
Finally, if $\mathcal{P}$ follows Definition~\ref{def:hyper:w}, the correlation among the elements keeps the weight for $\bm{J}$ different from zero but asymptotically negligible because of being $O(\frac{1}{nK})$. Thus, asymptotically, $\mathcal{P}_3$ shares the same asymptotic efficiency of $\mathcal{P}_2$. 

It is well known that $\bm{H}\neq \bm{J}$ when using composite likelihood methods. Hence, the inference with $\hat{\bm{\theta}}$ must be based on $\bm{H}^{-1}\bm{J}\bm{H}^{-1}$ rather than $\bm{H}^{-1}$, which typically results in inflated variances for each parameter. Corollary~\ref{cor:comp} shows that $\mathcal{P}_1$ injects this same variability as noise in the optimization, while $\mathcal{P}_2$ and $\mathcal{P}_3$ constrain it to $\bm{H}^{-1}$.
Such a difference is evident, as the simulations in Section~\ref{sec:sims} highlight, with the estimates based on $\mathcal{P}_2$ and $\mathcal{P}_3$ exhibiting lower variability than those obtained with $\mathcal{P}_1$.
\section{Simulation Studies}
\label{sec:sims}
We investigate the results from some simulation experiments with $R=500$ replications for the models presented in Examples \ref{ex:graph} and \ref{ex:pl}. In particular, our goal is two-fold. First, we provide evidence to support Proposition~\ref{prop:conv}, namely to show that $\bar{\bm{\theta}}_\mathcal{P}$ converges to ${\bm{\theta}}^*$ when both $T_n$ and $n$ diverge by tracking the average mean square error of $\bar{\bm{\theta}}_\mathcal{P}$. i.e. 
$    MSE = \frac{1}{dR}\sum_{r = 1}^R\lVert \bar{\bm{\theta}}_\mathcal{P}^{(r,t)}-{\bm{\theta}}^* \rVert^2,
$
where $\bar{\bm{\theta}}_\mathcal{P}^{(r,t)}$ is the $d$-dimensional output of Algorithm~\ref{algo:csgd} on the $r$-th replication when stopped at the $t$-th iteration. Furthermore, we highlight that different choices of $\mathcal{P}$ characterize different behavior of the MSE trajectories because of the implied asymptotic variabilities outlined in Corollary~\ref{cor:comp}.

Second, we assess the empirical coverage performance of confidence intervals built from the asymptotic covariance matrices outlined in Theorem~\ref{th:main}. We aim to highlight the strength of asymptotic Regime 3, which 
compounds both the sampling variability of the data and the optimization noise.  To construct the confidence intervals, an estimate for both $\bm{H}$ and $\bm{J}$ is needed. Here, we use the usual sample estimators \cite[see e.g.][Section 5]{varin_overview_2011} $
    \hat {\bm{J}}^{(r,t)} = \frac{1}{n}\sum_{i=1}^n\left\{\nabla c \ell(\bar{\bm{\theta}}_\mathcal{P}^{(r,t)}; \bm{y}_i)\right\}\left\{\nabla c \ell(\bar{\bm{\theta}}_\mathcal{P}^{(r,t)}; \bm{y}_i)\right\}^\top$ and $
    \hat {\bm{H}}^{(r,t)}= \frac{1}{n}\sum_{i=1}^n\sum_{k=1}^K\left\{\nabla \ell_k(\bar{\bm{\theta}}_\mathcal{P}^{(r,t)}; \bm{y}_i)\right\}\left\{\nabla \ell_k(\bar{\bm{\theta}}_\mathcal{P}^{(r,t)}; \bm{y}_i)\right\}^\top$.
In all experiments, we burn the first $0.25n$ iterations and start drawing inference and tracking estimates variability from $T_n=.5n$. To correctly assess the empirical coverage of confidence intervals iteration-wise, all simulations track results for $T_n\in\{0.5n, \dots, 3n \}$. Hence, for each stopping time, we observe all the $R$ runs of Algorithm~\ref{algo:csgd} for a given $\mathcal{P}$. 
Results are shown for the decreasing stepsize scheduling outlined in Assumption~\ref{ass:stepsize} with $c$ set arbitrarily small at $c=0.501$. The initial step size instead, $\eta_0$, is chosen differently for the two examples.

In the experiments we compare the performances of $\bar{\bm{\theta}}_{\mathcal{P}_1}$ (\texttt{standard}), $\bar{\bm{\theta}}_{\mathcal{P}_2}$ (\texttt{bernoulli}), $\bar{\bm{\theta}}_{\mathcal{P}_3}$ (\texttt{hyper}) together with the implementations of $\bar{\bm{\theta}}_{\mathcal{P}_1}$ and $\bar{\bm{\theta}}_{\mathcal{P}_3}$ taking advantage of a recycled Sampling Step (\texttt{recycle\_standard} and \texttt{recycle\_hyper} respectively) as described in Section~\ref{sec:imprem}. We also compute $\hat{\bm{\theta}}$ numerically as a benchmark. 
\subsection{Experiments for Example \ref{ex:graph}}
Data are generated by using the exact probabilities of observing each of the possible $p$-variate realizations of the graph. 
We assume the true graph follows a two-row grid structure, similar to the simulation setting of \citet{lee_learning_2015}. In particular, horizontal edges are set at $0.5$, vertical ones at $-0.5$, intercepts at $-0.5$ for odd nodes and $0.5$ for even ones. The optimization always starts at the null vector.

We investigate the performances of Algorithm~\ref{algo:csgd} with $n\in\{2,500, 5,000, 10,000\}$ and $p\in\{10, 20\}$,  implying $d\in\{55, 210\}$. The value of $\eta_0$,  is picked by minimising over a grid of possible candidates the mean square error of \texttt{standard} at $T_n = 3n$ in the most challenging simulation setting, i.e. $n = 2,500, p = 20$. However, additional simulation results for different choices of $\eta_0$ are available in the online Appendix C.   
Figure~\ref{fig:isi:1:mse} shows the convergence of all instances of the proposed estimator in terms of average mean square distance from ${\bm{\theta}}^*$. 
Interestingly, the different noise levels introduced by different sampling schemes characterize the convergence behavior as the estimation proceeds. That is, $\bar{\bm{\theta}}_{\mathcal{P}_2}$ and $\bar{\bm{\theta}}_{\mathcal{P}_3}$ share the same asymptotic performances and are preferable to $\bar{\bm{\theta}}_{\mathcal{P}_1}$. As expected, the recycled implementations of $\mathcal{P}_1$ and $\mathcal{P}_3$ do not show any relevant discrepancy from their non-recycled versions. Furthermore, after reaching $T_n=3n$, none of the stochastic estimators has reached the MSE of the numerical optimizer. This phenomenon happens because when the stochastic algorithm is stopped, the optimization noise is still non-negligible, which affects the variance considered in the MSE computation. It follows that, since this noise can not be neglected, it must be appropriately considered when quantifying the uncertainty around $\bar{\bm{\theta}}_\mathcal{P}$.

\begin{figure}[t]\centering
	\includegraphics[width=.7\textwidth]{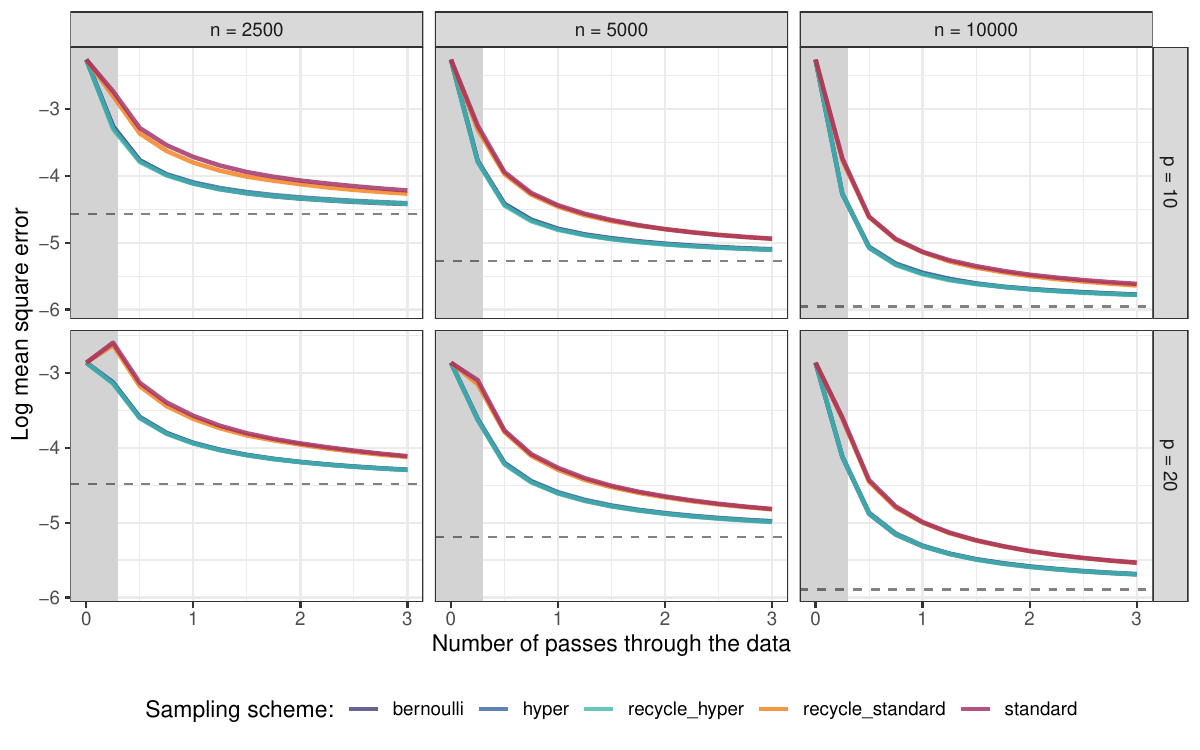}
	\caption{\label{fig:isi:1:mse}Ising model. Log mean square error trajectories along the optimization, grouped by $n$ and $p$. Solid lines refer to $\bar{\bm{\theta}}_\mathcal{P}$ under different sampling schemes. Dashed lines denote the performance of the numerical approximation of $\hat{\bm{\theta}}$. Grey areas highlight the burn-in.}
\end{figure}

\begin{figure}[!ht]\centering
	\includegraphics[width=.73\textwidth]{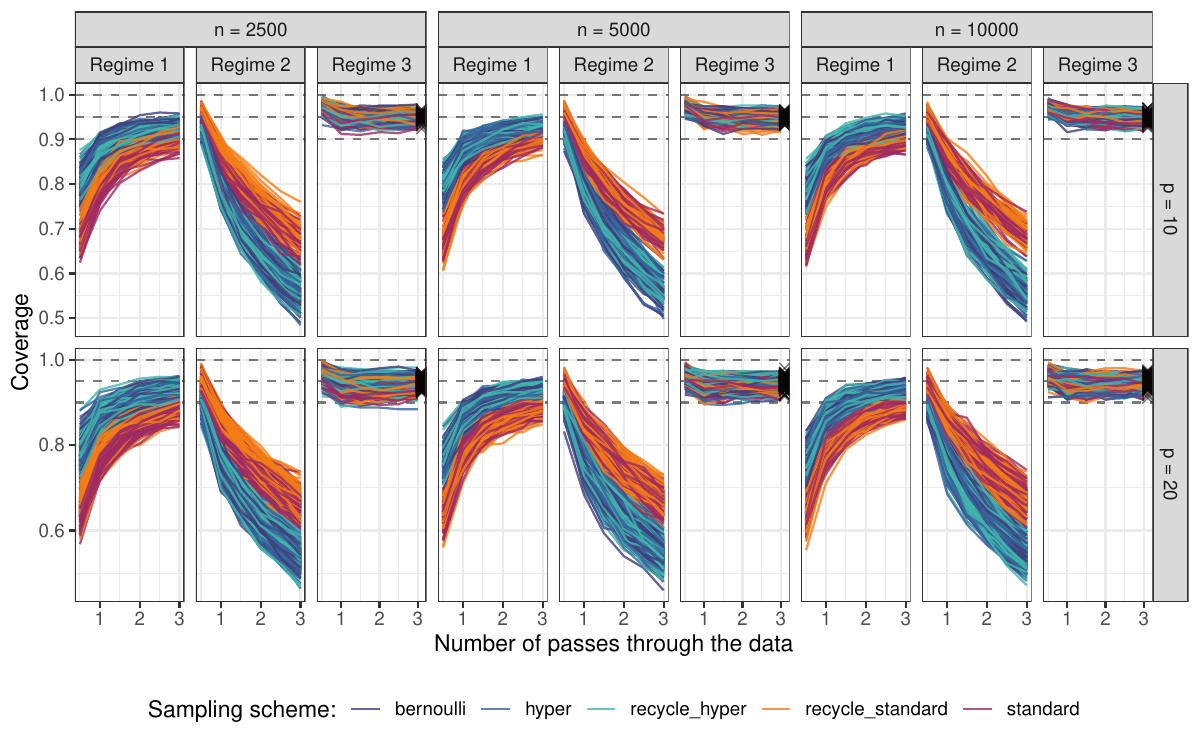}
	\caption{\label{fig:isi:1:cov_lines}Ising model. Empirical coverage of confidence intervals for $\bar{\bm{\theta}}_\mathcal{P}$ constructed according to Theorem~\ref{th:main}. Results are grouped by $n$ and $p$. Dashed lines highlight the nominal coverage level $95\%$. Solid lines refer to scalar elements of $\bar{\bm{\theta}}_\mathcal{P}$ under different sampling schemes. Dark crosses refer to scalar elements of the numerical approximation of $\hat{\bm{\theta}}$ (placed after the third pass for visualization purposes).}
\end{figure}
In this regard, Figure~\ref{fig:isi:1:cov_lines} shows that, by appropriately accounting for both sources of variability when the algorithm is stopped, it is possible to draw inference about ${\bm{\theta}}^*$ using $\bar{\bm{\theta}}_\mathcal{P}$, whatever the choice of $\mathcal{P}$.
It presents the empirical coverage levels obtained by constructing confidence intervals following the covariance matrices outlined in Theorem~\ref{th:main} under the three asymptotic regimes. As predicted by the theory, one should use Regime 1 when $T_n = \omega(n)$, and Regime 2 in the opposite scenario, $T_n = o(n)$. However, Regime 3 is the recommendable choice in practice because it directly compounds both the optimization uncertainty and the data sampling variability. As a reference, under Regime 3, Figure~\ref{fig:isi:1:cov_lines} reports the empirical coverage levels obtained by constructing confidence intervals for the numerical optimizer estimating the asymptotic covariance matrix in (\ref{eq:cl:asyd}). 

For space reasons, computational times are reported in the online Appendix C. Briefly, taking advantage of a recycling window of iterations is highly beneficial implementation-wise, especially with diverging $n$. In particular, it allows \texttt{recycle\_hyper} to be computationally competitive with \texttt{standard} and \texttt{recycle\_standard} while being systematically more efficient in statistical terms, whatever the choice of $T_n$ (and of $\eta_0$, as remarked in the additional experiments reported in the online Appendix).

\subsection{Experiments for Example \ref{ex:pl}}
While the previous example clearly shows the statistical convenience of relying on $\mathcal{P}_2$ or $\mathcal{P}_3$ rather than $\mathcal{P}_1$,
the experiments in this second example illustrate how these differences vary based on the model considered. Since the discrepancy in the asymptotic covariance of the considered estimators depends on the matrices $\bm{H}$ and $\bm{J}$, such a gap can be more or less evident according to the model analyzed. Compared to the previous example, this difference is much more apparent in the gamma frailty model, as illustrated below.  
 \begin{figure}[h]\centering
	\includegraphics[width=.73\textwidth]{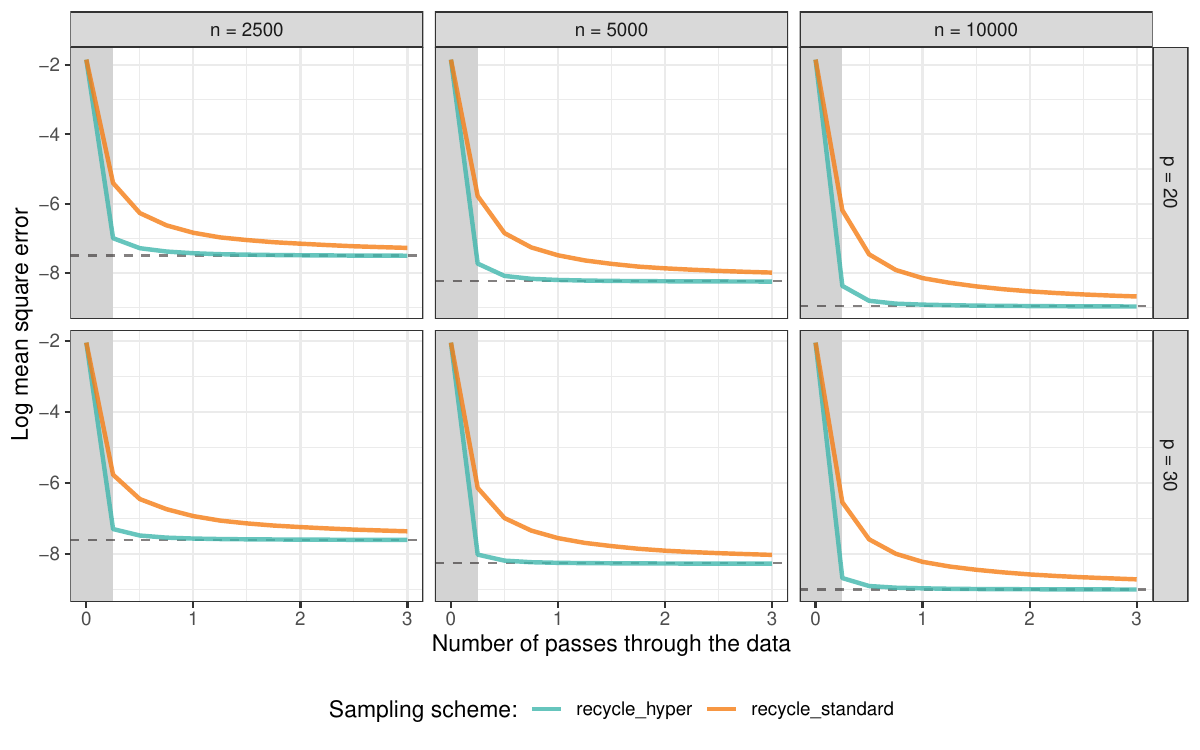}
	\caption{\label{fig:gf:1:mse} Gamma frailty model. Log mean square error trajectories grouped by $n$ and $p$. Solid lines refer to $\bar{\bm{\theta}}_\mathcal{P}$ under different sampling schemes. Dashed lines denote the performance of the numerical approximation of $\hat{\bm{\theta}}$. Grey areas highlight the burn-in.}
\end{figure}
Similarly to the previous experiments, we assess the performances of Algorithm~\ref{algo:csgd} for $n\in\{2,500, 5,000, 10,000\}$ and $p\in\{20,30\}$. Note that, from $p=20$ to $p=30$, the dimension of the parameter space goes from $d=22$ to $d=32$ while the computational burden per iteration more than doubles since $K$ increases from $K=190$ to $K = 435$. Given the large number of likelihood contributions considered compared to the Ising model, we only assess the performances of \texttt{recycle\_standard} and \texttt{recycle\_hyper}, not their exact versions. In particular, $\mathcal{P}_3$ is the sampling scheme suffering the most from the hefty $K$. Accordingly,  we increase the recycling window length to $l = 500$ for both estimators to have competitive computational times. For completeness, the online Appendix C provides further experiments showing how different values of $l$ lead to comparable estimates for \texttt{recycle\_hyper} while having massive impacts on computational times. Similarly to previous experiments, the value of $\eta_0$ shown is the stepsize minimizing the mean square error performance at $T_n = 3n$ of \texttt{recycle\_standard} in the most challenging setting, i.e., $n = 2,500, p = 30$.

Figure~\ref{fig:gf:1:mse} shows the trajectories along the optimization of the log mean square error for the proposed estimators. Similarly to the Ising model, the MSE differences are due to the asymptotic variabilities implied by $\mathcal{P}_1$ and $\mathcal{P}_3$, but they are more pronounced in this case. The estimates based on $\mathcal{P}_3$ exhibit a sharp drop at the beginning of the optimization and reach the performance of the numerical estimator almost always after one pass through the data. For the estimates based on $\mathcal{P}_1$, the convergence is much slower and does not match the numerical approximation even after the maximum length tested of three passes. 
The optimization noise of $\mathcal{P}_3$ drops almost immediately to negligible levels, leading the variance of $\bar{\bm{\theta}}_{\mathcal{P}_3}$ to overlap with the one from numerical estimation closely. The noise generated by $\mathcal{P}_1$ persists much longer, translating into higher variances for the stochastic estimates throughout the optimization and, hence, higher MSE. 
Note that \texttt{recycle\_standard} is faster than \texttt{recycle\_hyper} when both are stopped at the same $T_n$.
However, the simulations show that even after $T_n = n$, \texttt{recycle\_hyper} is already closer to the numerical estimates than \texttt{recycle\_standard} at $T_n = 3n$. Thus, it represents a more efficient alternative to \texttt{recycle\_standard} both computationally and statistically. In addition, the low variability of the \texttt{recycle\_hyper} also allows for larger steps than what shown here, which permits stopping the optimization even earlier than $T_n = n$. For space reasons, we refer to the online Appendix C for additional details and results about the experiments of this section.

\section{A Network Analysis of Mental Health Data}
\label{sec:real}
To illustrate the power of the proposed methodology, we consider an application of the Ising model to the mental health data from the Epidemiologic Survey on Alcohol and Related Conditions (NESARC) - Wave 1. The NESARC is a nationally representative survey of the United States adult population, which gathered data on alcohol behavior and mental health disorders from April 2001 to June 2002 \citep{grant2003}.

We take the network psychometrics approach \cite[e.g.,][]{epskamp2018, borsboom2022}, viewing symptoms as nodes of an unknown graph and direct symptom-to-symptom interactions as edges whose parameters are to be estimated. 
We select $p = 32$ items related to antisocial disorders, high mood, low mood, panic and personality disorders, and social and other specific forms of phobia. Therefore, the dimension of the parameter space is $d = 528$. 
The items are selected among the ones with the lowest missing response rates, avoiding screening items and related ones, and the remaining observations with missing values were discarded, leaving the dataset with a total of $31,826$ respondents. See the online Appendix D for the description of the 32 items considered. We hold out $10\%$ of the available observations as a validation set to monitor the out-of-sample behavior of the negative composite log-likelihood during the iterations. The training partition retains $n=28,643$ observations.

The model is estimated using the hypergeometric sampling of Definition~\ref{def:hyper:w}. Given the large sample size, we set the recycling window at $l=1,000$ and burn-in period $B = 0.25 n$. The stepsize scheduling is defined by $c=.501$ and $\eta_0$ chosen by halving an initial proposal until the holdout negative composite log-likelihood performance ceases improving when evaluated at $T_n = n$. The selected value is $\eta_0=5$.
After every $0.25n$ iterations, the algorithm performs a new evaluation of the holdout negative log-likelihood. When the improvement falls under $0.1\%$, the algorithm stops. In our case, it stops at $T_n = 1.75n$. The full estimation procedure, including the initial stepsize selection, took almost $15$ seconds when run on a single core of a personal laptop\footnote{Intel i5-2520M; RAM 8 GB; R version 4.3.0; gcc version 13.1.1; 4x 3.2GHz, OS Manjaro Linux 23.0.0}. As a benchmark, the numerical estimator took more than half an hour to converge on the same hardware, providing similar results.

At the end of the stochastic estimation, the asymptotic standard errors are computed by estimating the covariance matrix of $\bar{\bm{\theta}}_{\mathcal{P}_3}$ under Regime 3 of Theorem~\ref{th:main} using the usual sample estimators of $\bm{H}$ and $\bm{J}$. To investigate the structure of the estimated graphical model, all the $d$ parameters are tested against the null hypothesis of being zero. The resulting p-values are then adjusted via the Holm correction \citep{holm1979} to control for the family-wise error rate across the $d$ hypothesis at level $0.01$. The procedure identifies $17.1\%$ of the possible edges as statistically significant, as visualized in Figure~\ref{fig:mh:struct}. 

\begin{figure}[t]\centering
	\includegraphics[width=.7\textwidth]{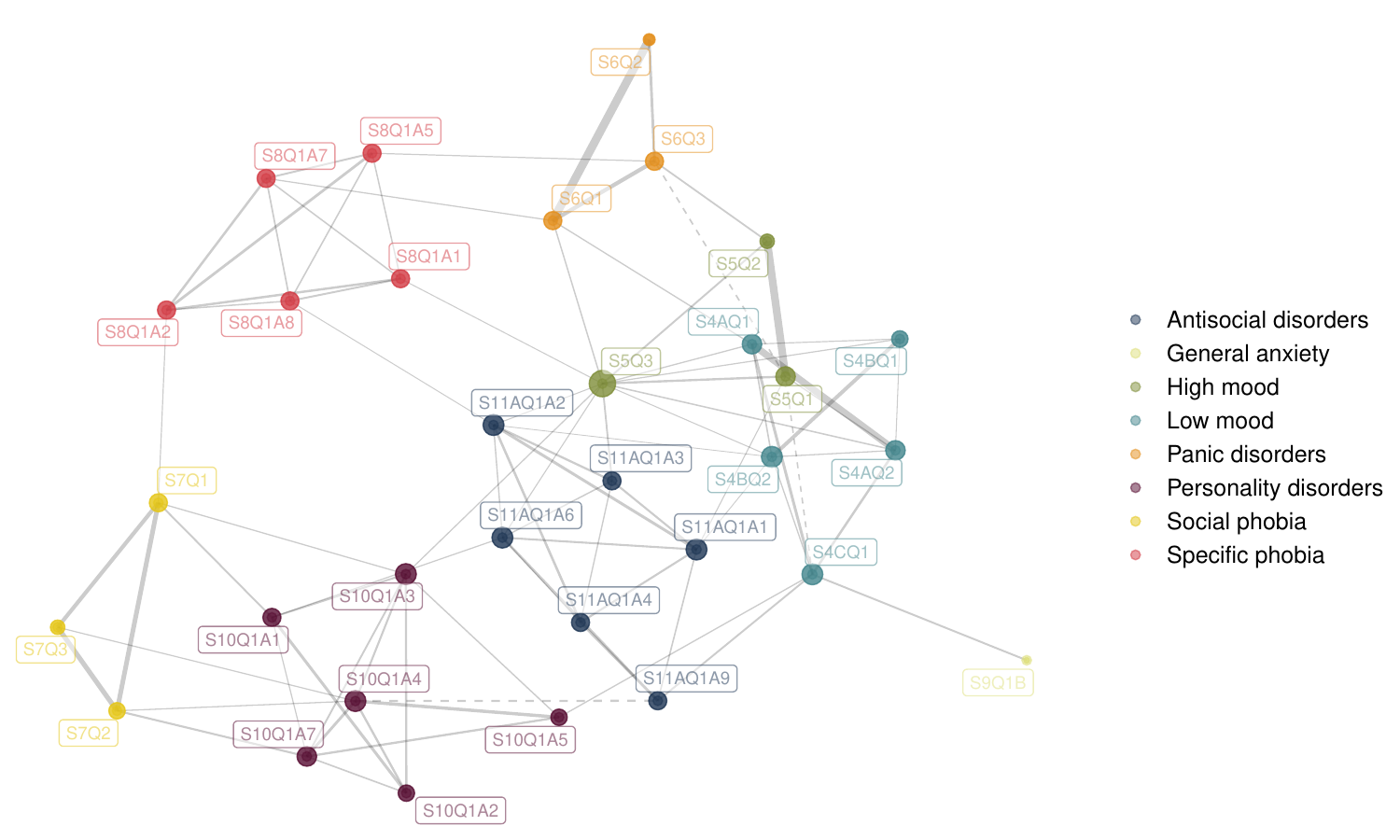}
	\caption{\label{fig:mh:struct} Graphical structure of mental health disorders. Node colors refer to specific survey areas. Solid and dashed lines stand for positive and negative estimated edges. Edge width is proportional to the absolute estimated coefficient.}
\end{figure}

The identified graphical structure highlights how symptoms of the same disorder tend to cluster together with dense positive connections. 
Instead, relationships among different disorders are less strict, with some of them being slightly negative.  Furthermore, the sparsity induced by the non-significance of many edges allows for providing conditional independence statements among symptom areas. For example, panic disorders are conditionally independent of the rest of the graph, given the specific phobias and mood disorders areas. The same can be said for the social phobia area, which is independent of mood disorders, for example, when conditioned on personality disorders and other specific types of phobias. Similar reasoning can be used to investigate symptoms belonging to the same area. For example, item \texttt{S6Q2} concerns the experience of feeling erroneously in danger after a panic attack. This symptom is isolated from the others when the remaining two items related to panic attacks, \texttt{S6Q1} and \texttt{S643}, are considered. In particular, such items refer to experiencing panic episodes for no real reason and misinterpreting nerves as a heart attack. Finally, it can also happen that single items separate two specific portions of the graph. That is the case of the node \texttt{S4CQ1}, which concerns having experienced two or more years of depression and separates from the rest of the symptoms item \texttt{S9Q1B},  which is related to experiencing six months or longer of nervousness about everyday problems.

\section{Discussion}
When the optimization noise is non-negligible, it is crucial to properly quantify the uncertainty around stochastic approximations if the goal is to run valid frequentist inferences about true parameters. We show how the asymptotic variance of such estimators compounds two sources of uncertainty: the sampling variability of the data and the noise injected in the procedure by the SGs. We optimize composite likelihoods by constructing the SGs using a hypergeometric sampling of sub-likelihood contributions, which enhances their statistical and computational efficiencies. The resulting estimator is a flexible inferential tool for applied research. In contrast to existing methods that discard part of the data to reduce computational times \citep{dillon2010stochastic, mazo2023randomized}, our proposal utilizes available information more parsimoniously, leading to improved statistical efficiency. Additionally, a small experiment  in online Appendix E compares our method with the randomized pairwise likelihood estimator of \citet{mazo2023randomized}. The results confirm that spreading the usage of likelihood components across iterations via stochastic optimization rather than discarding them leads to improved estimation performances. 

Various extensions of the proposed method are possible by expanding the scope of the parameter update in Algorithm~\ref{algo:csgd}. A first straightforward extension relates to quasi-Newton alternatives of standard SG descent \citep{byrd2011}. Such extensions can be quite effective in practice because they adapt the steps to the different scales of each parameter, which typically improves the convergence of the estimator. 

A second forthright extension enriches the update step to account for proximal operators, allowing for non-differentiable terms like projections and lasso penalties, as investigated in \citet{atchade2017} and \citet{zhang2022}. 

Nevertheless, the current proposal still has limitations, particularly when making inferences with an increasing parameter space. From a computational perspective, it does not address the challenge of computing $\bm{H}^{-1}$ and $\bm{J}$. 

It is important to have accurate estimates of both $\bm{H}^{-1}$ and $\bm{J}$ to construct reliable confidence intervals. However, estimating these quantities can be computationally challenging, especially with large numbers of parameters due to the matrix inversion required. 
Furthermore, this work focuses on settings where traditional frequentist estimation is theoretically adequate but computationally inconvenient, such as moderate parameter spaces with much larger sample sizes. Further research is needed to expand the current theoretical framework to settings where a regularization term is necessary to identify the parameters of interest. Conducting inference in such settings is complicated due to the bias introduced by regularization. We are exploring potential solutions based on recent advances in debiasing techniques for lasso-based estimators. These methods have gained popularity in both offline settings \citep{jankova2018} and with streaming data  \citep{han2023}. Addressing these theoretical challenges could enable composite likelihood inference
for large-scale data. 


    \printbibliography

\section{Consistency}
Let $f(\bm{\theta})=n^{-1}E_Y\{-\cln\}$. Then, we can decompose the stochastic gradient with $n^{-1}{\bm{S}_t} = \nabla f(\bm{\theta}_{t-1}) + {\bm{\xi}}({\bm{\theta}_{t-1}},  {\bm{W}_t}, \bm{Y}_1, \dots, \bm{Y}_n)$, where ${\bm{\xi}}({\bm{\theta}},  {\bm{W}_t}, \bm{Y}_1, \dots, \bm{Y}_n)$ represents the difference between the stochastic gradient and the gradient of the target objective function when both are evaluated at $\bm{\theta}$ during iteration $t$. 
For conciseness, in the following, we use the notation ${\bm{\xi}} = {\bm{\xi}}({\bm{\theta}}, {\bm{W}}, \bm{Y}_1, \dots, \bm{Y}_n)$, ${\bm{\xi}_t} = {\bm{\xi}}({\bm{\theta}}_{t-1}, {\bm{W}_t}, \bm{Y}_1, \dots, \bm{Y}_n)$ and ${\bm{\xi}_t}^* = {\bm{\xi}}({\bm{\theta}}^*, {\bm{W}_t}, \bm{Y}_1, \dots, \bm{Y}_n)$. Furthermore, we denote with $\mathcal{F}_{t-1}$, the filtration up to the iteration $t$ and with $\bm{\Delta}_{t}=\bm{\theta}_t-\bm{\theta}^*$ the error of the non-averaged estimate.

\textit{Proof of Proposition~1:} The proof of Proposition~1 is an adaptation of Lemma 17 in \citet{su2023} where we center the algorithm around $\bm{\theta}^*$ by taking expectations with respect to the joint distribution of $\bm{W}$ and $\bm{Y}$. In our case, thus, bounds on the stochastic gradients are influenced not only by the behaviour of the likelihood function but also by the noise introduced by the weights. To ease the presentation of the proof, we make use of Lemmata \ref{lemma:asssu} and \ref{lemma:stgr}.


\begin{lemma}
\label{lemma:asssu}
    Under Assumptions~1-4, the function $f(\bm{\theta})$ is convex and its gradient $\nabla f(\bm{\theta})$ is Lipschitz continuous. In addition, $\nabla^2f(\bm{\theta})$ exists in a neighborhood of $\bm{\theta}^*$ with $\nabla^2f(\bm{\theta})$ being positive-definite in $\bm{\theta}^*$ and satisfying 
    $\lVert \nabla^2f(\bm{\theta})-\nabla^2 f(\bm{\theta}^*)\rVert \leq L\lVert\bm{\theta}-\bm{\theta}^*\rVert$
    with $\lVert\bm{\theta}-\bm{\theta}^*\rVert<\delta$ for some $\delta,L>0$.
\end{lemma}
\textit{Proof of Lemma~\ref{lemma:asssu}:} The proof follows straightforwardly given Assumptions ~1-3. Note that, by Assumption~2, we can exchange the order of integration and differentiation using Corollary 5.9 in \citet{bartle1966}. Since, by Assumption~3, $E_{\bm{Y}}\{\ell_k(\bm{\theta},\bm{Y})\}$ is concave for all $k=1,\dots,K$, then it holds that $\sum_{k=1}^KE_{\bm{Y}}\{\ell_k(\bm{\theta},\bm{Y})\}$ is still concave and, thus, $f(\bm{\theta})$ is convex. Given the Lipschitz continuity of the single expected log sub-likelihood gradients, for some $\bm{\theta},\bm{\theta}'$, we can write
\begin{align*}
&\rVert \nabla f(\bm{\theta}') - \nabla f(\bm{\theta})\lVert = \rVert \sum_{k=1}^K[E_{\bm{Y}}\{\nabla\ell_k(\bm{\theta}',\bm{Y})\} - E_{\bm{Y}}\{\nabla\ell_k(\bm{\theta},\bm{Y})\}]\lVert \leq KL_{1}\rVert \bm{\theta} - \bm{\theta}' \lVert
\end{align*}
implying that $\nabla f(\bm{\theta})$ is Lipschitz continuous with constant $KL_{1}>0$. In addition, by Assumption~2 it directly follows that $\nabla^2f(\bm{\theta})$ exists for $\lVert\bm{\theta}-\bm{\theta}^*\rVert\leq\delta$. Furthermore, $\nabla^2f(\bm{\theta}^*)$ is positive-definite since $\sum_{k=1}^K\nabla^2E_{\bm{Y}}\{\ell_k(\bm{\theta}^*,\bm{Y})\}$ is negative definite by Assumption~3. Finally, we can show the local Lipschitz continuity of $\nabla^2f(\bm{\theta}^*)$ by stating that, for $\lVert\bm{\theta}-\bm{\theta}^*\rVert<\delta$, it holds that
\begin{align*}
\rVert \nabla^2 f(\bm{\theta}') - \nabla^2 f(\bm{\theta})\lVert = \rVert \sum_{k=1}^K[E_{\bm{Y}}\{\nabla^2\ell_k(\bm{\theta}',\bm{Y})\} - E_{\bm{Y}}\{\nabla^2\ell_k(\bm{\theta},\bm{Y})\}]\lVert \leq KL_{2}\rVert \bm{\theta} - \bm{\theta}' \lVert,
\end{align*}
which completes the proof of Lemma~\ref{lemma:asssu}.

\begin{lemma}
\label{lemma:stgr}
Under Assumptions~1, 2 and 4, it holds that $E\big\{\big\lVert\bm{\xi}_t\big\rVert^2|\mathcal{F}_{t-1}\big\}\leq C\big(1+\lVert\Delta_{t-1}\rVert^2\big)$ for some $C>0$. 

\end{lemma}
\textit{Proof of Lemma~\ref{lemma:stgr}:}
Let us start by noting that $\lVert\bm{\xi}_t\rVert^2\leq 2 \lVert n^{-1}\bm{S}_t\rVert^2 + 2\lVert \nabla f(\bm{\theta}_{t-1})\rVert^2$. For the first term, consider that 
\begin{align*}
    E\big\{\big\lVert n^{-1}\bm{S}_t\big\rVert^2|\mathcal{F}_{t-1}\big\} &= \frac{1}{n^2\gamma_1^2}E\big\{\big\lVert \sum_{i=1}^n\sum_{k=1}^K w_{ik}\nabla\ell_k(\bm{\theta}_{t-1}; \bm{Y}_i) \big\rVert^2|\mathcal{F}_{t-1}\big\}\\
    &=\frac{1}{n^2\gamma_1^2}\sum_{i=1}^n\sum_{j=1}^n\sum_{k=1}^K\sum_{h=1}^K E\big\{w_{ik}w_{jh}\big\}E\big\{\big\lvert\nabla\ell_k(\bm{\theta}_{t-1}; \bm{Y}_i)^\top\nabla\ell_h(\bm{\theta}_{t-1}; \bm{Y}_j)\big\rvert|\mathcal{F}_{t-1}\big\}\\
    &=\frac{1}{n^2\gamma_1^2}\sum_{i=1}^n\sum_{k=1}^K\sum_{h=1}^K E\big\{w_{ik}w_{ih}\big\}E\big\{\big\lvert\nabla\ell_k(\bm{\theta}_{t-1}; \bm{Y}_i)^\top\nabla\ell_h(\bm{\theta}_{t-1}; \bm{Y}_i)\big\rvert|\mathcal{F}_{t-1}\big\} +\\
    &+\frac{2}{n^2\gamma_1^2}\sum_{i<j}\sum_{k=1}^K\sum_{h=1}^K E\big\{w_{ik}w_{jh}\big\}E\big\{\big\lvert\nabla\ell_k(\bm{\theta}_{t-1}; \bm{Y}_i)^\top\nabla\ell_h(\bm{\theta}_{t-1}; \bm{Y}_j)\big\rvert|\mathcal{F}_{t-1}\big\}\\
    &=\frac{1}{n}\sum_{i=1}^n\sum_{k=1}^K\sum_{h=1}^K\frac{E\big\{w_{ik}w_{ih}\big\}}{n\gamma_1^2} E\big\{\big\lvert\nabla\ell_k(\bm{\theta}_{t-1}; \bm{Y}_i)^\top\nabla\ell_h(\bm{\theta}_{t-1}; \bm{Y}_i)\big\rvert|\mathcal{F}_{t-1}\big\} +\\
    &+\frac{2}{n^2}\sum_{i<j}\sum_{k=1}^K\sum_{h=1}^K \frac{E\big\{w_{ik}w_{jh}\big\}}{\gamma_1^2}E\big\{\big\lvert\nabla\ell_k(\bm{\theta}_{t-1}; \bm{Y}_i)^\top\nabla\ell_h(\bm{\theta}_{t-1}; \bm{Y}_j)\big\rvert|\mathcal{F}_{t-1}\big\}
\end{align*}
Note that, by Assumption~2, both $E\big\{\big\lVert\nabla\ell_k(\bm{\theta}; \bm{Y})\big\lVert\big\}$ and
$E\big\{\big\lvert\nabla\ell_k(\bm{\theta}_{t-1}; \bm{Y})^\top\nabla\ell_h(\bm{\theta}_{t-1}; \bm{Y})\big\rvert\big\}$ are finite.  By Assumption~4, instead, we have $E\big\{w_{ik}w_{jh}\big\}/\gamma_1^2<\infty$ and $E\big\{w_{ik}w_{ih}\big\}/(n\gamma_1^2)<\infty$ for all $i\neq j$ and $k,h=1,\dots,K$. 
Then, by means of the Cauchy-Schwarz inequality, for some $C,C'>0$ we can write
\begin{align*}
    E\big\{\big\lVert n^{-1}\bm{S}_t\big\rVert^2|\mathcal{F}_{t-1}\big\}
    &\leq C'\sum_{k=1}^K\sum_{h=1}^KE\big\{\big\lvert\nabla\ell_k(\bm{\theta}_{t-1}; \bm{Y})^\top\nabla\ell_h(\bm{\theta}_{t-1}; \bm{Y})\big\rvert|\mathcal{F}_{t-1}\big\} +\\
    &+C'\sum_{k=1}^K\sum_{h=1}^K E\big\{\big\lVert\nabla\ell_k(\bm{\theta}_{t-1}; \bm{Y})\big\rVert|\mathcal{F}_{t-1}\big\}E\big\{\big\lVert\nabla\ell_h(\bm{\theta}_{t-1}; \bm{Y})\big\rVert|\mathcal{F}_{t-1}\big\}\\
    &\leq C.
\end{align*}
Finally, since by Lemma~\ref{lemma:asssu} we can write $\lVert\nabla f(\bm{\theta}_{t})\rVert^2\leq L \lVert \bm{\Delta}_t\rVert^2$, then it holds that $E\big\{\big\lVert\bm{\xi}_t\big\rVert^2|\mathcal{F}_{t-1}\big\}\leq C(1+\lVert \bm{\Delta}_{t-1} \rVert^2)$ for some $C$ large enough, which completes the proof of Lemma~\ref{lemma:stgr}. 

The proof of Proposition~1 proceeds by showing that the error of the Robbins-Monro update, $\Delta_t=\bm{\theta}_t-\bm{\theta}^*$, vanishes as $n\xrightarrow{} \infty$. The general argument is that $\bm{\theta}_t$ will get increasingly closer to $\bm{\theta}^*$ in $L_2$ terms because of the convexity of $f(\bm{\theta})$ and the smoothness of its gradient.  As soon as the estimate gets sufficiently close to $\bm{\theta}^*$, then the convergence starts taking advantage of the local strong convexity geometry of $f(\bm{\theta})$. 

First, recall that we defined $\bm{\xi}_t=n^{-1}\bm{S}_t-\nabla f(\bm{\theta}_{t-1})$. By Lemma~\ref{lemma:stgr} it holds that $E\big\{\big\lVert\bm{\xi}_t\big\rVert^2|\mathcal{F}_{t-1}\big\}\leq C\big(1+\lVert\Delta_{t-1}\rVert^2\big)$ for some $C>0$. Additionally, by Lemma~\ref{lemma:asssu}, we can write $\rVert \nabla f(\bm{\theta})\lVert\leq L\rVert \bm{\theta}-\bm{\theta}^* \lVert$ with $L>0$ for all $\bm{\theta}$. Therefore,  it follows that
\begin{align}
\label{eq:EDelta2}
    E\big\{\lVert\bm{\Delta}_t\rVert^2|\mathcal{F}_{t-1}\big\}&=E\big\{\big\lVert \bm{\theta}_{t-1}-\frac{\eta_t}{n}\bm{S}_t-\bm{\theta}^*\big\rVert^2|\mathcal{F}_{t-1}\big\} \\
    &= \big\lVert \bm{\theta}_{t-1}-\eta_t\nabla f(\bm{\theta}_{t-1})-\bm{\theta}^*\big\rVert^2+E\big\{\big\lVert \eta_t\bm{\xi}_t\big\rVert^2|\mathcal{F}_{t-1}\big\}\notag\\
    &=\big\lVert \bm{\Delta}_{t-1}\big\rVert^2 - 2\eta_t\bm{\Delta}_{t-1}^\top\nabla f(\bm{\theta}_{t-1}) + \eta_t^2\lVert\nabla f(\bm{\theta}_{t-1})\rVert^2 + \eta_t^2E\big\{\big\lVert \bm{\xi}_t\big\rVert^2|\mathcal{F}_{t-1}\big\}\notag\\
    &\leq \big\lVert \bm{\Delta}_{t-1}\big\rVert^2 - 2\eta_t\bm{\Delta}_{t-1}^\top\nabla f(\bm{\theta}_{t-1}) + \eta_t^2L^2\lVert\bm{\Delta_{t-1}}\rVert^2 + \eta_t^2C\big(1+\lVert\bm{\Delta}_{t-1}\rVert^2\big)\notag\\
    &\leq \big(1+\eta_t^2L^2+\eta_t^2C'\big)\lVert\bm{\Delta}_{t-1}\rVert^2 -  2\eta_t{\bm{\Delta}_{t-1}}^\top\nabla f(\bm{\theta}_{t-1}) + \eta_t^2C'\notag\\
    &\leq \big(1+\eta_t^2L^2+\eta_t^2C'\big)\lVert\bm{\Delta}_{t-1}\rVert^2 -  2\eta_t\delta\lVert\bm{\Delta}_{t-1}\rVert\min\{\lVert\bm{\Delta}_{t-1}\rVert,\delta\} + \eta_t^2C,\notag
\end{align}
where the last inequality holds by Lemma 21 in \citet{su2023}, which elicit the above-used upper bound for ${\bm{\Delta}_{t-1}}^\top\nabla f(\bm{\theta}_{t-1})$ when $f(\bm{\theta})$ is a generic globally convex function strongly convex on a $\delta$-neighbourhood of $\bm{\theta^*}$.
By definition, $\lVert\bm{\Delta}_t\rVert^2, \eta_t, C, L>0$, while $\lim_{n\xrightarrow{}\infty}\sum_{t}^{T_n}\eta_t^2C<\infty$ and $\lim_{n\xrightarrow{}\infty}\sum_{t}^{T_n}\eta_t^2(L^2+C)<\infty$ by Assumption~5. Thus, we can apply the Robbins-Siegmund Theorem (\citealp{robbins1971}), which guarantees that 
\begin{align}
\label{eq:convergenceDelta2}
    &\lVert \bm{\Delta}_{T_n} \rVert^2 \xrightarrow[n]{}\psi\text{ with $0\leq\psi<\infty$}\quad\text{ a.s.},
\end{align}
and $\sum_{t}^{T_n}2\eta_t\delta\lVert\bm{\Delta}_{t-1}\rVert\min\{\lVert\bm{\Delta}_{t-1}\rVert,\delta\}<\infty$ almost surely.
Provided that $\lim_{n\xrightarrow{}\infty}\sum_{t}^{T_n}\eta_t=\infty$ by Assumption~5, the latter convergence statement implies
\begin{align}
\label{eq:convergenceDelta}
    \lVert\bm{\Delta}_{T_n}\rVert\xrightarrow[n]{}0\quad\text{ a.s.,}
\end{align} 
which concludes the proof of Proposition~1.

\section{Asymptotic Normality }
\label{app:main}
\subsection{Negligibility of the error terms}
\textbf{Proposition 2.}
Let Assumptions 1-5 be satisfied. Let the average error of the algorithm be defined as $\bar{\bm{\Delta}}_t=t^{-1}\sum_{j=1}^t{\bm{\Delta}}_j$ with ${\bm{\Delta}}_t = {\bm{\theta}}_{t}-{\bm{\theta}}^*$. Then, for $0\leq \alpha \leq 1$, it holds that
\begin{itemize}
\item Regime 1: If $\alpha=0$, then
\begin{align*}
 \sqrt{n}\bar{\bm{\Delta}}_{T_n} =\frac{\sqrt{n}}{T_n}\sum_{t=1}^{T_n-1}{\bm{H}}^{-1}{\bm{\xi}_{t}^*} + o_p(1),
\end{align*}
\item Regime 2: If $\alpha=1$, then
\begin{align*}
 \sqrt{T_n}\bar{\bm{\Delta}}_{T_n} =\frac{\sqrt{T_n}}{T_n}\sum_{t=1}^{T_n-1}{\bm{H}}^{-1}{\bm{\xi}_{t}^*} + o_p(1),
\end{align*}
\item Regime 3: If $0<\alpha<1$, then
\begin{align*}
 \sqrt{T_n+ n}\bar{\bm{\Delta}}_{T_n} =\frac{\sqrt{T_n+n}}{T_n}\sum_{t=1}^{T_n-1}{\bm{H}}^{-1}{\bm{\xi}_{t}^*} + o_p(1),
\end{align*}
\end{itemize}
where $T_n\rightarrow\infty$ as $n\rightarrow\infty$.

\textit{Proof of Proposition~2:}
Recall that the generic update takes the form $\bm{\theta}_t=\bm{\theta}_{t-1}-\eta_t\bm{S}_t/n$, and that we can decompose the stochastic gradient with $n^{-1}{\bm{S}_t} = \nabla f(\bm{\theta}_{t-1}) + \bm{\xi}_t$. By rearranging the terms we reformulate the update as $\nabla f(\bm{\theta}_{t-1})=(\bm{\theta}_t-\bm{\theta}_{t-1})/\eta_t-\bm{\xi}_t$. Furthermore, by taking the expansion of $\nabla f(\bm{\theta}_{t-1})$ around $\bm{\theta}^*$, i.e. $\nabla f(\bm{\theta}_{t-1})=\bm{H}\bm{\Delta}_{t-1} + \bm{r}_{t-1}$, the update takes the form
\begin{align}
\label{eq:upd_decomp}
\bm{\Delta}_{t-1} = \bm{H}^{-1}\frac{\bm{\theta}_t-\bm{\theta}_{t-1}}{\eta_t}-\bm{H}^{-1}\bm{\xi}_t - \bm{H}^{-1}\bm{r}_{t-1}.
\end{align}
Recall that we are interested in the asymptotic behaviour of $\bar{\bm{\Delta}}_{T_n}$, which can be rearranged as
\begin{align*}
    \bar{\bm{\Delta}}_{T_n}=\frac{1}{T_n}\sum_{t=1}^{T_n}\bm{\Delta}_t = \frac{1}{T_n}\sum_{t=1}^{T_n}\bm{\Delta}_{t-1} + \frac{1}{T_n}\sum_{t=1}^{T_n}\bm{\theta}_t-\bm{\theta}_{t-1} = \frac{1}{T_n}\sum_{t=1}^{T_n}\bm{\Delta}_{t-1} + \frac{\bm{\theta}_{T_n}-\bm{\theta}_{0}}{T_n}
\end{align*}
We proceed assuming $0<\alpha <1$. Then, the other two cases follow similarly. Note that, in such an asymptotic regime,  $(T_n+n)/T_n\xrightarrow[n]{}1/(1-\alpha)$, which is a positive finite constant, and therefore we rescale that sum by $\sqrt{T_n+n}$.
Thus, by  using (\ref{eq:upd_decomp}), it follows that
\begin{align*}    
\frac{\sqrt{T_n+n}}{T_n}\sum_{t=1}^{T_n}\bm{\Delta}_{t} &= \frac{\sqrt{T_n+n}}{T_n} \left(\bm{\theta}_{T_n}-\bm{\theta}_{0}\right)\tag{$\mathbf{I}^{(0)}$}\\
&-\frac{\sqrt{T_n+n}}{T_n}\bm{H}^{-1}\sum_{t=1}^{T_n}\bm{\xi}_t+ \tag{$\mathbf{I}^{(1)}$}\\
&- \frac{\sqrt{T_n+n}}{T_n}\bm{H}^{-1}\sum_{t=1}^{T_n}\bm{r}_{t-1}+\tag{$\mathbf{I}^{(2)}$}\\
&+\frac{\sqrt{T_n+n}}{T_n}\bm{H}^{-1}\sum_{t=1}^{T_n}\frac{\bm{\theta}_t-\bm{\theta}_{t-1}}{\eta_t}\tag{$\mathbf{I}^{(3)}$}\\
&=\mathbf{I}^{(0)}+\mathbf{I}^{(1)}+\mathbf{I}^{(2)}+\mathbf{I}^{(3)}.
\end{align*}
Let us start considering $\mathbf{I}^{(0)}$ and note that $\bm{\theta}_{T_n}-\bm{\theta}_{0}=\bm{\Delta}_{T_n}-\bm{\Delta}_0$. Since $\lVert\bm{\Delta}_0\rVert$ is finite and $\bm{\Delta}_{T_n}\xrightarrow{}0$ almost surely by Proposition~1, then $\mathbf{I}^{(0)}\xrightarrow{}0$ in probability as $n\xrightarrow{}\infty$. The proof of the asymptotic negligibility of $\mathbf{I}^{(2)}$ and $\mathbf{I}^{(3)}$ is identical to the proofs of Lemma 13 and Lemma 14 in \citet{su2023} after remarking that $\sqrt{T_n+n}/\sqrt{T_n}$ is finite in our setting. In particular, $\mathbf{I}^{(2)}$ is shown to converge to zero in probability by noting that $\lVert \bm{r}_{t-1}\rVert$ can be bounded, up to a constant, with $\lVert\bm{\Delta}_{t-1}\rVert$, when $\lVert\bm{\Delta}_{t-1}\rVert>\delta$, by taking advantage of the Lipschitz continuity of the gradients (see Lemma~\ref{lemma:asssu}), and with $\lVert\bm{\Delta}_{t-1}\rVert^2$ when $\lVert\bm{\Delta}_{t-1}\rVert\leq\delta$ and $\nabla f(\bm{\theta_{t-1}})$ can be expanded for all $\bm{\theta}'=c\bm{\theta}_{t-1}+(1-c)\bm{\theta}^*$, with $0\leq c\leq 1$, by aknowledging $\nabla^2 f(\bm{\theta_{t-1}})$ being Lipschitz continuous (see Lemma~\ref{lemma:asssu}). 
Since $\lVert\bm{\Delta}_{t-1}\rVert$ goes to zero almost surely by Proposition~1, the proof focuses on showing that also $T_n^{-1/2}\sum_t\lVert\bm{\Delta}_{t-1}\rVert^2$ converges in probability to zero by taking advantage of (\ref{eq:EDelta2}). The proof of $\mathbf{I}^{(3)}$ going to zero in probability follows by recognizing that, under Assumption~5, $\lim_{n\xrightarrow{}\infty} T_n\eta_{T_n}=\infty$.
We refer the readers to Appendix B.1 in \citet{su2023} for the detailed proofs.

With $\mathbf{I}^{(0)}+\mathbf{I}^{(2)}+\mathbf{I}^{(3)}$ converging to zero in probability, the only remaining term defining the asymptotic distribution of $\bm{\Delta}_{t}$ is given by $\mathbf{I}^{(1)}$. Given  that $\nabla\ell_k(\bm{\theta},\bm{y})$ is continuous in $\bm{\theta}$ by Assumption~2, also $\bm{\xi}$ is continuous in $\bm{\theta}$ for all $\bm{y}\in\mathcal{Y}$. Therefore, by Proposition~2 and continuous mapping, $\bm{\xi}_t\xrightarrow{}\bm{\xi}^*$ almost surely as $n$ diverges (see, for example, Theorem 2.3 in \citealp{van2000}), which completes the proof of Proposition~2.

\subsection{Normality}
\begin{lemma}
\label{lemma:w4} If $\gamma_1=\Theta(1/n)$, under Assumption~4, for $i,j=1,\dots,n$ and $i\neq j$, it holds that
    \begin{align*}
    \lim_{n\xrightarrow{}\infty} \frac{1}{T_n^4\gamma_1^4} \sum_{t,t'}\sum_{l,l'}\text{Cov}\left\{W_{i,k,t}W_{i,k',t'}, W_{j,h,l}W_{j,h',l'}\right\} = 0,
\end{align*}
if $n^{2/3}/T_n\xrightarrow{}0$ as $n$ diverges.
\end{lemma}
\textit{Proof of Lemma~\ref{lemma:w4}:}
As long as there are no iterations in common between the terms $W_{i,k,t}W_{i,k',t'}$ and $W_{j,h,l}W_{j,h',l'}$, their covariance is equal to zero. Therefore, we need only to focus on the cases of one or two shared iterations. In the following, we divide the total sum in partial sums according to the possible configurations of the iteration indexes $t,t',l,l'=1,\dots,T_n$ and show that the limit holds on each of them.
\begin{enumerate}
    \item One index shared (e.g. $t\neq t' = l \neq l' \neq t$): Such a configuration has $O(T_n^3)$ summands. In particular, we get
    \begin{align*}
        &\lim_{n\xrightarrow{}\infty} \frac{O(T_n^3)}{T_n^4\gamma_1^4} \text{Cov}\left\{W_{i,k,t}W_{i,k',t'}, W_{j,h,t'}W_{j,h',l'}\right\} =\\
        &\quad= \lim_{n\xrightarrow{}\infty} \frac{O(1)}{T_n\gamma_1^4} \left[E\{W_{i,k}\}E\{W_{j,h'}\}E\{W_{i,k'} W_{j,h}\}-E\{W_{i,k}\}E\{W_{i,k'}\}E\{W_{j,h}\}E\{W_{j,h'}\}\right]\\
        &\leq \lim_{n\xrightarrow{}\infty} \frac{O(1)}{T_n\gamma_1^4} (O(1)\gamma_1^4 - \gamma_1^4)=\lim_{n\xrightarrow{}\infty} \frac{O(1)}{T_n}=0
    \end{align*}
    \item Two varying indexes, one shared (e.g. $t = t' = l \neq l'$): Such a configuration has $O(T_n^2)$ summands. Hence
    \begin{align*}
        &\lim_{n\xrightarrow{}\infty} \frac{O(T_n^2)}{T_n^4\gamma_1^4} \text{Cov}\left\{W_{i,k,t}W_{i,k',t}, W_{j,h,t}W_{j,h',l'}\right\} =\\
        &\quad=\lim_{n\xrightarrow{}\infty} \frac{O(1)}{T_n^2\gamma_1^4} \left[E\left\{W_{i,k}W_{i,k'}W_{j,h}\right\}E\{W_{j,h'}\}-E\{W_{i,k}W_{i,k'}\}E\{W_{j,h}\}E\{W_{j,h'}\}\right]\\
        &\quad= \lim_{n\xrightarrow{}\infty} \frac{O(1)}{T_n^2\gamma_1^4}\left[E\left\{W_{i,k}W_{i,k'}W_{j,h}\right\}\gamma_1-\gamma_1^2E\{W_{i,k}W_{i,k'}\}\right]\\
        &\quad\leq\lim_{n\xrightarrow{}\infty} \frac{O(1)}{T_n^2\gamma_1^4}\left(O(1)\gamma_1^3-O(1)\gamma_1^3\right)=\lim_{n\xrightarrow{}\infty} \frac{O(1)}{T_n^2\gamma_1}=0\iff \lim_{n\xrightarrow{}\infty} \frac{n^{1/2}}{T_n}=0
    \end{align*}
    \item Two varying indexes, two shared (e.g. $t = l \neq t' = l'$): Such a configuration has $O(T_n^2)$ summands. Therefore
    \begin{align*}
        &\lim_{n\xrightarrow{}\infty} \frac{O(T_n^2)}{T_n^4\gamma_1^4} \text{Cov}\left\{W_{i,k,t}W_{i,k',t'}, W_{j,h,t}W_{j,h',t'}\right\} =\\
        &\quad=\lim_{n\xrightarrow{}\infty} \frac{O(1)}{T_n^2\gamma_1^4} \left[E\{W_{i,k} W_{j,h}\}E\{W_{i,k'}W_{j,h'}\}- E\{W_{i,k}\}E\{W_{i,k'}\}E\{W_{j,h}\}E\{W_{j,h'}\}\right]\\
        &\leq  \lim_{n\xrightarrow{}\infty} \frac{O(1)}{T_n^2\gamma_1^4}(O(1)\gamma_1^4 - \gamma_1^4) = \lim_{n\xrightarrow{}\infty} \frac{O(1)}{T_n^2}=0
    \end{align*}
    \item All indexes equal (e.g. $t = t' = l = l'$): Such a configuration has $O(T_n)$ summands. Then,
    \begin{align*}
        &\lim_{n\xrightarrow{}\infty} \frac{O(T_n)}{T_n^4\gamma_1^4} \text{Cov}\left\{W_{i,k,t}W_{i,k',t}, W_{j,h,t}W_{j,h',t}\right\}=\\
        &\quad=\lim_{n\xrightarrow{}\infty} \frac{O(1)}{T_n^3\gamma_1^4}\left[E\{W_{i,k}W_{i,k'}W_{j,h}W_{j,h'}\}- E\{W_{i,k}W_{i,k'}\}E\{W_{j,h}W_{j,h'}\}\right]\\
        &\quad\leq \lim_{n\xrightarrow{}\infty} \frac{O(1)}{T_n^3\gamma_1^4}\left(O(1)\gamma_1^2-O(1)\gamma_1^2\right)= \lim_{n\xrightarrow{}\infty} \frac{O(1)}{T_n^3\gamma_1^2} =0 \iff \lim_{n\xrightarrow{}\infty} \frac{n^{2/3}}{T_n}=0.
    \end{align*}
\end{enumerate}
\begin{lemma}
\label{lemma:Ew3} If $\gamma_1=\Theta(1/n)$, Under Assumption~4, for $i=1,\dots,n$, it holds that
\begin{align*}
    \lim_{n\xrightarrow{}\infty}\frac{1}{n^{4/6}T_n^4\gamma_1^4}\sum_{t,t'}\sum_{l,l'}\sum_{k,k'}\sum_{h,h'}E\left(W_{ikt}W_{ik't'}W_{ihl}W_{ih'l'}\right) = 0,
\end{align*}
 if $n^{7/9}/T_n\xrightarrow{}0$ as $n$ diverges,
\end{lemma}
\textit{Proof of Lemma~\ref{lemma:Ew3}:} Similarly to Lemma~\ref{lemma:w4}, we proceed by grouping the summands by the number of iteration indexes shared by the weights and take the limit of their partial sum. 
\begin{itemize}
    \item No shared index: Such a configuration appears with $O(T_n^4)$ summands, hence
    \begin{align*}
        \lim_{n\xrightarrow{}\infty} \frac{O(T_n^4)}{n^{4/6}T_n^4\gamma_1^4}E\{W_{i,k,t}W_{i,k',t'} W_{i,h,l}W_{i,h',l'}\}=\lim_{n\xrightarrow{}\infty}\frac{O(1)}{n^{4/6}}\frac{E\{W_{i,k}\}E\{W_{i,k'}\}E\{ W_{i,h}\}E\{W_{i,h'}\}}{\gamma_1^4}=\lim_{n\xrightarrow{}\infty}\frac{O(1)}{n^{4/6}}=0.
    \end{align*}
    \item One shared index (e.g. $t\neq t' = l \neq l'$, with $l' \neq t$): Such a configuration appears in $O(T_n^3)$ summands, therefore
    \begin{align*}
        \lim_{n\xrightarrow{}\infty} \frac{O(T_n^3)}{n^{4/6}T_n^4\gamma_1^4} E\left\{W_{i,k,t}W_{i,k',t'} W_{i,h,t'}W_{i,h',l'}\right\} &= \lim_{n\xrightarrow{}\infty} \frac{E\left\{W_{i,k}\}E\{W_{i,k'} W_{i,h}\}E\{W_{i,h'}\right\}}{n^{4/6}T_n\gamma_1^4} = \\
        &=\lim_{n\xrightarrow{}\infty}\frac{E\{W_{i,k'} W_{i,h}\}}{n^{4/6}T_n\gamma_1^2} \\
        &\leq \lim_{n\xrightarrow{}\infty}\frac{O(1)}{n^{4/6}T_n\gamma_1}=0\iff \lim_{n\xrightarrow{}\infty}\frac{n^{1/3}}{T_n}=0.
        \end{align*}
    \item Two shared indexes (e.g. $t=t'\neq l=l'$): Such a configuration appears in $O(T_n^3)$ summands, consequently
    \begin{align*}
    \lim_{n\xrightarrow{}\infty} \frac{O(T_n^2)}{n^{4/6}T_n^4\gamma_1^4}E\left\{ W_{i,k,t}W_{1,k',t}W_{i,h,l}W_{i,h',l}\right\} &= \lim_{n\xrightarrow{}\infty} \frac{O(1)}{n^{4/6}T_n^2}\frac{E\{W_{i,k}W_{1,k'}\}E\{W_{i,h}W_{i,h'}\}}{\gamma_1^4} \\
    &\leq \lim_{n\xrightarrow{}\infty} \frac{O(1)}{n^{4/6}T_n^2\gamma_1^2}=0\iff\lim_{n\xrightarrow{}\infty}\frac{n^{4/6}}{T_n}=0.
    \end{align*}
    \item Two shared indexes (e.g. $t=t=l\neq l')$: Such a configuration appears in $O(T_n^2)$ summands, thus
    \begin{align*}
        \lim_{n\xrightarrow{}\infty} \frac{O(T_n^2)}{n^{4/6}T_n^4\gamma_1^4}E\left\{ W_{i,k,t}W_{i,k',t}W_{i,h,t}W_{i,h',l}\right\} &= \lim_{n\xrightarrow{}\infty} \frac{O(1)}{T_n^2}\frac{E\left\{ W_{i,k}W_{1,k',t}W_{i,h}\}E\{W_{i,h'}\right\}}{n^{4/6}\gamma_1^4}\\
        &=\lim_{n\xrightarrow{}\infty} \frac{O(1)}{n^{4/6}T_n^2}\frac{E\{ W_{i,k}W_{1,k'}W_{i,h}\}}{\gamma_1^3}\\
        &\leq \lim_{n\xrightarrow{}\infty} \frac{O(1)}{n^{4/6}T_n^2\gamma_1^2}=0\iff\lim_{n\xrightarrow{}\infty}\frac{n^{4/6}}{T_n}=0.
    \end{align*}
    \item All indexes equal: Such a configuration appears in $O(T_n)$ summands, hence
    \begin{align*}
        \lim_{n\xrightarrow{}\infty} \frac{O(T_n)}{n^{4/6}T_n^4\gamma_1^4}E\left\{ W_{i,k,t}W_{i,k',t}W_{i,h,t}W_{i,h',t}\right\} &=\lim_{n\xrightarrow{}\infty} \frac{O(1)}{n^{4/6}T_n^3}\frac{E\left\{ W_{i,k}W_{i,k'}W_{i,h}W_{i,h'}\right\}}{\gamma_1^4}\\
        &\leq \lim_{n\xrightarrow{}\infty} \frac{O(1)}{n^{4/6}T_n^3\gamma_1^3}=0\iff\lim_{n\xrightarrow{}\infty}\frac{n^{7/9}}{T_n}=0.
    \end{align*}
\end{itemize}

\textit{Proof of Theorem~1:} Let Assumptions 1-5 be satisfied. Furthermore, let us write
\begin{align*}
\bar{\bm{\xi}}^* &= \frac{1}{T_n}\sum_{t=1}^{T_n}{\bm{\xi}_{n,t}^*}= \frac{1}{nT_n}\sum_{t=1}^{T_n}\sum_{i=1}^n\bm{S}_{i,t}^* = \frac{1}{n}\sum_{i=1}^n\bm{S}_{i}^*,
\end{align*}
where $\bm{S}_{i,t}^*= \gamma_1^{-1} \sum_k W_{i,k,t}\nabla\ell_k({\bm{\theta}}^*;Y_i)$ represents the contribution of the $i$-th observation to the stochastic gradient, and $\bm{S}_{i}^*={T_n}^{-1}\sum_{t}^{T_n}S_{i,t}^*$ its average along the optimisation, for $i=1,\dots,n$ and $t=1,\dots,T_n$.
Then, the conditional  variance of $\bm{S}_{i,t}^*$ takes the form
\begin{align*}
\text{Var}_{\bm{W}|\bm{Y}}(\bm{S}_{i,t}^*)&=\gamma_1^{-2}(\gamma_1-\gamma_2)\sum_{k=1}^K\left\{\nabla\ell_k({\bm{\theta}}^*;\bm{Y}_i)\right\}\left\{\nabla\ell_k({\bm{\theta}}^*;\bm{Y}_i)\right\}^\top +\\
&+\left(\gamma_1^{-2}\gamma_2-1\right)\sum_{k=1}^K\sum_{h=1}^K\left\{\nabla\ell_k({\bm{\theta}}^*;\bm{Y}_i)\right\}\left\{\nabla\ell_h({\bm{\theta}}^*;\bm{Y}_i)\right\}^\top.
\end{align*}
and its expected value is $E_{\bm{\bm{Y}}}\text{Var}_{\bm{W}|\bm{Y}}(\bm{S}_{i,t}^*) = \gamma_1^{-2}(\gamma_1-\gamma_2)\bm{H} +\left(\gamma_1^{-2}\gamma_2-1\right)\bm{J}$, since
\begin{align*}
    E_{\bm{\bm{Y}}}\left\{\sum_{k=1}^K\left\{\nabla\ell_k({\bm{\theta}}^*;\bm{Y}_i)\right\}\left\{\nabla\ell_k({\bm{\theta}}^*;\bm{Y}_i)\right\}^\top\right\}&=\bm{H},
\end{align*}
by exploiting the second Bartlett's identity on the single sub-likelihood components and
\begin{align*}
    E_{\bm{\bm{Y}}}\left\{\sum_{k=1}^K\sum_{h=1}^K\left\{\nabla\ell_k({\bm{\theta}}^*;\bm{Y}_i)\right\}\left\{\nabla\ell_h({\bm{\theta}}^*;\bm{Y}_i)\right\}^\top\right\}&=\bm{J},
\end{align*}
by definition of $\bm{J}$. Then, $\text{Var}(\bm{S}_{i}^*) = {T_n}^{-1}E_{{\bm{Y}}}\text{Var}_{\bm{W}|\bm{Y}}(\bm{S}_{i,t}^*) + \bm{J}$. 

In order to prove Theorem 1, we need to show the asymptotic multivariate normality of the vector $\bar{\bm{\xi}^*}$. We take advantage of the Central Limit Theorem for exchangeable processes outlined in  \citet{blum_chernoff_rosenblatt_teicher_1958} (Theorem 2, page 227) by exploiting the conditional independence of the random vectors $\bm{S}_1^*, \dots, \bm{S}_n^*$ given the weighting sequence $\bm{W}_1, \dots, \bm{W}_{T_n}$. For conciseness, let us focus on Regime 3; the other two are simpler and follow similarly. 
With such an asymptotic regime, we need to show that 
\begin{align}
\label{eq:app:thesis}
    \sqrt{T_n+n}\bar{\bm{\xi}^*}\xrightarrow[n]{d}\mathcal{N}\left(0, \bm{V}_{\mathcal{P}}/(1-\alpha) + \bm{J}/\alpha\right).
\end{align}
We proceed by using the Cramér-Wold device, such that the problem reduces to verify that every linear combination $\bm{c}^\top\bar{\bm{\xi}^*}$, with $\bm{c}\in\mathbb{R}^d$, converges to the univariate normal distribution
\begin{align}
\label{eq:app:linearcomb}
    \sqrt{T_n+n}\bm{c}^\top\bar{\bm{\xi}^*}\xrightarrow[n]{d}\mathcal{N}\left(0, \bm{c}^\top \bm{V}_{{\mathcal{P}}}\bm{c}/(1-\alpha) + \bm{c}^\top \bm{J} \bm{c}/\alpha\right).
\end{align}
Otherwise stated, if we show that (\ref{eq:app:linearcomb}) holds, then the Cramér-Wold Theorem implies (\ref{eq:app:thesis}) to hold as well.

Let us start by writing $\bm{c}^\top\bar{\bm{\xi}^*} = n^{-1}\sum_i\bm{c}^\top \bm{S}_i^*$. In particular, note that the scalar quantities $\bm{c}^\top \bm{S}_1^*, \dots, \bm{c}^\top \bm{S}_{n}^*$ define a sequence of exchangeable random variables since they are i.i.d. conditioned on the values of $\bm{W}_{1},\dots, \bm{W}_{T_n}$.
In order to apply Theorem 2 in \citet{blum_chernoff_rosenblatt_teicher_1958}, we need to show for $i\ne i'$ with $i,i'=1,\dots, n$, that
\begin{align*}
    (i)\quad &\quad E\left\{(\bm{c}^\top \bm{S}_i^*) (\bm{c}^\top \bm{S}_{i'}^*) \right\} = o\left(\frac{1}{n}\right);\\
    (ii) \quad &\quad\lim_{n}\text{Cov}\left\{(\bm{c}^\top \bm{S}_i^*)^2, (\bm{c}^\top \bm{S}_{i'}^*)^2\right\} = 0;\\
    (iii) \quad &\quad E\left\{\lvert \bm{c}^\top \bm{S}_i^*\rvert^3\right\} = o(\sqrt{n}).
\end{align*}
Since the quantities $\bm{S}_{i}^*$ and $\bm{S}_{i'}^*$ are linearly independent, condition (i) is automatically satisfied. Note, in fact, that \begin{align*}
    E\left\{\bm{S}_{i}^* \bm{S}_{i'}^*\right\} &= \frac{1}{T_n^2\gamma_1^2}\sum_{t=1}^{T_n}\sum_{t'=1}^{T_n}\sum_{k=1}^K\sum_{k'=1}^KE\left\{ W_{i,k,t}W_{i',k',t'}\right\}E\left\{\nabla\ell_k({\bm{\theta}}^*;\bm{Y}_i)\nabla\ell_{k'}({\bm{\theta}}^*;\bm{Y}_{i'})^\top\right\}=0,
\end{align*}
given that $E\left\{\nabla\ell_k({\bm{\theta}}^*;\bm{Y}_i)\nabla\ell_{k'}({\bm{\theta}}^*;\bm{Y}_{i'})^\top\right\}$ is always null as long as $\bm{Y}_i$ and $\bm{Y}_{i'}$ are independent. It follows that $E\left\{(\bm{c}^\top \bm{S}_i^*)(c^\top \bm{S}_{i'}^*)\right\} = \bm{c}^\top E\left\{\bm{S}_{i}^* {\bm{S}_{i'}^*}^\top\right\} \bm{c} = 0$.

Verifying condition (ii) is slightly more involved since
\begin{align*}
    \text{Cov}\left\{ \left(\bm{c}^\top \bm{S}_i^*\right)^2 , \left(\bm{c}^\top \bm{S}_{i'}^*\right)^2\right\} &= E_{\bm{W}} \text{Cov}_{\bm{Y}|\bm{W}}\left\{\left(\bm{c}^\top \bm{S}_i^*\right)^2 , \left(\bm{c}^\top \bm{S}_{i'}^*\right)^2\right\}  +\\
    &+\text{Cov}_{\bm{W}}\left[E_{\bm{Y}|\bm{W}}\left\{\left(\bm{c}^\top \bm{S}_i^*\right)^2\right\}, E_{\bm{Y}|\bm{W}}\left\{\left(\bm{c}^\top \bm{S}_{i'}^*\right)^2\right\}\right].
\end{align*}
The first term on the right-hand-side is null, since it holds that
\begin{align*}
    E_{\bm{Y}|\bm{W}}\left\{\bm{c}^\top \left(\bm{S}_{i}^*\right) \left(\bm{S}_{i}^*\right)^\top \bm{c} \bm{c}^\top \left(\bm{S}_{i'}^*\right)\left(\bm{S}_{i'}^*\right)^\top \bm{c}\right\} = \bm{c}^\top E_{\bm{Y}|\bm{W}}\left\{\left(\bm{S}_{i}^*\right) \left(\bm{S}_{i}^*\right)^\top \right\}\bm{c} \bm{c}^\top E_{\bm{Y}|\bm{W}}\left\{\left(\bm{S}_{i'}^*\right)\left(\bm{S}_{i'}^*\right)^\top\right\} \bm{c},
\end{align*}
which implies $\text{Cov}_{\bm{Y}|\bm{W}}\left\{\left(\bm{c}^\top \bm{S}_i^*\right)^2 , \left(\bm{c}^\top \bm{S}_{i'}^*\right)^2\right\}=0$. The second term instead is non-null, and to investigate it further, we need to elicit its sum structure. Namely
\begin{align*}
    &\text{Cov}_{\bm{W}}\left[E_{\bm{Y}|\bm{W}}\left\{\left(\bm{c}^\top S_i^*\right)^2\right\}, E_{\bm{Y}|\bm{W}}\left\{\left(\bm{c}^\top \bm{S}_{i'}^*\right)^2\right\}\right]=\\
    &\quad = \frac{1}{T_n^4\gamma_1^4}\sum_{k,k'}\sum_{k'',k'''} \bm{c}^\top E\left\{\nabla\ell_k({\bm{\theta}}^*;\bm{Y})\nabla\ell_{k'}({\bm{\theta}}^*;\bm{Y})^\top\right\}\bm{c}\bm{c}^\top E\left\{\nabla\ell_{k''}({\bm{\theta}}^*;\bm{Y})\nabla\ell_{k'''}({\bm{\theta}}^*;\bm{Y})^\top\right\}\bm{c}\times\\
    &\times \sum_{t,t'}\sum_{t'',t'''} E\left\{ W_{i,k,t}W_{i,k',t'}W_{i',k'',t''}W_{i',k''',t'''}\right\} - E\left\{ W_{i,k,t}W_{i,k',t'}\right\}E\left\{W_{i',k'',t''}W_{i',k''',t'''}\right\} \leq\\
    & \leq \frac{C}{T_n^4\gamma_1^4}\sum_{k,k'}\sum_{k'',k'''}  \sum_{t,t'}\sum_{t'',t'''}\text{Cov}\left\{W_{i,k,t}W_{i,k',t'}, W_{i',k'',t''}W_{i',k''',t'''}\right\}.
\end{align*}
where the inequality holds by Assumption~2 and the sum goes to zero by Lemma~\ref{lemma:w4}, completing the verification of condition (ii). 

Finally, to verify condition (iii), we proceed by noting that by combining Cauchy-Schwarz and Lyapunov inequalities (see, for example, \citealp{Pinelis2015}), for some $C$ large enough, we have
\begin{align*}
    E\left\{\lvert \bm{c}^\top\bm{S}_i^*\rvert^3\right\}\leq \lVert \bm{c}
\rVert^3 E\left\{\lVert\bm{S}_i^*\rVert^3\right\} \leq C E\left\{\lVert\bm{S}_i^*\rVert^3\right\} \leq C E\left\{\lVert\bm{S}_i^*\rVert^4\right\}^{3/4}.
\end{align*}
Furthermore, Assumption~2 ensures that the expectation of the cross product of four partial derivatives of single log sub-likelihood objects is finite, which implies
\begin{align*}
    E\left\{\lVert\bm{S}_i^*\rVert^4\right\} &= \frac{1}{T_n^4\gamma_1^4}\sum_{t,t'}\sum_{l,l'}\sum_{k,k'}\sum_{h,h'}E\left(W_{ikt}W_{ik't'}W_{ihl}W_{ih'l'}\right)E\left\{\nabla\ell_k(\bm{\theta}^*;\bm{Y})^\top\nabla\ell_{k'}(\bm{\theta}^*;\bm{Y})\nabla\ell_{h}(\bm{\theta}^*;\bm{Y})^\top\nabla\ell_{h'}(\bm{\theta}^*;\bm{Y})\right\}\\
    &\leq \frac{C}{T_n^4\gamma_1^4}\sum_{t,t'}\sum_{l,l'}\sum_{k,k'}\sum_{h,h'}E\left(W_{ikt}W_{ik't'}W_{ihl}W_{ih'l'}\right),
\end{align*}
for some $C$ large enough. Therefore, for condition (iii) to hold, we need to ensure that 
\begin{align*}
    \lim_{n\xrightarrow{}\infty}\frac{1}{\sqrt{n}}E\left\{\frac{1}{T_n^4\gamma_1^4}\sum_{t,t'}\sum_{l,l'}\sum_{k,k'}\sum_{h,h'}E\left(W_{ikt}W_{ik't'}W_{ihl}W_{ih'l'}\right)\right\}^{3/4} = 0,
\end{align*}
which is verified by Lemma~\ref{lemma:Ew3}. Thus, with (i), (ii) and (iii) being satisfied, we can apply the Central Limit Theorem for exchangeable processes in \citet{blum_chernoff_rosenblatt_teicher_1958}, which guarantees the asymptotic normality of the scalar quantity $\bm{c}^\top\bar{\bm{\xi}^*} = n^{-1}\sum_i\bm{c}^\top \bm{S}_i^*$, with variance $n^{-1}\bm{c}^\top \text{Var}(\bm{S}_i^*)\bm{c}$. Let us now define  
\begin{align}
\label{eq:Vp}
   \bm{V}_{{\mathcal{P}}} = \lim_{n\xrightarrow{}\infty}\frac{1}{n} E_{{\bm{Y}}}\text{Var}_{\bm{W}|\bm{Y}}(S_{i,t}^*) = \lim_{n\xrightarrow{}\infty}\frac{1}{\gamma_1^{2}n}(\gamma_1-\gamma_2)\bm{H} +\frac{1}{n}\left(\gamma_1^{-2}\gamma_2-1\right)\bm{J},
\end{align}
where $\bm{V}_{{\mathcal{P}}} = O(1)$, since $\lim_n n\gamma_1 >0$ and $\gamma_2/(n\gamma_1^2) = O(1)$ by Assumption~4.
Then, since Regime 3 holds, we need to rescale the variance of $\bm{c}^\top\bar{\bm{\xi}^*}$  by $(T_n + n)$. Hence, we obtain 
\begin{align*}
   \frac{T_n+n}{n}\bm{c}^\top \text{Var}(\bm{S}_i^*)\bm{c} &= \frac{T_n+n}{nT_n}\bm{c}^\top E_{\bm{Y}}\text{Var}_{\bm{W}|\bm{Y}}(\bm{S}_{i,t}^*)\bm{c} + \frac{T_n+n}{n}\bm{c}^\top \bm{J} \bm{c} \\&\xrightarrow[n]{} \frac{1}{1-\alpha}\bm{c}^\top \bm{V}_{{\mathcal{P}}} \bm{c} + \frac{1}{\alpha} \bm{c}^\top \bm{J} \bm{c},
\end{align*}
and therefore $\sqrt{T_n+n}\bm{c}^\top\bar{\bm{\xi}^*}\xrightarrow[n]{d}\mathcal{N}\left(0, \bm{c}^\top \bm{V}_{{\mathcal{P}}}\bm{c}/(1-\alpha) + \bm{c}^\top \bm{J} \bm{c}/\alpha\right)$ as required by the Cramér-Wold Theorem in order for (\ref{eq:app:thesis}) to hold. In Regimes 1 and 2, this last step changes because of the different variance scaling.  While the rest of the proof remains unchanged under Regime 1, note that, under Regime 2, Lemmata~\ref{lemma:Ew3} and \ref{lemma:w4} require the number of iterations to grow such that $n^{7/9}=o(T_n)$ if $\gamma_1=\Theta(1/n)$ as in the cases of $\mathcal{P}_1,\mathcal{P}_2,\mathcal{P}_3$.

Thus, following the asymptotic equivalences outlined in Proposition~2, the efficiency of $\bar{\bm{\Delta}}_{T_n}$ depends on the variability of $\bm{H}^{-1}\bar{\bm{\xi}^*}$. By ignoring the negligible terms in $\text{Var}(\bar{\bm{\xi}^*})$ according to the three asymptotic regimes, it then holds that
\begin{itemize}
\item Regime 1: If $\alpha = 0$, then $n\text{Var}(\bar {\bm{\xi}^*})\xrightarrow[n]{}\bm{J}$. Hence, 
$\sqrt{n}\bar{\bm{\Delta}}_{T_n}\xrightarrow[n]{}\mathcal{N}(0,\bm{H}^{-1}\bm{J}\bm{H}^{-1})$.
\item Regime 2: If $\alpha = 1$, then $T_n\text{Var}(\bar{\bm{ \xi}^*})\xrightarrow[n]{}\bm{V}_{{\mathcal{P}}}$. Hence, $\sqrt{T_n}\bar{\bm{\Delta}}_{T_n}\xrightarrow[n]{}\mathcal{N}(0, \bm{H}^{-1}\bm{V}_{{\mathcal{P}}}\bm{H}^{-1}) ;$
\item Regime 3: If $0<\alpha <1$, then $(T_n+n)\text{Var}(\bar{\bm{\xi}^*})\xrightarrow[n]{}\bm{V}_{{\mathcal{P}}}/(1-\alpha) + \bm{J}/\alpha$. Hence it holds that $\sqrt{T_n+n}\bar{\bm{\Delta}}_{T_n}\xrightarrow[n]{}\mathcal{N}(0,\bm{H}^{-1}\bm{V}_{{\mathcal{P}}}\bm{H}^{-1}/(1-\alpha) + \bm{H}^{-1}\bm{J}\bm{H}^{-1}/\alpha )$.
\end{itemize}
which completes the proof of Theorem 1.
\subsection{Role of the sampling schemes}
Let us start by verifying the conditions of Assumption~4 for $\mathcal{P}_1$, $\mathcal{P}_2$, $\mathcal{P}_3$. Since all three definitions describe a system of binary weights, the following calculations often take advantage of the fact that $W_{ik}^r=W_{ik}$ for any positive real number $r$ and that $W_{ik}W_{jh}\leq W_{ik}$ for all $i,j=1,\dots,n$ and $k,h=1,\dots,K$.

\textit{Verification of Assumption~4 for $\mathcal{P}_1$:}
Consider $\mathcal{P}_1$ defined as in Definition~1, where $W_{ik}=W_{ih}=W_i$ by construction for $i=1,\dots,n$ and $k,h=1,\dots,K$. Thus, it is easy to show that 
\begin{align*}
    \gamma_1&=E\left\{W_{ik}\right\} = E\left\{W_{i}\right\} =\frac{1}{n}&\text{for all $i=1,\dots,n$;}\\
    \gamma_2&=E\left\{W_{ik}W_{ih}\right\}=E\left\{W_{i}^2\right\}=\frac{1}{n}&\text{for all $i=1,\dots,n$ and $k\neq h$;}
\end{align*}
and
\begin{align*}
    E\{W_{ik}W_{jh}\}&=E\{W_{i} W_{j}\}=0;\\
    E\{W_{ik}W_{ik'}W_{jh}\}&=E\{W_{i}^2 W_{j}\}=E\{W_{i} W_{j}\}=0;\\
    E\{W_{ik}W_{ik'}W_{jh}W_{jh'}\}&=E\{W_{i}^2 W_{j}^2\}=E\{W_{i} W_{j}\}=0;\\
    E\{W_{ik}W_{ik'}W_{ih}\}&=E\{W_{i}^3\}=E\{W_{i}\}=\frac{1}{n};\\
    E\{W_{ik}W_{ik'}W_{ih}W_{ih'}\}&=E\{W_{i}^4\}=E\{W_{i}\}=\frac{1}{n};
\end{align*}
which satisfies Assumption~4.

\textit{Verification of Assumption~4 for $\mathcal{P}_2$:}
Consider now $\mathcal{P}_2$ defined as in Definition~2, where all $W_{ik},W_{jh}$ are independent by construction for all $i\neq j$, or $k\neq h$. Then, it is straightforward to write
\begin{align*}
    \gamma_1&=E\left\{W_{ik}\right\} = \frac{1}{n}&\text{for all $i=1,\dots,n$;}\\
    \gamma_2&=E\left\{W_{ik}W_{ih}\right\} = \frac{1}{n^2}&\text{for all $i=1,\dots,n$ and $k\neq h$;}
\end{align*}
and
\begin{align*}
    E\{W_{ik}W_{jh}\}&=\frac{1}{n^2}\\
    E\{W_{ik}W_{ik'}W_{jh}\}&\leq E\{W_{ik}^2W_{jh}\}=E\{W_{ik}W_{jh}\}=\frac{1}{n^2}\\
    E\{W_{ik}W_{ik'}W_{jh}W_{jh'}\}&\leq E\{W_{ik}^2W_{jh}^2\}=E\{W_{ik}W_{jh}\}=\frac{1}{n^2}\\
    E\{W_{ik}W_{ik'}W_{ih}\} &\leq E\{W_{ik}^3\}= E\{W_{ik}\} = \frac{1}{n}\\
    E\{W_{ik}W_{ik'}W_{ih}W_{ih'}\}&\leq E\{W_{ik}^4\}= E\{W_{ik}\} = \frac{1}{n}\\
\end{align*}
which satisfies Assumption~4.

\textit{Verification of Assumption~4 for $\mathcal{P}_3$:} Finally, consider $\mathcal{P}_3$ defined as in Definition~3. The draw without replacement of $K$ elements out of $nK$ implies
\begin{align*}
    \gamma_1&=E\left\{W_{ik}\right\} = \frac{1}{n}&\text{for all $i=1,\dots,n$;}\\
    \gamma_2&=E\left\{W_{ik}W_{ih}\right\} = \frac{1}{n^2}\left(1-\frac{n-1}{nK-1}\right)&\text{for all $i=1,\dots,n$ and $k\neq h$;}
\end{align*}
and 
\begin{align*}
    E\{W_{ik}W_{jh}\}&=\frac{1}{n^2}\left(1-\frac{n-1}{nK-1}\right);\\
    E\{W_{ik}W_{ik'}W_{jh}\}&\leq E\{W_{ik}W_{jh}\}=\frac{1}{n^2}\left(1-\frac{n-1}{nK-1}\right);\\
    E\{W_{ik}W_{ik'}W_{jh}W_{jh'}\}&\leq E\{W_{ik}W_{jh}\}=\frac{1}{n^2}\left(1-\frac{n-1}{nK-1}\right);\\
    E\{W_{ik}W_{ik'}W_{ih}\} &\leq E\{W_{ik}\} = \frac{1}{n};\\
    E\{W_{ik}W_{ik'}W_{ih}W_{ih'}\}&\leq E\{W_{ik}\} = \frac{1}{n};\\
\end{align*}
which satisfies Assumption~4.
\\

The sampling schemes affect the asymptotic covariance matrix of $\bar{\bm{\theta}}_\mathcal{P}$ via $\bm{V}_{{\mathcal{P}}}$, as outlined in Corollary 1. In the following, we elicit the implications of choosing $\mathcal{P}_1$, $\mathcal{P}_2$ or $\mathcal{P}_3$.

\textit{Proof of Corollary~1:}
Recall that by (\ref{eq:Vp}) we have 
\begin{align*}
   \bm{V}_{{\mathcal{P}}} = \lim_{n\xrightarrow{}\infty}\frac{1}{\gamma_1^{2}n}(\gamma_1-\gamma_2)\bm{H} +\frac{1}{n}\left(\gamma_1^{-2}\gamma_2-1\right)\bm{J}.
\end{align*}
Then, the proof of Corollary~1 reduces to plugging in the above display the values $\gamma_1$ and $\gamma_2$ implied by the sampling scheme $\mathcal{P}_1$,$\mathcal{P}_2$ and $\mathcal{P}_3$.

\begin{itemize}
    \item Let $\mathcal{P}_1$ be defined as in Definition~1. Then $\gamma_1=\gamma_2=1/n$.  It follows that 
    \begin{align*}
        \bm{V}_{\mathcal{P}_1} =\lim_{n\xrightarrow{}\infty}\left(1-\frac{1}{n}\right)\bm{J} = \bm{J}.
    \end{align*}
    \item Let $\mathcal{P}_2$ be defined as in Definition~2. Then $\gamma_1=1/n$ and $\gamma_2 = 1/n^2$.  It follows that 
    \begin{align*}
        \bm{V}_{\mathcal{P}_2} =\lim_{n\xrightarrow{}\infty}\left(1-\frac{1}{n}\right)\bm{H} = \bm{H}.
    \end{align*}
    \item Let $\mathcal{P}_3$ be defined as in Definition~3. Then $\gamma_1=1/n$ and 
    $$ \gamma_2 = \frac{1}{n^2}\left(1-\frac{n-1}{nK-1}\right).$$
    It follows that 
    \begin{align*}
        \bm{V}_{\mathcal{P}_3} =\lim_{n\xrightarrow{}\infty}\left(1-\frac{1}{n}\right)\bm{H} - \left(\frac{n-1}{n^2K-n}\right)(\bm{J} - \bm{H})= \bm{H}.
    \end{align*}
\end{itemize}

\newpage
\section{Supplementary simulation experiments}
\label{app:sims}
Implementation-wise, Algorithm 1 runs via custom \texttt{Rcpp} \citep{rcpp} code on both examples. The numerical approximation of ${\hat{\bm{\theta}}}$, instead, relies on the quasi-Newton BFGS provided by the \texttt{ucminf} function \citep{ucminf}, running with custom \texttt{Rcpp} implementations of both $\cln$ and $\nabla \cln$. Computational times are collected via the \texttt{RcppClock} package \citep{rcppclock}.

In principle, a different stepsize for each combination of $\mathcal{P}$, $n$ and $p$ should be possible. However, for a fair comparison among different settings, in the paper, we only report one fixed stepsize for each of the two models examined. In these additional experiments, we show how the choice of the stepsize influences the results in practice.
In general, larger stepsizes imply larger updates and, thus, faster convergence to the target value. However, the variability of the stochastic gradient limits the possibility for arbitrarily large values of $\eta_0$, such that the higher the variance of $\bm{S}_t$, the lower the largest allowed value of $\eta_0$. It follows that an additional advantage of $\mathcal{P}_2$ and $\mathcal{P}_3$ is to accept higher values of the stepsize compared to $\mathcal{P}_1$. While this effect has not been investigated in the main simulation section, it will be particularly evident in the additional experiments for the gamma frailty model reported below. 

Typically, a practical strategy to tune the stepsize requires starting from large values and progressively reducing them until some performance criterion stops improving. In this way, one tries to choose the largest value possible that avoids divergent trajectories but still allows for reasonably large steps. 
In the following, we replicate this kind of strategy by showing simulation results for progressively halved stepsizes.
The results reported in the paper coincide with the ones minimising the mean square error of \texttt{standard} (or its recycled counterpart) in the most challenging setting considered.

\subsection{Ising model}
Data are generated using the \texttt{IsingSampler} package \citep{IsingSampler}. In the following, we report the simulation results accounting for $n\in\{2500, 5000, 10000\}$, $p\in\{10,20\}$, $\eta_0\in\{0.25, 0.5, 1, 2, 4\}$ $B= 0.25n$ and $T_n \in\{0.5n, n, 1.5n, 2n, 2.5n, 3n\}$. While results for $\eta_0 = 1$ are investigated in the main simulation section, here we highlight what happens when the stepsize is too small or too large. 

Figure \ref{fig:isi:app:mse} shows that by choosing $\eta_0>1$, methods based on $\mathcal{P}_1$ occasionally exhibit some initial high spikes of the mean square error during the burn-in period but then continue their trajectories regularly. However, the mean square error increases for all methods when compared to $\eta_0 = 1$. With $\eta_0<1$, convergence is more regular but slower. Intercepts, in particular, struggle the most to approach the numerical estimator performance. 
Similarly, Figure~\ref{fig:isi:app:stepsize} reports the behaviour of the overall mean square error as a function of the stepsize for the estimators considered to highlight how the ``optimal" value of $\eta_0$ changes with $n$, $p$, $\mathcal{P}$ and $T_n$. In most cases, $\eta_0 = 1$ would be picked, but $\eta_0 = 0.5$ would result in a slightly better performance when $p = 10$ and $T_n = 3n$. Nevertheless, the main insights would remain the same.

To deepen the understanding of the convergence under different sampling schemes, Figure~\ref{fig:isi:1:var_tr} reports the trace of the empirical variance across the simulation of each instance of Algorithm 1 when the stepsize is chosen as $\eta_0=1$, as reported in Section 4.1. As expected from the theory, for all choices of $\mathcal{P}$, the variance of $\bar{\bm{{\bm{\theta}}}}_\mathcal{P}$ decreases both with $n$ and $T_n$ diverging. Additionally, as anticipated in Figure~\ref{fig:isi:app:mse}, both $\bar{{\bm{\theta}}}_{\mathcal{P}_2}$ and $\bar{{\bm{\theta}}}_{\mathcal{P}_3}$ exhibit lower variability compared to $\bar{{\bm{\theta}}}_{\mathcal{P}_1}$, with this difference decreasing along the optimization because of being scaled by $T_n$. 

Finally, Figure \ref{fig:isi:app:cov} shows how inference results change according to the stepsize chosen. When the value of $\eta_0$ is too tiny, confidence intervals suffer from the poor estimation of both $\bm{H}$ and $\bm{J}$ since the parameters used in their evaluation are still too far from the target because of slow convergence. When the stepsize is too large, instead, the stochastic estimators have too much freedom in exploring the parameter space. Typically, this leads to observing more variability than expected and conduces to under covering confidence intervals. Thus, Figure \ref{fig:isi:app:cov} highlights that for a reliable inference procedure, the stepsize must be selected carefully, as done in Figure~\ref{fig:isi:app:stepsize}.

To conclude, we report a detailed overview of the computational efficiency (with $\eta_0 = 1$) of the proposed estimators, showing also the profiled times for each Step of Algorithm 1. In particular, Table~\ref{tab:ising:timeapp} highlights how the Sampling Step avoids $\mathcal{P}_2$ and $\mathcal{P}_3$ to scale efficiently with an increasing number of components to consider. However, recycling the sampling step across iterations easily solves such a problem for $\mathcal{P}_3$.
\begin{figure}[!h]\centering
	\includegraphics[width=\textwidth]{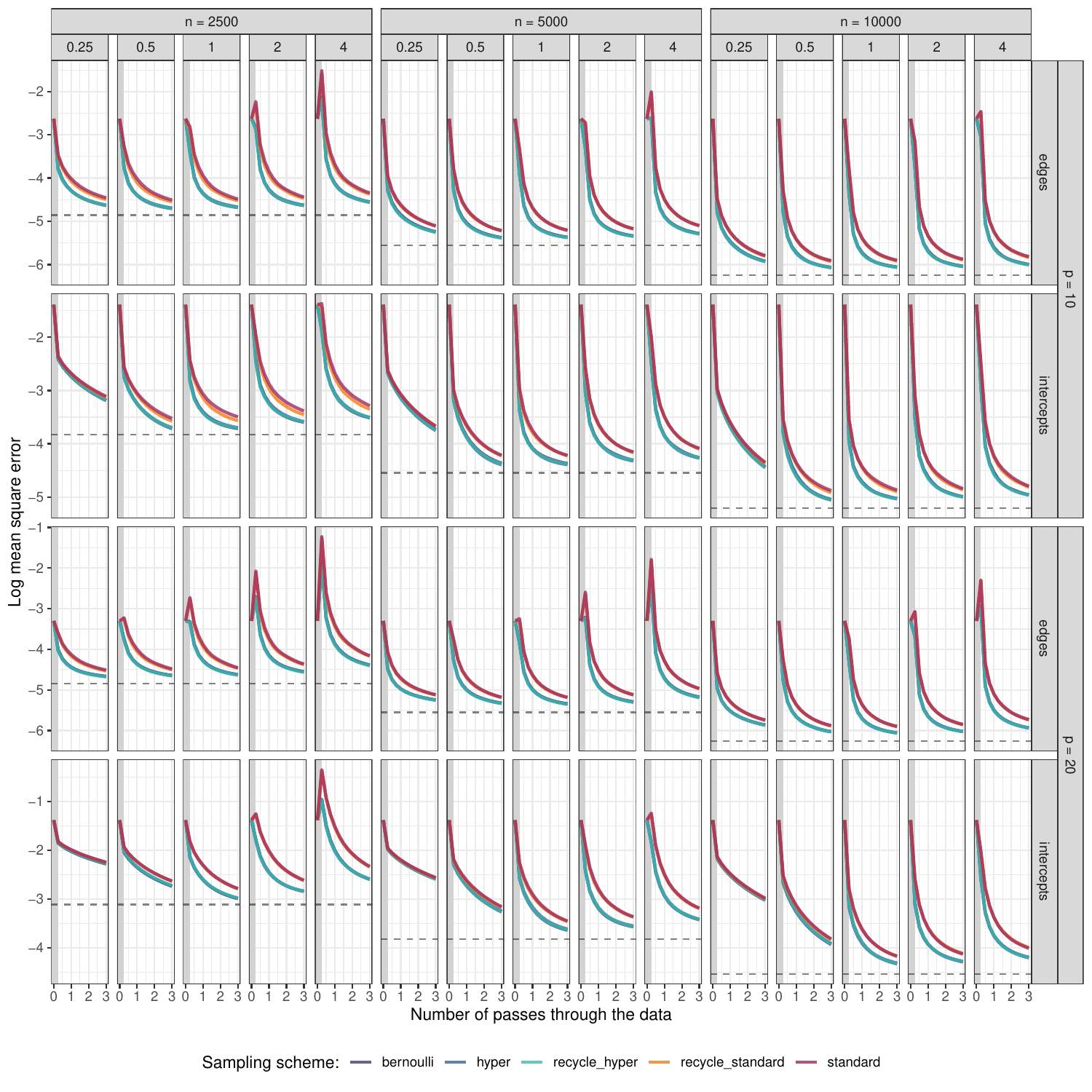}
	\caption{\label{fig:isi:app:mse}Ising model, $\eta_0\in\{0.25, 0.50, 1, 2, 4\}$. Log mean square error trajectories along the optimisation, grouped by sample size, stepsize, parameter type and number of nodes. Coloured lines refer to ${\bar{\bm{\theta}}}_{\mathcal{P}}$ under different sampling schemes, while dashed lines indicate the performance of the numerical approximation of ${\hat{\bm{\theta}}}$.}
\end{figure}

\begin{figure}[t]\centering
	\includegraphics[width=\textwidth]{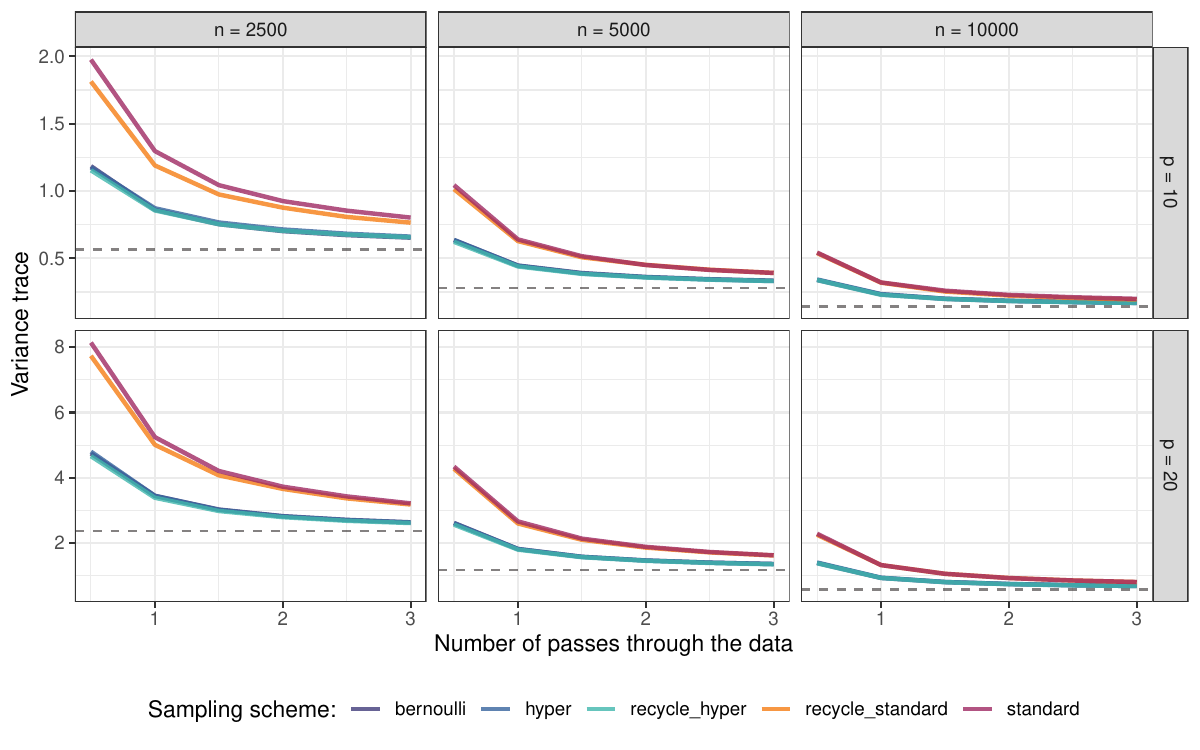}
	\caption{\label{fig:isi:1:var_tr}Ising model. Trajectories of the empirical variance trace  for different $n$ and $p$. Solid lines refer to $\bar{\bm{{\bm{\theta}}}}_\mathcal{P}$ under different sampling schemes. Dashed lines denote the variability observed for the numerical approximation of $\hat{\bm{{\bm{\theta}}}}$.}
\end{figure}
\begin{figure}[!h]\centering
	\includegraphics[width=\textwidth]{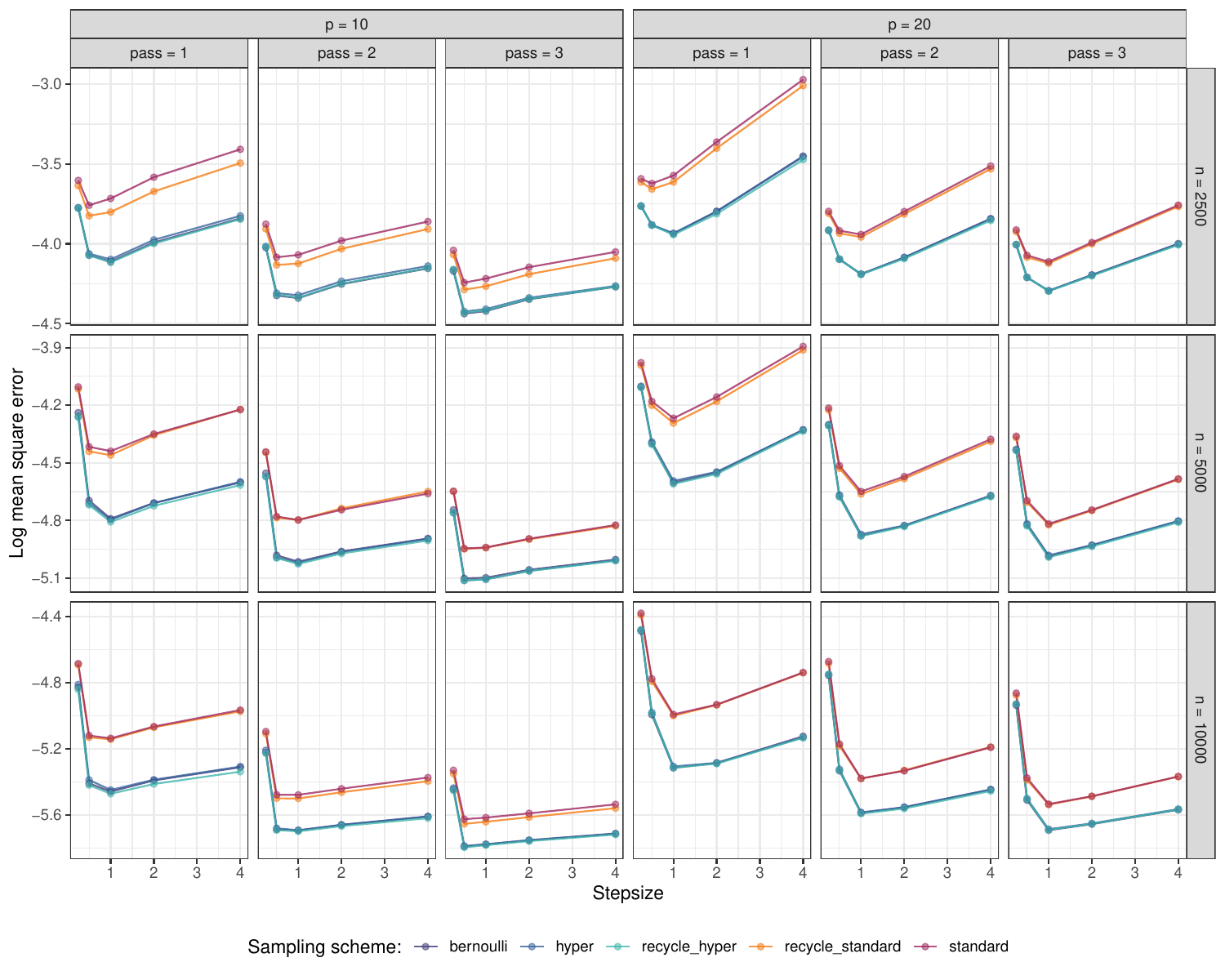}
	\caption{\label{fig:isi:app:stepsize} Ising model. Overall log mean square error as a function of the stepsize, grouped by $n$, $p$, and the number of passes through the data. Coloured lines refer to ${\bar{\bm{\theta}}}_{\mathcal{P}}$ under different sampling schemes.}
\end{figure}

\begin{figure}[!h]\centering
	\includegraphics[width=\textwidth]{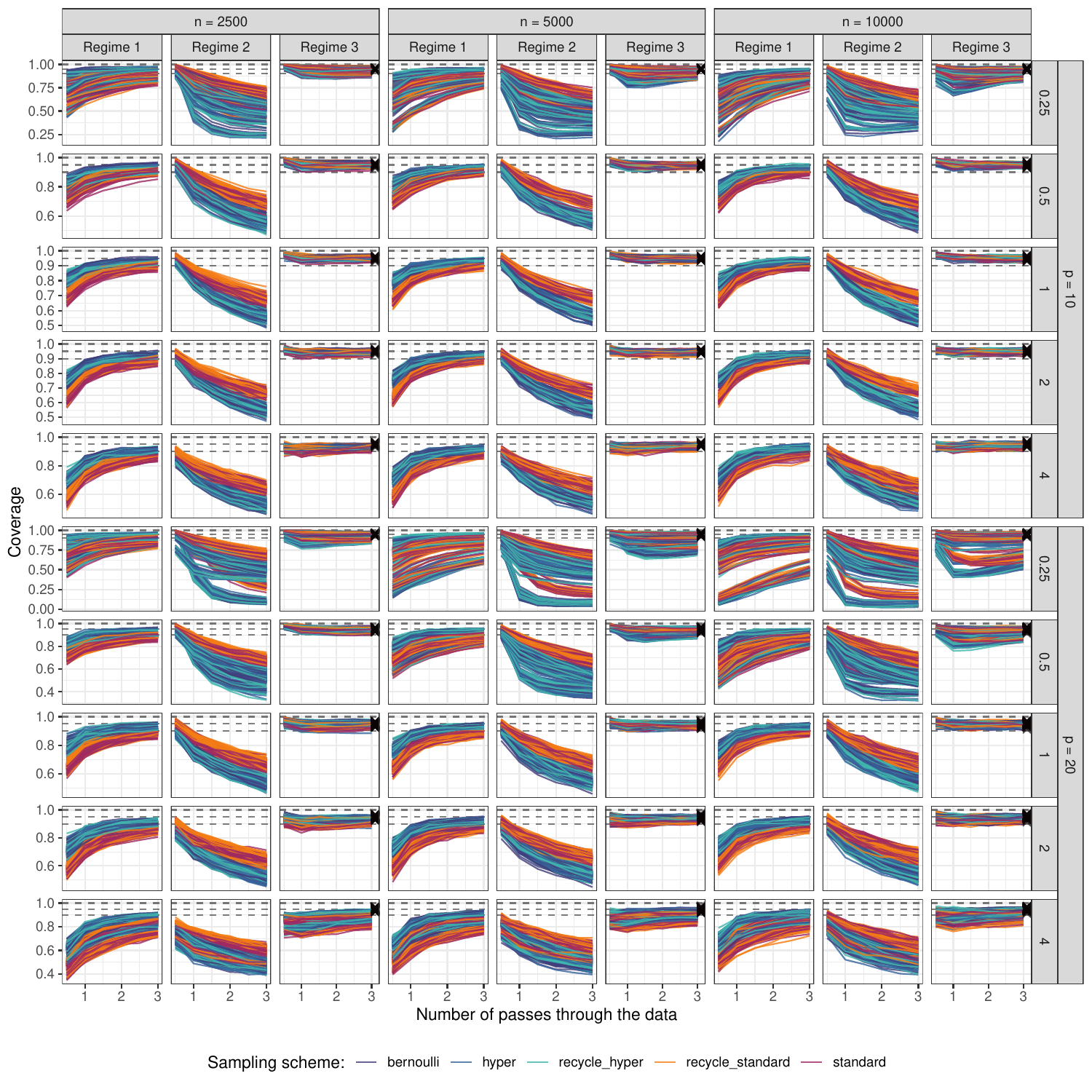}
	\caption{\label{fig:isi:app:cov} Ising model, $\eta_0\in\{0.25, 0.50, 1, 2, 4\}$. Empirical coverage of confidence intervals for ${\bar{\bm{\theta}}}_{\mathcal{P}}$ constructed according to the three asymptotic regimes in Theoem 1. Results are grouped by sample size and graph dimension. Dashed lines highlight the nominal coverage level. Coloured lines refer to scalar elements of ${\bar{\bm{\theta}}}_{\mathcal{P}}$ under different sampling schemes. Dark crosses equivalently refer to scalar elements of the numerical approximation of ${\hat{\bm{\theta}}}$.}
\end{figure}

\begin{table}[h!]
\centering
\caption{\label{tab:ising:timeapp}Comparison of computational times (s) for ${\bar{\bm{\theta}}}_{\mathcal{P}}$ with $T_n=3n$ under different sampling schemes and the \texttt{numerical} approximation of ${\hat{\bm{\theta}}}$. }
\tiny
\begin{tabular}{cclcccccc}
  \toprule
 $p$ & $n$ &  & \texttt{bernoulli} & \texttt{hyper} & \texttt{recycle\_hyper} & \texttt{recycle\_standard} & \texttt{standard} & \texttt{numerical} \\ 
  \midrule
\multirow{16}{*}{$10$}&\multirow{4}{*}{2500}& Total & $1.81$ & $1.70$  & $8.12 \times 10^{-2}$ & $6.66 \times 10^{-2}$ & $2.48 \times 10^{-1}$ & $2.46$ \\ 
\cmidrule(lr){3-9}
  && Sampling Step & $2.33 \times 10^{-4}$ & $2.17 \times 10^{-4}$ & $2.29 \times 10^{-6}$ & $3.93 \times 10^{-7}$ & $2.46 \times 10^{-5}$ &  \\ 
  & & Approximation Step & $5.99 \times 10^{-6}$ & $6.47 \times 10^{-6}$ & $6.18 \times 10^{-6}$ & $6.11 \times 10^{-6}$ & $6.13 \times 10^{-6}$ &  \\ 
  & & Update Step & $1.65 \times 10^{-7}$ & $1.77 \times 10^{-7}$ & $1.68 \times 10^{-7}$ & $1.66 \times 10^{-7}$ & $1.62 \times 10^{-7}$ &  \\ 
\cmidrule(lr){2-9}
  &\multirow{4}{*}{5000} & Total & $7.16$ & $7.48$ & $2.01 \times 10^{-1}$ & $1.35 \times 10^{-1}$ & $8.03 \times 10^{-1}$ & $4.99$ \\ 
  \cmidrule(lr){3-9}
  & & Sampling Step & $4.69 \times 10^{-4}$ & $4.89 \times 10^{-4}$ & $4.99 \times 10^{-6}$ & $5.96 \times 10^{-7}$ & $4.50 \times 10^{-5}$ &  \\ 
  & & Approximation Step & $6.21 \times 10^{-6}$ & $6.94 \times 10^{-6}$ & $6.15 \times 10^{-6}$ & $6.12 \times 10^{-6}$ & $6.18 \times 10^{-6}$ &  \\ 
  & & Update Step & $1.72 \times 10^{-7}$ & $2.09 \times 10^{-7}$ & $1.72 \times 10^{-7}$ & $1.74 \times 10^{-7}$ & $1.77 \times 10^{-7}$ &  \\ 
\cmidrule(lr){2-9}
  &\multirow{4}{*}{10000} & Total & $28.5$ & $41.9 $ & $6.82 \times 10^{-1}$ & $2.82 \times 10^{-1}$ & $2.88$ & $10.0$ \\ 
    \cmidrule(lr){3-9}
  & & Sampling Step & $9.39 \times 10^{-4}$ & $1.39 \times 10^{-3}$ & $1.42 \times 10^{-5}$ & $1.01 \times 10^{-6}$ & $8.76 \times 10^{-5}$ &  \\ 
  & & Approximation Step & $6.61 \times 10^{-6}$ & $7.89 \times 10^{-6}$ & $6.23 \times 10^{-6}$ & $6.11 \times 10^{-6}$ & $6.26 \times 10^{-6}$ &  \\
  & & Update Step & $1.64 \times 10^{-7}$ & $2.10 \times 10^{-7}$ & $1.64 \times 10^{-7}$ & $1.62 \times 10^{-7}$ & $1.67 \times 10^{-7}$ &  \\ 
  \midrule
  \multirow{16}{*}{$20$}&$\multirow{4}{*}{2500}$ & Total & $3.74$ & $3.90$ & $2.39 \times 10^{-1}$ & $2.02 \times 10^{-1}$ & $3.84 \times 10^{-1}$ & $19.5$ \\ 
        \cmidrule(lr){3-9}
  & & Sampling Step & $4.71 \times 10^{-4}$ & $4.87 \times 10^{-4}$ & $5.00 \times 10^{-6}$ & $4.29 \times 10^{-7}$ & $2.47 \times 10^{-5}$ &  \\ 
  & & Approximation Step & $2.40 \times 10^{-5}$ & $2.63 \times 10^{-5}$ & $2.21 \times 10^{-5}$ & $2.19 \times 10^{-5}$ & $2.19 \times 10^{-5}$ &  \\ 
  & & Update Step & $3.79 \times 10^{-7}$ & $4.28 \times 10^{-7}$ & $3.56 \times 10^{-7}$ & $3.51 \times 10^{-7}$ & $3.43 \times 10^{-7}$ &  \\ 
        \cmidrule(lr){2-9}
  &\multirow{4}{*}{5000} & Total & $14.4$ & $21.2$ & $6.14 \times 10^{-1}$ & $3.98 \times 10^{-1}$ & $1.07$ & $44.4$ \\ 
        \cmidrule(lr){3-9}
  & & Sampling Step & $9.36 \times 10^{-4}$ & $1.38 \times 10^{-3}$ & $1.43 \times 10^{-5}$ & $6.37 \times 10^{-7}$ & $4.56 \times 10^{-5}$ &  \\ 
  & & Approximation Step & $2.20 \times 10^{-5}$ & $2.60 \times 10^{-5}$ & $2.21 \times 10^{-5}$ & $2.15 \times 10^{-5}$ & $2.16 \times 10^{-5}$ &  \\ 
  & & Update Step & $3.33 \times 10^{-7}$ & $4.28 \times 10^{-7}$ & $3.53 \times 10^{-7}$ & $3.46 \times 10^{-7}$ & $3.46 \times 10^{-7}$ &  \\ 
        \cmidrule(lr){2-9}
  &\multirow{4}{*}{10000} & Total & $57.6$ & $90.4$ & $1.73$ & $8.07 \times 10^{-1}$ & $3.40$ & $98.5$ \\ 
        \cmidrule(lr){3-9}
  & & Sampling Step & $1.89 \times 10^{-3}$ & $2.98 \times 10^{-3}$ & $3.03 \times 10^{-5}$ & $1.03 \times 10^{-6}$ & $8.69 \times 10^{-5}$ &  \\ 
  & & Approximation Step & $2.28 \times 10^{-5}$ & $2.69 \times 10^{-5}$ & $2.29 \times 10^{-5}$ & $2.15 \times 10^{-5}$ & $2.20 \times 10^{-5}$ &  \\ 
  & & Update Step & $3.30 \times 10^{-7}$ & $4.17 \times 10^{-7}$ & $3.52 \times 10^{-7}$ & $3.46 \times 10^{-7}$ & $3.52 \times 10^{-7}$ &  \\ 
   \bottomrule
\end{tabular}

\end{table}
\FloatBarrier
\subsection{Gamma frailty model}
We set the true parameters as $\xi = 0.25$ and $\rho = 0.5$, while $\lambda_j$ equals $0.25$ if $j$ is even and $-0.25$ if odd. As in the Ising model experiments, optimization always starts from the null vector, and all the performance criteria are evaluated for the different running lengths of the algorithm. Results in the main paper are reported for $\eta_0=2$, which is the initial stepsize value minimizing the mean square error performance at $T_n = 3n$ of \texttt{recycle\_standard} in the most challenging setting, i.e., $n = 2,500, p = 30$.

Differently to the pairwise likelihood outlined in \citet{henderson2003}, we scale the objective function by $p(p-1)/2$ (i.e. the number of pairs) rather than just $p-1$. This helps to keep the same value of $\eta_0$ across settings with different $p$ since it assures the objective function to be comparable in magnitude across values of $p$.
Implementation-wise, the parameters $\xi$ and $\rho$ are reparametrized to avoid introducing explicit constraints in the optimization. In particular, we use $\check \xi = \exp\{\xi^{-1}\}$ and $\check\rho = \text{artanh}(\rho)$.
 In the following, we report the simulation results accounting for $n\in\{2500, 5000, 10000\}$, $p\in\{10,20\}$, $B= 0.25n$ and $T_n \in\{0.5n, n, 1.5n, 2n, 2.5n, 3n\}$. Similar to previous experiments, we assess results under different choices of the stepsize, i.e. $\eta_0\in\{0.5, 1, 2, 4, 8\}$. As anticipated in the main simulation section, $\mathcal{P}_3$ becomes very computationally intensive when the number of likelihood components increases. For this reason, we only consider the results for \texttt{recycle\_hyper}. However, the choice of $l=500$ in the simulation section is arbitrary, so here we report results for \texttt{recycle\_hyper} with $l\in\{100, 500, 1000\}$, showing that such choice does not have noticeable impacts on the accuracy of the estimation but is crucial for the practical computational efficiency of the estimator.

 Figure~\ref{fig:gf:app:mse} reports the log mean square error of the stochastic estimators considered under different simulation settings. First, as anticipated at the beginning of the section, estimators based on $\mathcal{P}_3$ allow for larger stepsizes without suffering from divergence problems. Considering the reparametrisations of $\xi$ and $\rho$, in fact, \texttt{recycle\_standard} starts experiencing diverging trajectories when $\eta_0>2$, while all the instances of \texttt{recycle\_hyper} remain stable in their converging trajectories. Second, the effects of different choices of $l$ are unnoticeable. To better grasp them, we need to focus on some specific settings. Figure~\ref{fig:gf:app:zoom}, for example, outlines the behaviour of the overall log mean square error for $\eta = 2$ (the value reported in the paper). For visualisation purposes, we report only the trajectories after the burn-in period. All the lines are almost overlapping, and no clear distinction can be identified among the three instances considered. The only noticeable differences appear analysing the computational times, as reported in Table~\ref{tab:gf:timeapp}.

 Similarly to the previous example, Figure~\ref{fig:gf:app:step} plots the overall log mean square error as a function of the stepsize. Clearly, the graphs show how methods based on $\mathcal{P}_3$ are much more stable when increasing the stepsize, while the mean square error of \texttt{recycle\_standard} starts exploding. Figure~\ref{fig:gf:1:var_tr} highlights the trajectories of the observed variance trace instead, showing that the  mean square error differences among $\mathcal{P}_1$ and $\mathcal{P}_3$ observed in Figure~\ref{fig:gf:app:mse} are due to the different variability they induce in the estimates. Trajectories are shown for $\eta_0=2$.
 
 Finally, Figure~\ref{fig:gf:app:cov} show the empirical coverage performance for all the possible values of $\eta_0$. For all the estimators considered, most difficulties arise when the stepsize is too small. In those cases, in fact, the estimates for the reparameterisations of $\xi$ and $\rho$ are still too distant from their targets. Hence, the evaluation of the empirical estimator of $\bm{H}$ and $\bm{J}$ is still not accurate enough.

\begin{figure}[!h]\centering
	\includegraphics[width=\textwidth]{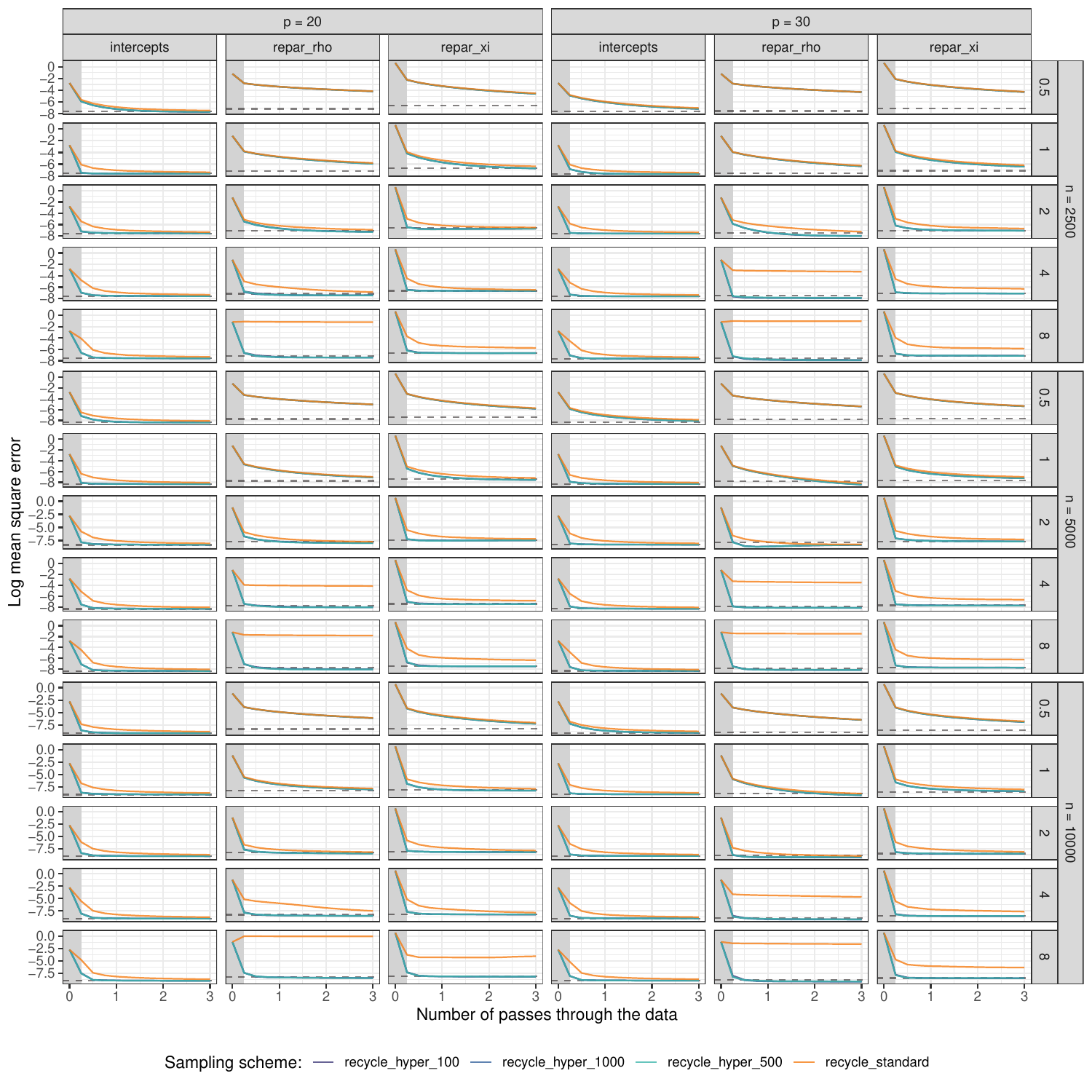}
	\caption{\label{fig:gf:app:mse}Gamma frailty, $\eta_0\in\{0.50, 1, 2, 4, 8\}$. Log mean square error trajectories along the optimisation, grouped by sample size, stepsize and parameter type. Coloured lines refer to ${\bar{\bm{\theta}}}_{\mathcal{P}}$ under different sampling schemes, while the dashed lines indicate the performance of the numerical approximation of ${\hat{\bm{\theta}}}$.}
\end{figure}

\begin{figure}[!h]\centering
	\includegraphics[width=\textwidth]{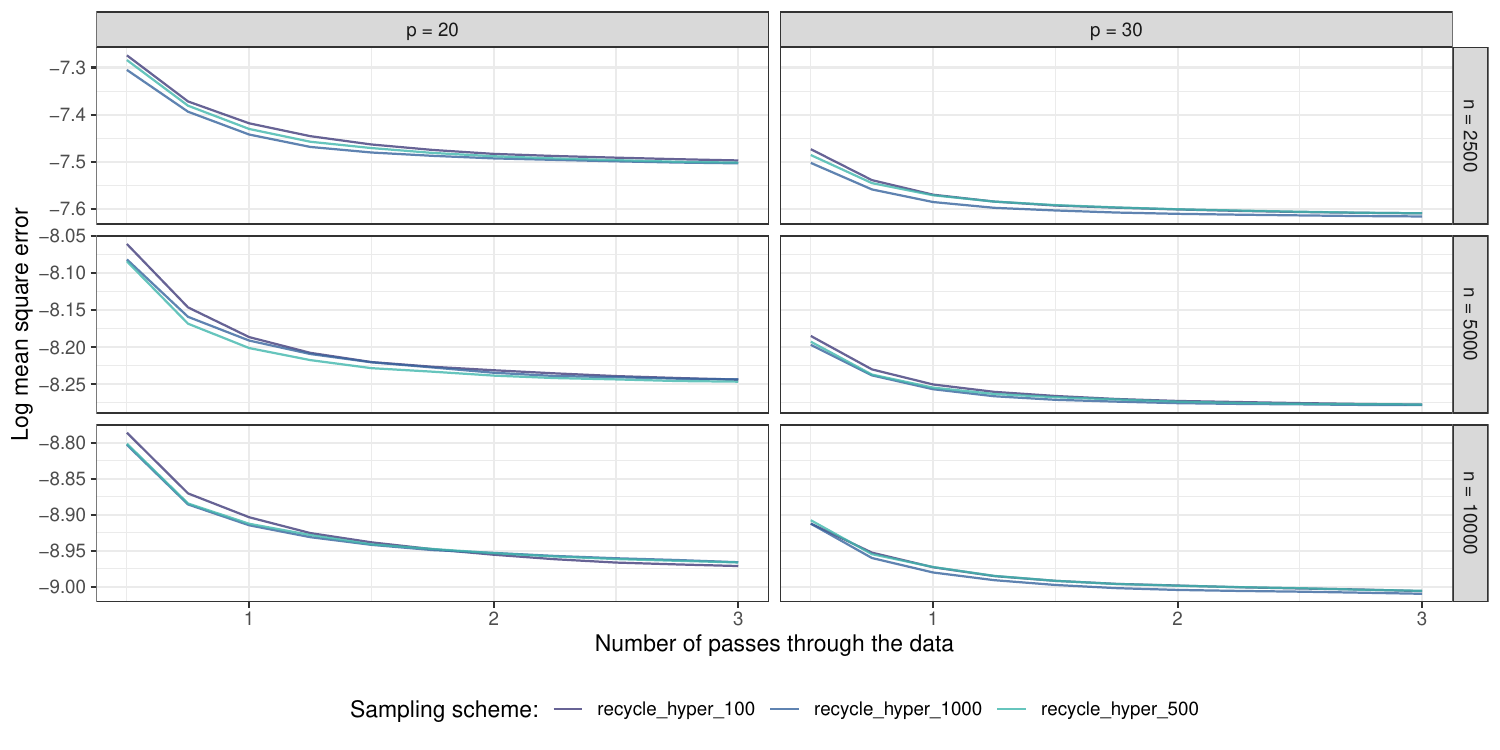}
	\caption{\label{fig:gf:app:zoom}Gamma frailty. Overall log mean square error trajectories along the optimisation for $\eta_0 = 2$, grouped by $n$ and $p$. Coloured lines refer to the different lengths $l$ of the recycling window used for $\mathcal{P}_3$.}
\end{figure}

\begin{figure}[!h]\centering
	\includegraphics[width=\textwidth]{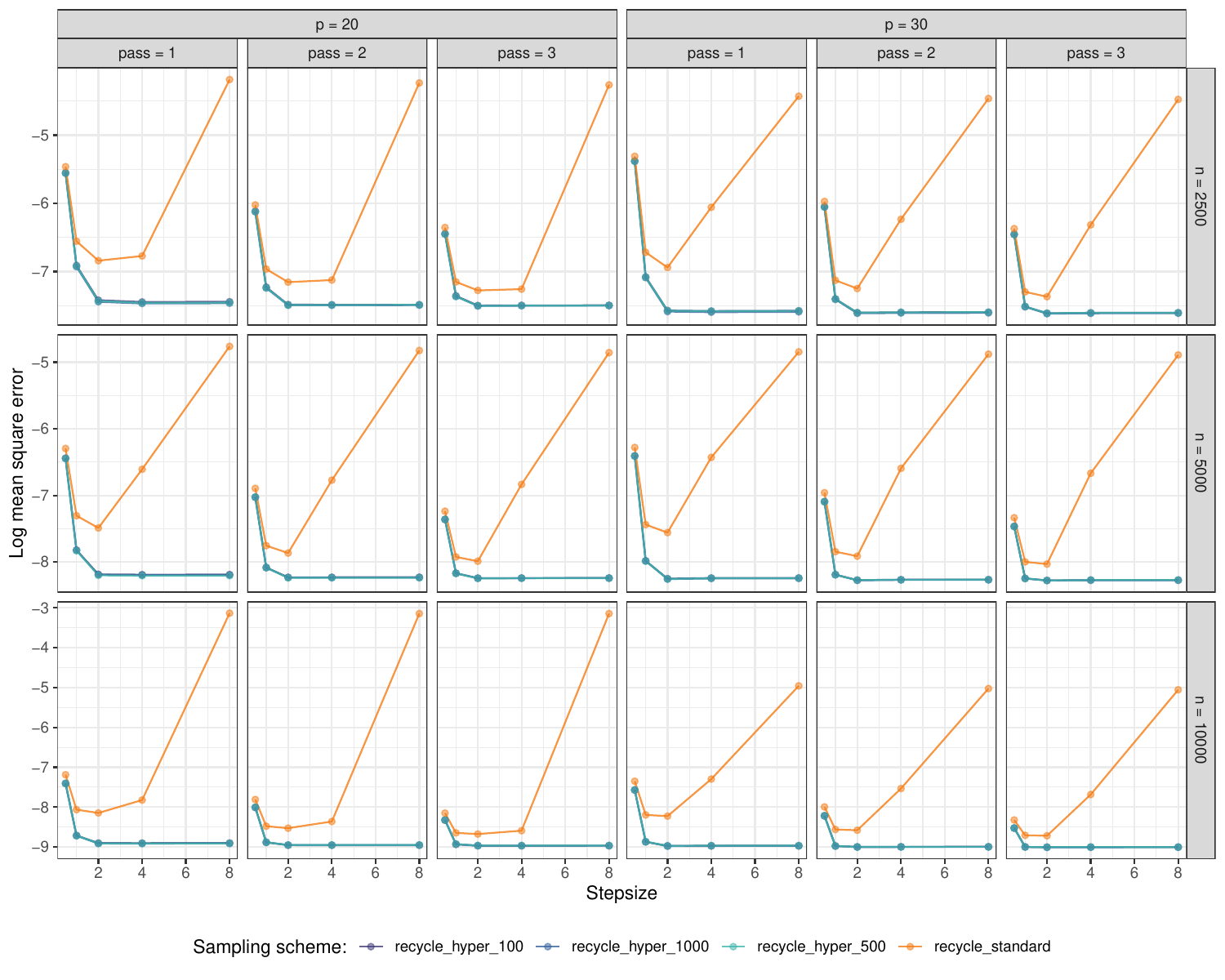}
	\caption{\label{fig:gf:app:step}Gamma frailty, log mean square error as a function of the stepsize, grouped by $n$, $p$, and the number of passes through the data. Coloured lines refer to ${\bar{\bm{\theta}}}_{\mathcal{P}}$ under different sampling schemes.}
\end{figure}

\begin{figure}[t]\centering
	\includegraphics[width=.73\textwidth]{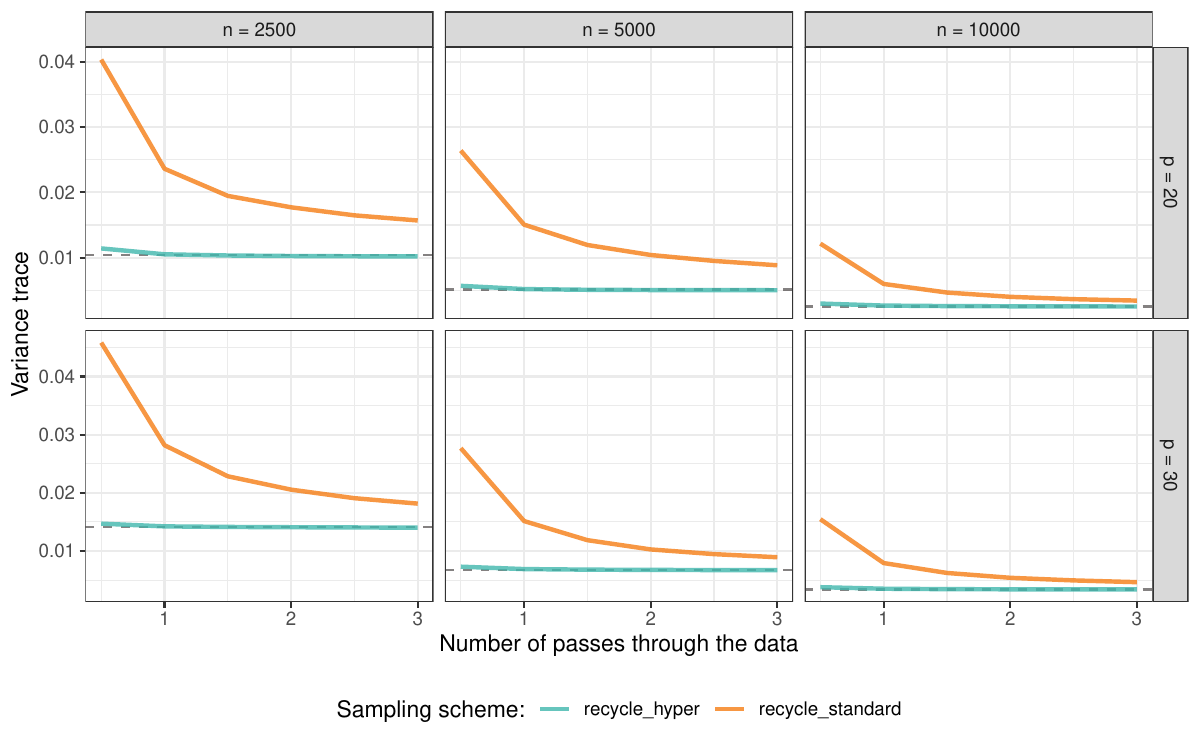}
	\caption{\label{fig:gf:1:var_tr}Gamma frailty model. Trajectories of empirical variance traces for different $n$ and $p$. Solid lines refer to $\bar{\bm{{\bm{\theta}}}}_\mathcal{P}$ under different sampling schemes. Dashed ones indicate values observed for the numerical approximation of $\hat{\bm{{\bm{\theta}}}}$.}
\end{figure}
\begin{figure}[!h]\centering
	\includegraphics[width=\textwidth]{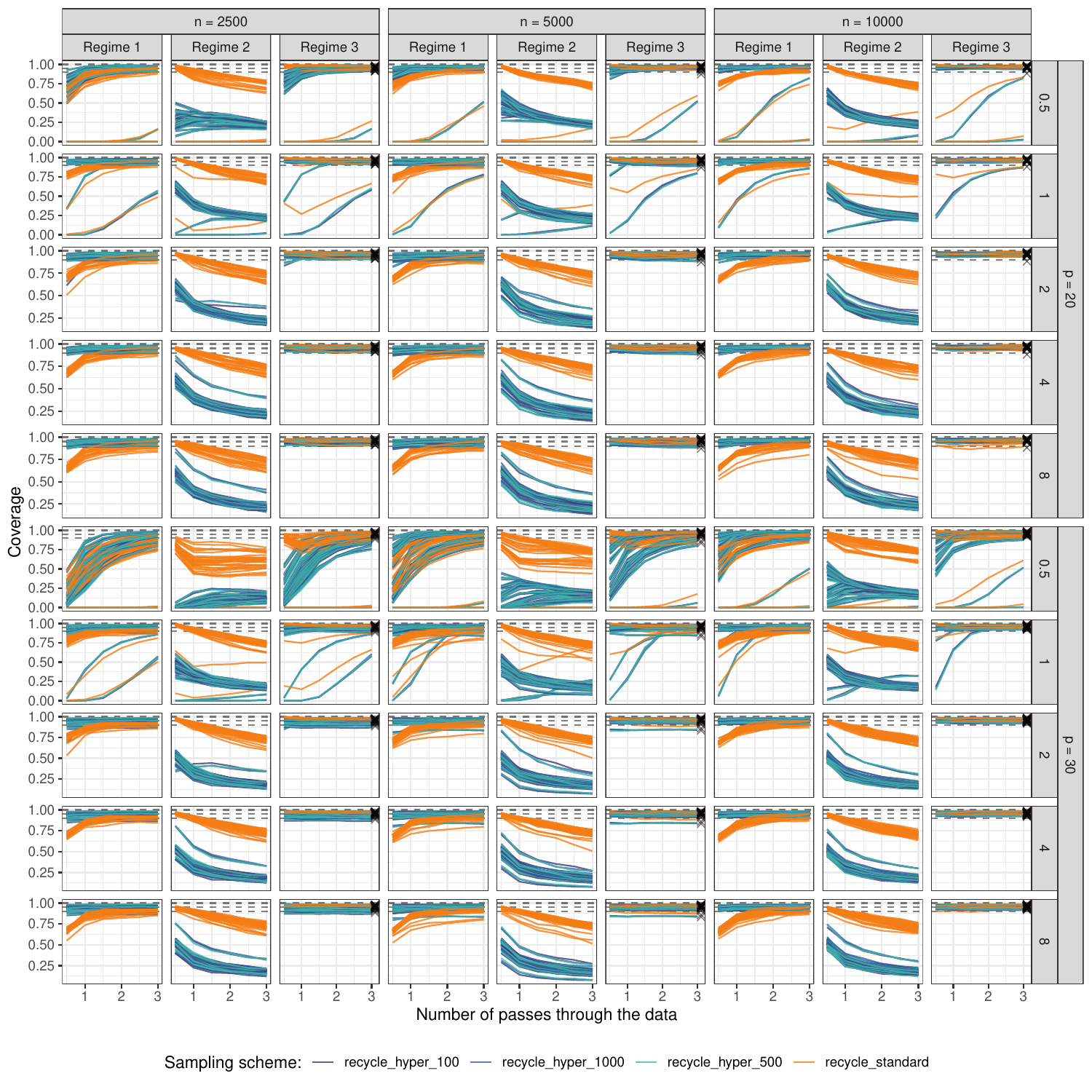}
	\caption{\label{fig:gf:app:cov}Gamma frailty, $\eta_0\in\{0.50, 1, 2, 4, 8\}$. Empirical coverage of confidence intervals for ${\bar{\bm{\theta}}}_{\mathcal{P}}$ constructed according to the three asymptotic regimes in Theorem 1. Results are grouped by sample size. Dashed lines highlight the nominal coverage level. Coloured lines refer to scalar elements of ${\bar{\bm{\theta}}}_{\mathcal{P}}$ under different sampling schemes. Dark crosses equivalently refer to scalar elements of the numerical approximation of ${\hat{\bm{\theta}}}$.}
\end{figure}

\begin{table}[!ht]
\centering
\caption{\label{tab:gf:timeapp}Comparison of computational times (s) for ${\bar{\bm{\theta}}}_{\mathcal{P}}$ with $T_n=3n$ under different sampling schemes and the \texttt{numerical} approximation of ${\hat{\bm{\theta}}}$.}
\tiny
\begin{tabular}{cclccccc}
  \toprule
 $p$ & $n$ & operation & \texttt{recycle\_hyper\_100} & \texttt{recycle\_hyper\_1000} & \texttt{recycle\_hyper\_500 }& \texttt{recycle\_standard} & \texttt{numerical} \\ 
  \midrule
\multirow{16}{*}{20}&\multirow{4}{*}{2500}& Total & $1.62$ & $1.22 $ & $1.25 $ & $8.93 \times 10^{-1}$ & $6.76 $ \\ 
\cmidrule(lr){3-8}
  & & Sampling Step & $1.57 \times 10^{-6}$ & $1.60 \times 10^{-6}$ & $1.55 \times 10^{-6}$ & $4.07 \times 10^{-7}$ &  \\ 
  & & Approximation Step & $1.52 \times 10^{-4}$ & $1.54 \times 10^{-4}$ & $1.53 \times 10^{-4}$ & $1.17 \times 10^{-4}$ &  \\  
  &  & Update Step & $1.31 \times 10^{-6}$ & $1.31 \times 10^{-6}$ & $1.31 \times 10^{-6}$ & $1.17 \times 10^{-7}$ &  \\
    \cmidrule(lr){2-8}
  &\multirow{4}{*}{5000} & Total & $4.66 $ & $2.61 $ & $2.77 $ & $1.80 $ & $13.5 $ \\ 
    \cmidrule(lr){3-8}
  & & Sampling Step & $1.60 \times 10^{-6}$ & $1.68 \times 10^{-6}$ & $1.59 \times 10^{-6}$ & $4.29 \times 10^{-7}$ &  \\ 
  && Approximation Step & $1.55 \times 10^{-4}$ & $1.58 \times 10^{-4}$ & $1.55 \times 10^{-4}$ & $1.18 \times 10^{-4}$ &  \\
  & & Update Step & $1.31 \times 10^{-6}$ & $1.32 \times 10^{-6}$ & $1.31 \times 10^{-6}$ & $1.20 \times 10^{-7}$ &  \\ 
    \cmidrule(lr){2-8}
  &\multirow{4}{*}{10000} & Total & $18.1$ & $6.32 $ & $7.24 $ & $3.64$ & $27.2 $ \\ 
  \cmidrule(lr){3-8}
  & & Sampling Step & $1.68 \times 10^{-6}$ & $1.68 \times 10^{-6}$ & $1.62 \times 10^{-6}$ & $4.48 \times 10^{-7}$ &  \\ 
  & & Approximation Step & $1.64 \times 10^{-4}$ & $1.64 \times 10^{-4}$ & $1.59 \times 10^{-4}$ & $1.19 \times 10^{-4}$ &  \\ 
  & & Update Step & $1.32 \times 10^{-6}$ & $1.31 \times 10^{-6}$ & $1.31 \times 10^{-6}$ & $1.20 \times 10^{-7}$ &  \\  
   \midrule
\multirow{16}{*}{30}&\multirow{4}{*}{2500} & Total & $4.23 $ & $2.95 $ & $3.03 $ & $2.23 $ & $12.0$ \\ 
    \cmidrule(lr){3-8}
  & & Sampling Step & $1.86 \times 10^{-6}$ & $1.85 \times 10^{-6}$ & $1.81 \times 10^{-6}$ & $7.45 \times 10^{-7}$ &  \\ 
  & & Approximation Step & $3.74 \times 10^{-4}$ & $3.72 \times 10^{-4}$ & $3.68 \times 10^{-4}$ & $2.94 \times 10^{-4}$ &  \\ 
  & & Update Step & $1.35 \times 10^{-6}$ & $1.34 \times 10^{-6}$ & $1.33 \times 10^{-6}$ & $1.31 \times 10^{-7}$ &  \\ 
    \cmidrule(lr){2-8}
  &\multirow{4}{*}{5000} & Total & $13.8 $ & $6.31 $ & $7.09 $ & $4.45 $ & $24.4$ \\ 
      \cmidrule(lr){3-8}
  & & Sampling Step & $2.00 \times 10^{-6}$ & $1.83 \times 10^{-6}$ & $1.88 \times 10^{-6}$ & $8.03 \times 10^{-7}$ &  \\ 
  & & Approximation Step & $3.94 \times 10^{-4}$ & $3.73 \times 10^{-4}$ & $3.77 \times 10^{-4}$ & $2.94 \times 10^{-4}$ &  \\
  & & Update Step & $1.35 \times 10^{-6}$ & $1.33 \times 10^{-6}$ & $1.33 \times 10^{-6}$ & $1.31 \times 10^{-7}$ &  \\    
  \cmidrule(lr){2-8}
  &\multirow{4}{*}{10000} & Total &$54.0$ & $15.4$ & $19.3$ & $8.88$ & $48.1 $ \\ 
      \cmidrule(lr){3-8}
  & & Sampling Step & $2.03 \times 10^{-6}$ & $1.86 \times 10^{-6}$ & $1.88 \times 10^{-6}$ & $8.00 \times 10^{-7}$ &  \\ 
  & & Approximation Step & $4.02 \times 10^{-4}$ & $3.86 \times 10^{-4}$ & $3.88 \times 10^{-4}$ & $2.93 \times 10^{-4}$ &  \\
  & & Update Step & $1.35 \times 10^{-6}$ & $1.33 \times 10^{-6}$ & $1.34 \times 10^{-6}$ & $1.32 \times 10^{-7}$ &  \\
     \bottomrule
\end{tabular}

\end{table}

\newpage
\FloatBarrier
\section{Supplementary material NESARC data}
\label{app:mh}
A description of the $p=32$ items considered can be found in Table~\ref{tab:dict}. As described in Section 5, the optimisation stops by controlling the holdout average composite log-likelihood.
Panel (A) in Figure~\ref{fig:mh:checks} shows the estimated edge parameters along the optimization. The stable behavior of most of their trajectories suggests a decreasing utility of additional iterations. Panel (B) instead tracks the holdout negative composite log-likelihood, which exhibits a sharp drop at the beginning of the optimization and then slowly decreases until the procedure ends at $T_n$. Thus, while the algorithm does not strictly minimise the holdout objective function, the decreasing trajectory suggests it is avoiding overfitting. 
\begin{figure}[!ht]\centering
	\includegraphics[width=.85\textwidth]{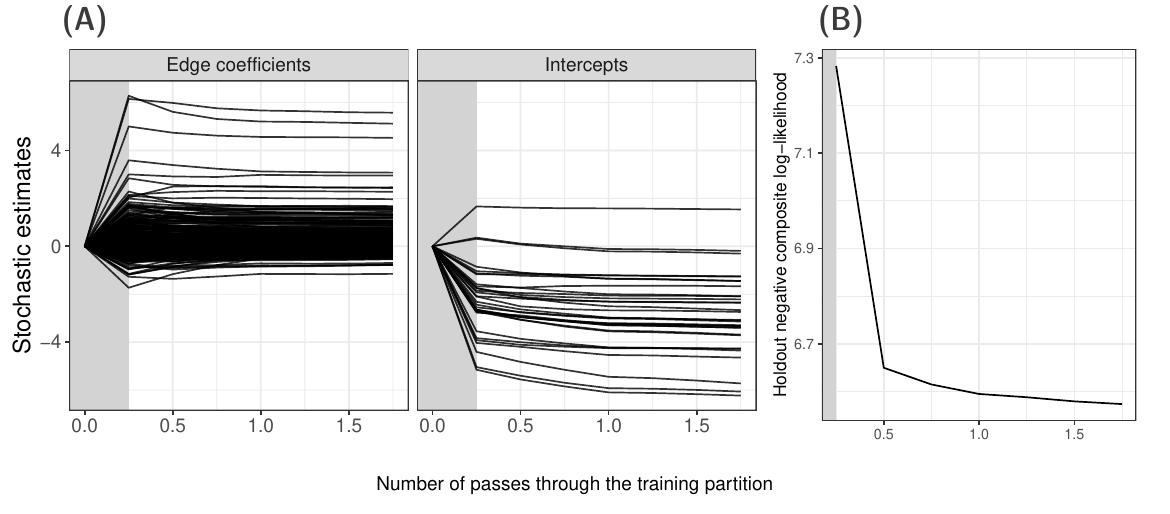}
	\caption{\label{fig:mh:checks} (A) Trajectories of the estimated parameters. (B) Average negative log-likelihood on the holdout partition. Grey areas denote the burn-in period.}
\end{figure}

Since the original NESARC data is not publicly available anymore,
the analysis carried out in Section 5 can be reproduced on synthetic data.
In particular, we provide a dataset simulated from the Ising model with parameters given by the numerical estimates of $\hat{\bm{\theta}}$ computed on the original NESARC data. Figure~\ref{fig:mh:mockest} shows how the stochastic estimates for the edges of the Ising model computed on the simulated dataset compare with the parameters generating the data and with the stochastic estimates computed on the original NESARC answers. While not identical, the results are still very similar and can be used as a reference for reproducing the analysis.

\begin{figure}[!h]\centering
	\includegraphics[width=.8\textwidth]{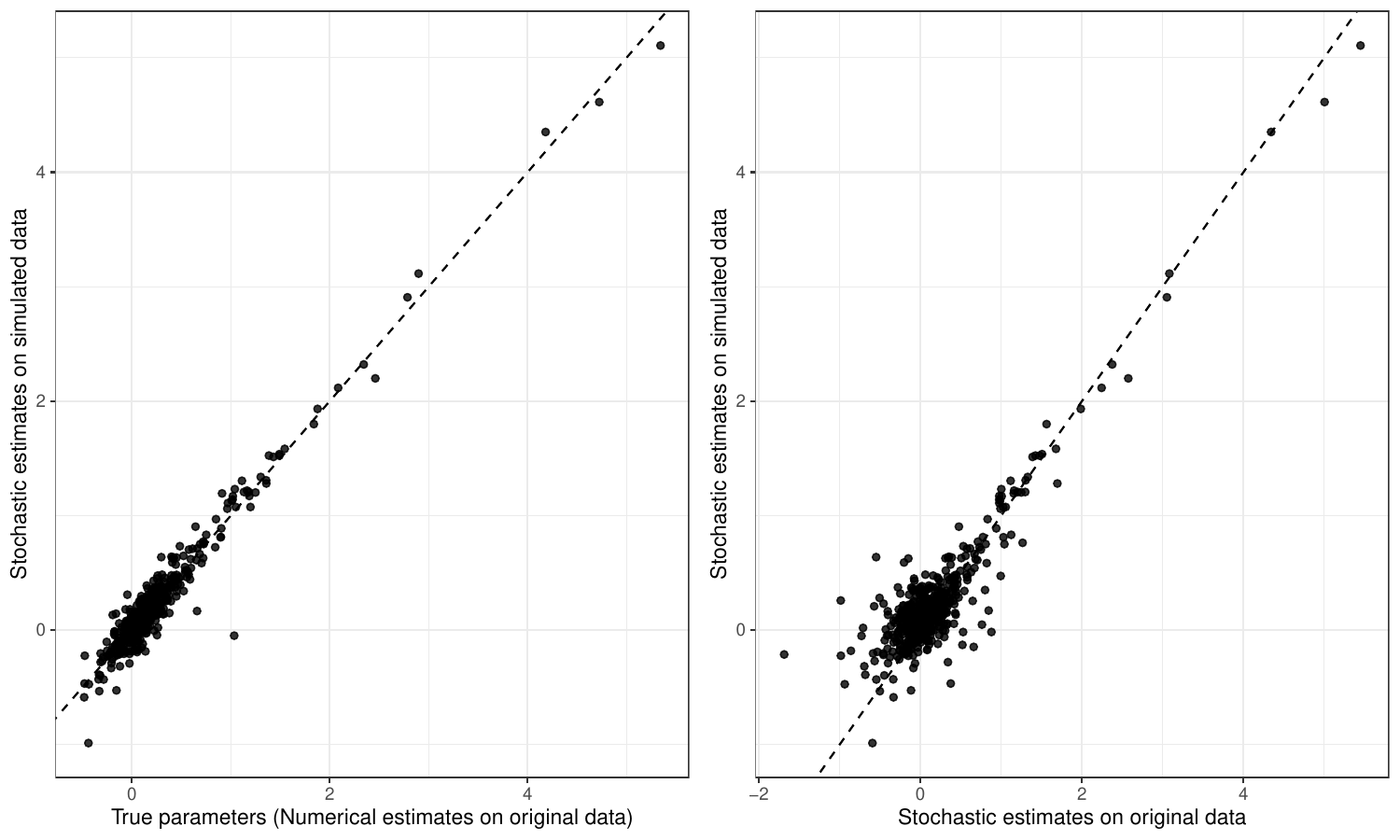}
	\caption{\label{fig:mh:mockest} Comparing estimated edges on the simulated dataset against numerical and stochastic estimates on original data.}
\end{figure}
\begin{footnotesize}
\begin{longtable}{llp{3in}}
\caption{\label{tab:dict}Description of the items of the NESARC survey considered in the analysis.}
\\Item & Area & Description \\
\midrule
S4AQ1 & Low mood & ever had 2-week period when felt sad, 
blue, depressed, or down most of time \\ 
  S4AQ2 & Low mood & ever had 2-week period when didn't care 
about things usually cared about \\ 
  S4BQ1 & Low mood & blood/natural father ever depressed \\ 
  S4BQ2 & Low mood & blood/natural mother ever depressed \\ 
  S4CQ1 & Low mood & had 2+ years period when mood was low, 
sad or depressed most of day, 
more than half the time \\ 
  S5Q1 & High mood & had 1+ week period of excitement/elation
that seemed not normal self \\ 
  S5Q2 & High mood & had 1+ week period of excitement/elation 
that made others concerned about you \\ 
  S5Q3 & High mood & had 1+ week period irritable/easily annoyed
that caused you to shout/break
 things/start fights or arguments \\ 
  S6Q1 & Panic disorders & had panic attack, suddenly felt 
frightened/overwhelmed/nervous as
if in great danger but were not \\ 
  S6Q2 & Panic disorders & was surprised by panic attack
 that happened out of the blue,
for no real reason or in a situation where
didn't expect to be frightened/nervous \\ 
  S6Q3 & Panic disorders & thought was having heart attack,
but doctor said just nerves or panic attack \\ 
  S7Q1 & Social phobia & ever had strong fear or
avoidance of social situation \\ 
  S7Q2 & Social phobia & had fear/avoidance of social situation due to fear of
embarrassment at what you might say/do around others \\ 
  S7Q3 & Social phobia & had fear/avoidance of social situation due to fear of
becoming speechless, having nothing to say or
saying something foolish \\ 
  S8Q1A1 & Specific phobia & ever had fear/avoidance of
insects, snakes, birds, other animals \\ 
  S8Q1A2 & Specific phobia & ever had fear/avoidance of
heights \\ 
  S8Q1A5 & Specific phobia & ever had fear/avoidance of
flying \\ 
  S8Q1A7 & Specific phobia & ever had fear/avoidance of
being in closed spaces \\ 
  S8Q1A8 & Specific phobia & ever had fear/avoidance of
seeing blood/getting an injection \\ 
  S9Q1B & General anxiety & ever had 6+ month period felt very tense/nervous/worrie
most of time abouteveryday problems \\ 
  S10Q1A1 & Personality disorders & avoid jobs or tasks that deal with a lot of people \\ 
  S10Q1A2 & Personality disorders & avoid getting involved with people
unless certain they will like you \\ 
  S10Q1A3 & Personality disorders & find it hard to be 'open' even with
people you are close to \\ 
  S10Q1A4 & Personality disorders & often worry about being criticized
or rejected in social situations \\ 
  S10Q1A5 & Personality disorders & believe that you are not as good, as smart,
or as attractive as most people \\ 
  S10Q1A7 & Personality disorders & afraid of trying new things or doing
things outside usual routine because
afraid of being embarrassed \\ 
  S11AQ1A1 & Antisocial disorders & often cut class, not go to class or go
to school and leave without permission \\ 
  S11AQ1A2 & Antisocial disorders & ever stay out late at night even though
parents told you to stay home \\ 
  S11AQ1A3 & Antisocial disorders & ever have time when bullied or pushed people
around or tried to make them afraid of you \\ 
  S11AQ1A4 & Antisocial disorders & ever run away from home at least twice or run
away and stay away for a longer time \\ 
  S11AQ1A6 & Antisocial disorders & more than once quit a job without knowing
where would find another one \\ 
  S11AQ1A9 & Antisocial disorders & ever have time lasting 1+ months when
had no regular place to live \\ 
   \bottomrule
\end{longtable}
\end{footnotesize}

\section{Experiments with randomized pairwise likelihood}
Another possible approach to our stochastic estimation strategy is to compute the numerical CLE on a subset of all likelihood components. This subset could be determined using a randomization strategy similar to the methods used in \citet{dillon2010stochastic} and \citet{mazo2023randomized}. However, this approach requires discarding part of the data to simplify the computational complexity of the problem. Unlike this method, our estimator spreads the data along the optimization rather than discarding it. Essentially, by spreading the sampling step across iterations rather than doing it at a single time point before optimization, we can average out the uncertainty caused by sampling components. As a result, our estimator is more efficient than alternatives that involve maximizing a random subset of the initial pool of components.

Below is a simulation experiment on the gamma frailty example, comparing it with the randomized pairwise likelihood estimator (RPLE) proposed by \citet{mazo2023randomized}. We conducted $500$ replications to align with the settings described in Section 4.2. 
We assess the mean square error performance of the RPLE 
with an average sampling rate fixed at  $\pi\in\{0.5, 0.25, 0.1\}$ (i.e. $\pi nK$ contribution out of $nK$) in the context of $p=30$. The results are presented in Figure~\ref{fig:rple}, overlaid on 
Figure in Section 4.2.
The horizontal dashed lines represent the average MSE performance across simulations, and the labels denote the corresponding value of $\pi$. It is important to note that our sampling schemes impose $\gamma_1=1/n$, much lower than the sampling rate provided by $\pi$.
\begin{figure}[t]\centering
	\includegraphics[width=\textwidth]{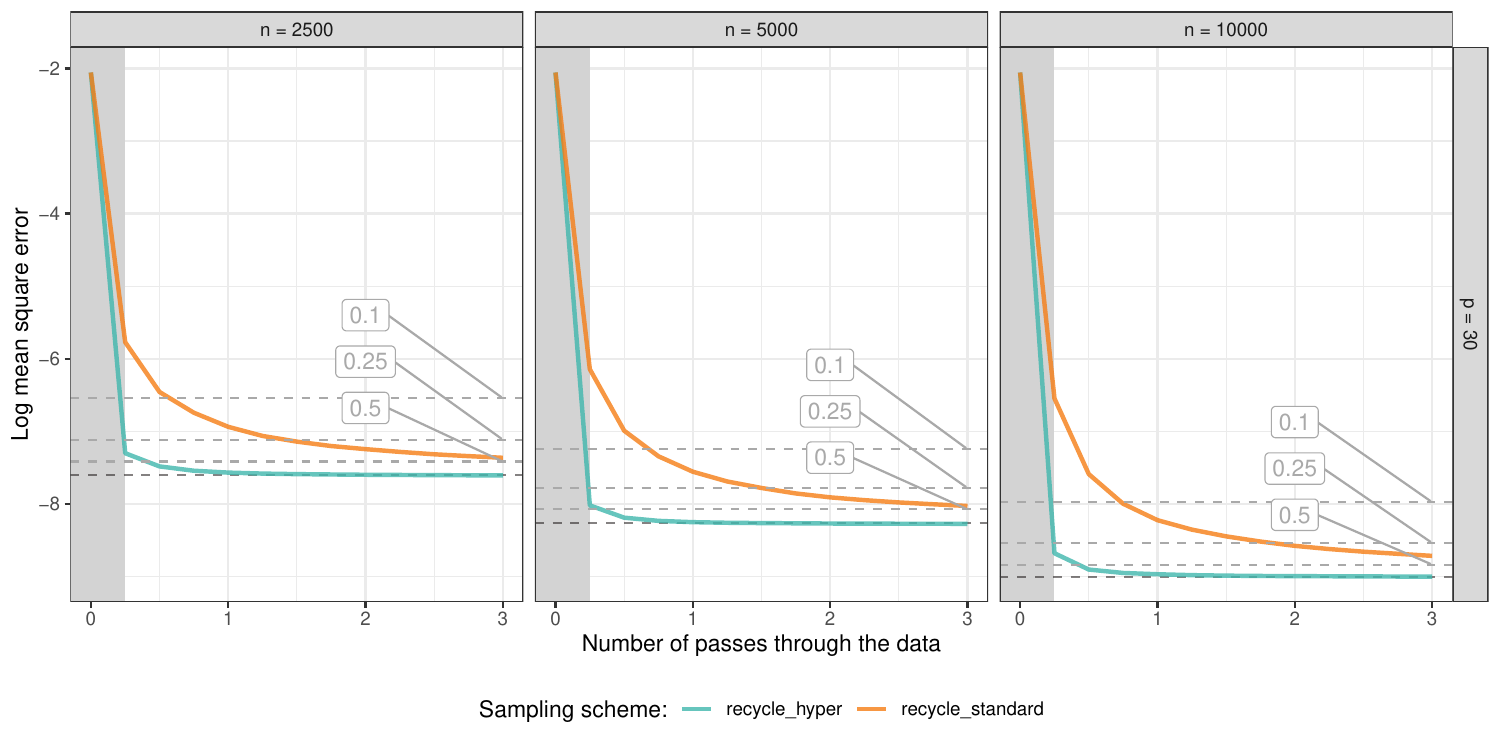}
	\caption{\label{fig:rple}Performance benchmark against the randomized pairwise likelihood estimator with $\pi\in\{0.5, 0.25, 0.1\}$. The lower dashed line denotes the performance of $\hat{\bm{\theta}}$.}
\end{figure}
The plots clearly show that our approach is more flexible in dealing with computational constraints and guarantees a better balance between computational and statistical efficiency.

\end{document}